\documentclass[preprint,11pt]{elsarticle}

\journal{\ }
\usepackage{graphicx}
\usepackage{pifont,latexsym,ifthen,amsthm,rotating,calc,textcase,booktabs}
\usepackage{amsfonts,amssymb,amsbsy,amsmath}
\usepackage{multirow}
\allowdisplaybreaks
\newtheorem{theorem}{Theorem}[section]
\newtheorem{lemma}[theorem]{Lemma}

\usepackage{natbib}
\newtheorem{prob}{RH problem}[section]

\numberwithin{equation}{section}
\usepackage{verbatim}
\usepackage{todonotes}
\usepackage{float}
\usepackage{subfigure}
\usepackage{hyperref}
\usepackage{appendix}
\usepackage[percent]{overpic}
\hypersetup{
    colorlinks=true, 
    linktoc=all,     
    linkcolor=blue,  
    citecolor=blue
}

\DeclareMathOperator*{\im}{Im}
\DeclareMathOperator*{\re}{Re}
\usepackage{color}
\usepackage{geometry}
\begin{document}

\begin{frontmatter}

\title{   The  long-time asymptotic of the derivative nonlinear Schr$\ddot{o}$dinger equation with step-like initial value }
\author[inst1]{Lili Wen}
\address[inst1]{School of Mathematical Sciences, Shanghai Key Laboratory of Pure Mathematics and Mathematical Practice, East China Normal University, Shanghai, 200241, China.}
\author[inst1,inst2]{Yong Chen}
\ead{ychen@sei.ecnu.edu.cn.}
\address[inst2]{ College of Mathematics and Systems Science, Shandong University of Science and Technology, Qingdao, 266590, China.}
\author[inst3]{Jian Xu}
\address[inst3]{College of Sciences, University of Shanghai for Science and Technology, Shanghai 200093, China.}

\begin{abstract}
\baselineskip=18pt
Consideration in this present paper is the long-time asymptotic of solutions to the derivative nonlinear Schr$\ddot{o}$dinger equation with the step-like initial value
 \begin{eqnarray}
q(x,0)=q_{0}(x)=\begin{cases}
\begin{split}
A_{1}e^{i\phi}e^{2iBx}, \quad\quad x<0,\\
A_{2}e^{-2iBx}, \quad\quad~~ x>0.
\end{split}\nonumber
\end{cases}
\end{eqnarray}
by Deift-Zhou method. The step-like initial problem described by a matrix Riemann-Hilbert problem. A crucial ingredient used in this paper is to introduce $g$-function mechanism for solving the problem of the entries of the jump matrix growing exponentially as $t\rightarrow\infty$. It is shown that the leading order term of the asymptotic solution of the DNLS equation expressed by the Theta function $\Theta$ about the Riemann-surface of genus 3 and the subleading order term expressed by parabolic cylinder and Airy functions.

\end{abstract}


\end{frontmatter}

\tableofcontents

\section{Introduction and Main Result}
This paper is devoted to the long-time asymptotic behavior of solutions to the derivative nonlinear Schr$\ddot{o}$dinger (DNLS) equation with step-like initial value
\begin{align}
&iq_{t}+q_{xx}-iq^{2}\overline{q}_{x}+\frac{1}{2}|q|^{4}q=0,\label{dnls3}\\
&
q(x,0)=\begin{cases}
A_{1}e^{i\phi_{1}}e^{-2iB_{1}x},\quad  x<0,\\
A_{2}e^{i\phi_{2}}e^{-2iB_{2}x}, \quad x>0.
\end{cases}\label{init}
\end{align}
Equation (\ref{dnls3}) alternatively termed by DNLS-III equation and sometimes referred as the Gerdjikov-Ivanov equation, to model weakly nonlinear dispersive water waves, Alfv$\acute{e}$n waves propagating along with the constant magnetic field in cold plasmas and ultrafast waves in optical fibers \cite{mw1976,me1976,ky1985}. Here and after, the overbar denotes the complex conjugation and the subscript denote differential  with respect to the corresponding variables. The DNLS-I,-II equations
\begin{align}
&iq_{t}+q_{xx}+i(|q|^{2}q)_{x}=0,\label{dnls1}\\
&iq_{t}+q_{xx}+i|q|^{2}q_{x}=0,\label{dnls2}
\end{align}
or termed by Kaup-Newell equation and Chen-Lee-Liu equation are also the canonical models of the DNLS equation. There exist a chain of gauge transformations to relate DNLS-I,-II,-III equation with each other \cite{wm1983}. Theoretically, a solution $\check{q}(x,t)$ of DNLS-I equation (\ref{dnls1}),  the invertible gauge transformation $q(x,t)=\check{q}(x,t)\exp\left(\mp i\int_{-\infty}^{x}|\check{q}|^{2}(y,t)\mathrm{d}{y}\right)$ maps the solutions of the DNLS equation (\ref{dnls3}) with
$t\rightarrow\frac{t}{2}$. However, it is very hard to be done explicitly due to the involved indefinite integration. Therefore, one needs to work with these equations separately.

Note that the problematic term $(|q|^{2}q)_{x}$ in (\ref{dnls1}) is replaced by the quintic term $|q|^{2}q$ without derivative and a derivative term $q^{2}\overline{q}_{x}$ with a better convolution structure in (\ref{dnls3}).  Therefore, the present paper is concerned with the long-time asymptotic under the step-like (asymmetric) initial value problem for $x\lessgtr0$ of the DNLS equation (\ref{dnls3}). The now well-known method of nonlinear steepest descent for studying the long-time asymptotic of solutions of integrable nonlinear equations with initial value was introduced in the early 1990s in a seminal paper by Deift and Zhou \cite{dz1993}, building on earlier works of Manakov \cite{msv1973} and Its \cite{its1981}. For a detailed historical review of this method please see \cite{dpa1993} and further extended by Deift, Venakides and Zhou \cite{dp1994,dp1997}. This method increasing perfect is based on the development of the nonlinear steepest descent
method for Riemann-Hilbert (RH) problem associated with integrable nonlinear equations. The intermesh of the RH formalism and the Deift-Zhou approach to the step-like initial value problems gradually been the subject of more works \cite{gav1973,ky1976,kvp1986,vs1986}. This idea was adapted by  Venakides to
problems in the shock problem with initial data for the integrable
equation \cite{dp1994}. Buckingham, Boutet de Monvel, Biondini, Minakov and Grava considered the long-time asymptotic for the step-like initial value \cite{bdm2011,rb2007,bm2009,bdm2022,bdm2021,gb2017,gb2021,gb2014,am2011,tg2020}.

Before stating our assumptions and result more precisely, we recall known results concerning the DNLS equation (\ref{dnls3}). The multiple soliton solutions are addressed for the DNLS equation (\ref{dnls3}) under the initial value with zero/nonzero boundary conditions as $x\rightarrow\pm\infty$ by analyzing a matrix Riemann-Hilbert (RH) problem \cite{zzc2021}. The Dirichlet initial-boundary value problem for the DNLS equation (\ref{dnls3}) is exhibited to
be locally well-posed in $H^{s}(\mathbb{R}^{+})$ for $s\in(1/2, 5/2)$ and $s= 3/2$ on the half-line \cite{emb2018}. For the nondecaying boundary value, the existence of a solution is classified
for the DNLS equation (\ref{dnls3}) with asymptotical time-periodic boundary values and two particularly families of parameters lying on the quarter plane $\{(x, t) \in \mathbb{R}^2|x \geq 0, t \geq 0\}$ \cite{fs2020}. Due to its integrability,
many explicit solutions in the closed form (including solitons and algebraic solitons,
breathers and rogue wave solutions, algebro-geometric solutions), Hamiltonian structures, integrable decompositions and similarity reductions have been presented for the DNLS equation (\ref{dnls3}) \cite{f2001,yh2015,zss2022,xsw2012,hy2013,zp2020,ks2004,ks2005,cjb2022}. The global existence for the DNLS equation (\ref{dnls3}) was proved by inverse scattering method in \cite{ljq2016}.

In the context of inverse scattering, the long-time asymptotics were studied for the DNLS equation (\ref{dnls3}) with step-like initial values  $q(x,0)=\begin{cases}
Ae^{i\phi}e^{-2iBx},  x\leq0,\\
0, \quad\quad\quad\quad~ x>0,
\end{cases}$
 time-periodic initial value on the quarter plane and the nonzero boundary condition by the nonlinear steepest descent method \cite{xu2013,ln2019,tsf2018}. Liu studied the long-time behavior of solutions to the DNLS equation (\ref{dnls3}) for soliton-free initial data \cite{ljq2018}. These works gave the leading order asymptotics where error is $\mathcal{O}(t^{-1/2})$. However, the subleading order asymptotics not derived.

In the present paper, we consider the long-time asymptotic in a shock case of DNLS equation (\ref{dnls3}) with the more general step-like initial value conditions (\ref{init}). Moreover, we derive the leading order and the subleading order asymptotics where error is $\mathcal{O}(t^{-1/2}\ln(t))$.

In the initial value (\ref{init}), $\{A_{j}, B_{j}, \phi_{j}\}_{1}^{2}\in\mathbb{R}$ and $A_{j}>0$. The equation (\ref{dnls3}) with initial condition (\ref{init}) admits the plane wave solution
$q_{j}^{\pm\infty}(x,t)=A_{j}e^{i\phi_{j}}e^{-2iB_{j}x+2i\omega_{j}t}$, where $\omega_{j}=\frac{1}{4}(A_{j}^{4}+4A_{j}^{2}B_{j}-8B_{j}^{2})$. For $B_{1}>B_{2}$ (the rarefaction case), the asymptotical does not depend on the values of $D_{j}$ ($D_{j}$ defined in subsection \ref{sec2.2}), the details see \cite{xu2013}. For $B_{1}<B_{2}$ (the shock case), the asymptotic is influenced by $D_{j}/(B_{2}-B_{1})$. For simplicity, we consider the symmetric shock $D_{1}=D_{2}=D>0$ and $B_{2}=-B_{1}=B>0$. The infinite branch of $\im{g}=0$ pass through the points $E_{1}$ and $\overline{E}_{1}$ before the two real zeros $\mu_{1}$ and $\mu_{2}$ of $\im{g}=0$ directly lead to the genus 3. The distribution of the genus see Figure \ref{spacetime}, where $\xi_{E_{1}}$ denotes $\xi=\xi_{E_{1}}=2(B+\sqrt{D^{2}+B^{2}})$ as the infinite branch pass through the points $E_{1}$ and $\overline{E}_{1}$, $\xi_{\mu}$ denotes $\xi=\xi_{\mu}$ the two real zeros $\mu=\mu_{1}=\mu_{2}$ of $\im g=0$. And the two zeros merge where $\xi_{mer}=4(-B+\sqrt{D})$. So the infinite branch pass through the two points $E_{1}$ and $\overline{E}_{1}$ before the zeros merge if and only if $\xi_{E_{1}}>\xi_{mer}$, i.e. $D/B<\frac{4+6\sqrt{2}}{7}$.

The present paper is devoted to study the long-time asymptotic of the genus 3 for the DNLS equation (\ref{dnls3}) by means of the matrix RH problem (RH problems in this paper are $2\times2$ matrix-valued). A critical step in the nonlinear steepest descent method consists in deforming the contour associated to the RH problem in a way adapted to the structure of the phase function that defines the oscillatory dependence on parameters. When the entries of the jump matrix are not analytic, they must be approximated by rational functions so that the deformation can be carried out. Therefore, we bring in the $g$-function mechanism which is introduced when the entries of the jump matrix grow exponentially or oscillate as $t\rightarrow\infty$ \cite{dv1994}. The core idea of $g$-function mechanism is transform the phase function $\theta$ of the basic RH problem to a $g$-function so that the jump matrix of RH problem is constant or decay to a identity matrix by some matrix deformations.

      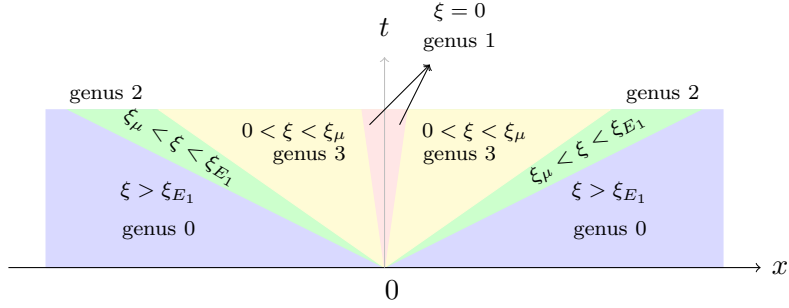
\begin{figure}[H]
	\begin{center}
		\begin{tikzpicture}
            \draw[blue!15, fill=blue!15] (-4.5,0)--(4.5,0)--(4.5,2.1)--(-4.5, 2.1)--(-4.5,0);
                   \draw[green!20, fill=green!20] (0,0 )--(-4.2,2.1)--(4.2,2.1)--(0,0);
            \draw[yellow!20, fill=yellow!20] (0,0 )--(-3.0,2.1)--(3.0,2.1)--(0,0);
          \draw[red!10, fill=red!10] (0,0 )--(-0.3,2.1)--(0.3,2.1)--(0,0);
		\draw [-> ](-5,0)--(5,0);
		\draw [gray!50,-> ](0,0)--(0,2.8);
		\node    at (0.1,-0.3)  {$0$};
		\node    at (5.26,0)  { $x$};
		\node    at (0,3.2)  { $t$};
	   \node  []  at (1,1.5) {\scriptsize genus 3};
     \node  []  at (-1,1.5) {\scriptsize genus 3};
       \node  []  at (1.2,1.8) {\scriptsize $0<\xi<\xi_{\mu}$};
     \node  []  at (-1.2,1.8) {\scriptsize $0<\xi<\xi_{\mu}$};
     \node [] at (1,3) {\scriptsize genus 1 };
     \node [] at (1,3.4) {\scriptsize  $\xi=0$};
         \draw [-> ](0.2,1.9)--(0.6,2.7);
         \draw [-> ](-0.2,1.9)--(0.6,2.7);
         \node [] at (-3,1) {\scriptsize $\xi>\xi_{E_{1}}$};
        \node [] at (-3,0.5) {\scriptsize genus 0};
            \node [] at (3,1) {\scriptsize $\xi>\xi_{E_{1}}$};
    \node [] at (3,0.5) {\scriptsize genus 0};
      \node  []  at (3.7,2.3) {\scriptsize genus 2};
        \node  []  at (-3.7,2.3) {\scriptsize genus 2};
         \node  []  at (-2.7,1.6) {\scriptsize \rotatebox{-30}{$\xi_{\mu}<\xi<\xi_{E_{1}}$}};
         \node  []  at (2.7,1.6) {\scriptsize \rotatebox{30}{$\xi_{\mu}<\xi<\xi_{E_{1}}$}};
		\end{tikzpicture}
	\end{center}
	\caption{\footnotesize The space-time region of $x$ and $t$ for $\frac{4+6\sqrt{2}}{7}>\frac{D}{B}>1$, where the blue region $\color{blue!20}\blacksquare$ denote the genus 0 region; the green region $\color{green!20}\blacksquare$ denote the genus 2 region; the yellow region $\color{yellow!30}\blacksquare$ denote the genus 3 region; the red region $\color{red!10}\blacksquare$ denote the genus 1 region.}
	\label{spacetime}
\end{figure}

For this purpose, we fix some notations for this paper. We define $\mathbb{C}^{+}$ and $\mathbb{C}^{-}$ are the upper and lower plane of the complex plane $\mathbb{C}$, see the second one in Figure \ref{con1}. All RH problems in this paper are considered in the $L^{2}$-RH problem \cite{pa1999,lp2002,as2006,zx1989,lj2017,lj2018}.

Our main result is addressed as follows:
\begin{theorem}\label{solution}
The long-time asymptotic of the solution to the DNLS equation (\ref{dnls3}) with the initial value condition (\ref{init})
is given by the following formula
\begin{equation}
q(x,t)=q_{0}+\frac{q_{1}}{\sqrt{t}}+\mathcal{O}(t^{-1}\ln{t}),\quad t\rightarrow\infty,
\end{equation}
where the leading order term shown as
\begin{equation}
q_{0}=e^{2i(tg^{(0)}+h(\infty))}\mathrm{Im}(E_{1}+E_{2}+\alpha+\beta)\frac{\Theta(\varphi(\infty^{+})+d)\Theta(\varphi(\infty^{+})-v(t)-d)}{\Theta(\varphi(\infty^{+})+v(t)+d)\Theta(\varphi(\infty^{+})-d)},
\end{equation}
and the subleading order term shown as
\begin{equation}
q_{1}=-2ie^{2i(tg^{(0)}+h(\infty))}\frac{(Y_{\mu}(x,t,\mu)m_{1}^{pc}Y_{\mu}^{-1}(x,t,\mu))_{12}}{\psi_{\mu}(\mu)}.
\end{equation}
The constants $g^{(0)}$ and $h(\infty)$ are given by (\ref{g0}) and (\ref{hinfty}).
The Riemann Theta function $\Theta$ and Abel map $\varphi$ are defined by (\ref{theta}) and (\ref{varphi}), respectively. The matrixes $Y_{\mu}$ and $m_{1}^{pc}$ are given by (\ref{ymu}) and (\ref{m1pc}). The constant $\psi_{\mu}$ defined by (\ref{psimu}).
\end{theorem}

\section{Preliminaries}\label{sec2}
In this section, we mainly introduce some preparations for studying the long-time asymptotic of the DNLS equation (\ref{dnls3}) with the initial value (\ref{init}), such as Jost solutions, scattering datas, basic RH problem, $g$-function.
\subsection{Jost solution}
It is well known that the DNLS equation (\ref{dnls3}) can be represented as the compatibility condition of two linear spectral problem (Lax pair). The Lax pair
makes it possible to reduce the long-time asymptotics of the solutions with the initial value problem for the equation (\ref{dnls3}) to the matrix RH problem, which involves the Jost solutions of the Lax pair. The DNLS equation (\ref{dnls3}) admits the Lax pair \cite{kav1997}
\begin{equation}
\Psi_{x}=U\Psi,\quad \quad  \Psi_{t}=V\Psi,\label{laxpair1}
\end{equation}
where $\Psi$ is a $2\times2$ matrix-valued function of $(x,t,k)$, $k\in\mathbb{C}$ is the spectral parameter and
\begin{subequations}
\begin{align}
&U=-ik^{2}\sigma_{3}+kQ+\frac{i}{2}|q|^{2}\sigma_{3},\nonumber\\
&V=-2ik^{4}\sigma_{3}+2k^{3}Q+ik^{2}|q|^{2}\sigma_{3}-ikQ_{x}\sigma_{3}+\frac{i}{4}|q|^{4}\sigma_{3}+\frac{1}{2}(q\overline{q}_{x}-\overline{q}q_{x})\sigma_{3},\nonumber
\end{align}
\end{subequations}
with
\begin{equation}
Q=\left(
\begin{array}{cc}
0&q\\
-\overline{q}&0
\end{array}
\right),\quad
\sigma_{3}=\left(
\begin{array}{cc}
1&0\\
0&-1
\end{array}
\right).\nonumber
\end{equation}
Setting $\Psi_{j}^{\pm\infty }$ (for convenient we omit $x, t, k$) is the asymptotic solutions of the Lax pair (\ref{laxpair1}) under the step-like initial value (\ref{init}). And $\Psi_{j}^{\pm\infty }$ satisfy the asymptotic Lax pair as $x\rightarrow\pm\infty$,
\begin{equation}
\Psi_{jx}^{\pm\infty}=U_{j}^{\pm\infty}\Psi_{j}^{\pm\infty},\quad \Psi_{jt}^{\pm\infty}=V_{j}^{\pm\infty}\Psi_{j}^{\pm\infty},\nonumber
\end{equation}
where $U_{j}^{\pm\infty}$ and $V_{j}^{\pm\infty}$ are defined by $U$ and $V$ with the $q$ instead by $q_{j}^{\pm\infty}$.
The asymptotic solutions $\Psi_{j}^{\pm\infty}$ with the form
\begin{equation}
\Psi_{j}^{\pm\infty}=e^{(-iB_{j}x+i\omega_{j}t)\sigma_{3}}Y_{j}(k)e^{(-ixX_{j}(k)-it\Omega_{j}(k))\sigma_{3}},\nonumber
\end{equation}
where
\begin{subequations}
\begin{align}
&X_{j}^{2}(k)=(k^{2}-B_{j}-\frac{A_{j}^{2}}{2})+k^{2}A_{j}^{2},\quad \Omega_{j}(k)=2(k^{2}+B_{j})X_{j}(k),\nonumber\\
&Y_{j}(k)=\frac{1}{2}e^{\frac{i\phi_{j}}{2}\sigma_{3}}\left(
\begin{array}{cc}
y_{j}(k)+y_{j}^{-1}(k)&y_{j}(k)-y_{j}^{-1}(k)\\
y_{j}(k)-y_{j}^{-1}(k)&y_{j}(k)+y_{j}^{-1}(k)
\end{array}
\right)e^{-\frac{i\phi_{j}}{2}\sigma_{3}},\nonumber\\
&y_{j}(k)=\left(\frac{k^{2}-B_{j}-\frac{A_{j}^{2}}{2}-ikA_{j}}{k^{2}-B_{j}-\frac{A_{j}^{2}}{2}+ikA_{j}}\right)^{\frac{1}{4}}.\nonumber
\end{align}
\end{subequations}
The branch cuts for $X_{j}$ and $y_{j}$ are taken along the segment
$\gamma_{j}\cup\overline{\gamma}_{j}=\{k\in\mathbb{C}|\re^{2}{k}-\im^{2}{k}=B_{j},\re^{2}{k}\leq B_{j}+\frac{A_{j}^{2}}{2}\}$,
where $\gamma_{j}=\{k\in\mathbb{C}|\re^{2}{k}-\im^{2}{k}=B_{j},\re^{2}{k}\leq B_{j}+\frac{A_{j}^{2}}{2},\im k^{2}>0\}$. $X_{j}$ and $y_{j}$ with the asymptotics
\begin{equation}
X_{j}(k)=k^{2}-B+\mathcal{O}(k^{-2}),\quad y_{j}(k)=1+\mathcal{O}(k^{-1}),\quad k\rightarrow\infty.\nonumber
\end{equation}
For $k^{2}\in\gamma_{j}\cup\overline{\gamma}_{j}$, $y_{j+}(k)=iy_{j-}(k)$ and $Y_{j}(k)$ satisfy the jump condition
\begin{equation}
Y_{j+}(k)=Y_{j-}(k)\left(
\begin{array}{cc}
0&ie^{i\phi_{j}}\\
ie^{-i\phi_{j}}&0
\end{array}
\right).\nonumber
\end{equation}
Introduce the transformation
$\Psi_{j}(x,t,k)=\mu_{j}(x,t,k)e^{-ix(X_{j}+B_{j})\sigma_{3}-it(\Omega_{j}-\omega_{j})\sigma_{3}}$.
The Lax pair (\ref{laxpair1}) rewritten as a new version about Jost solution $\Psi_{j}$
\begin{equation}
\Psi_{jx}=(U-U_{j})^{\pm\infty}\Psi_{j}+U_{j}^{\pm\infty}\Psi_{j},\quad\Psi_{jt}=(V-V_{j}^{\pm\infty})\Psi_{j}+V_{j}^{\pm\infty}\Psi_{j},\label{lax2}
\end{equation}
and one can note that the version of Lax pair (\ref{laxpair1}) about $\mu_{j}$
\begin{equation}
\mu_{jx}=i(X_{j}+B_{j})\mu_{2}\sigma_{3}+U\mu_{j},\quad \mu_{jt}=i(\Omega_{j}-\omega_{j})\mu_{2}\sigma_{3}+V\mu_{j}.
\end{equation}
Multiply both side by $(\Psi_{j}^{\pm\infty})^{-1}$ for the equation of (\ref{lax2}) and derive the full derivative form
\begin{equation}
\mathrm{d}\left[(\Psi_{j}^{\pm\infty})^{-1}\Psi_{j}\right]=(\Psi_{j}^{\pm\infty})^{-1}(U-U_{j}^{\pm\infty})\Psi_{j}\mathrm{d}x+(\Psi_{j}^{\pm\infty})^{-1}(V-V_{j}^{\pm\infty})\Psi_{j}\mathrm{d}t.\nonumber
\end{equation}
Furthermore, the solutions $\Psi_{j}(x,t,k)$ and $\mu_{j}(x,t,k)$ can be represented as the Volterra integral equations
\begin{equation}
\begin{split}
&\Psi_{j}(x,t,k)=\Psi_{j}^{\pm\infty}(x,t,k)+\int_{\pm\infty}^{x}\Lambda(x,y,t,k)\Lambda^{\natural}(y,t,k)\Psi_{j}(y,t,k)\mathrm{d}y.\\
&\begin{split}\mu_{j}(x,t,k)=
&e^{i(\omega_{j} t-B_{j}x)\widehat{\sigma}_{3}}Y_{j}(k)\\
&+\int_{\pm\infty}^{x}\Lambda(x,y,t,k)\Lambda^{\natural}(y,t,k)\mu_{j}(y,t,k)e^{-i(X_{j}+B_{j})(y-x)}\mathrm{d}y.
\end{split}
\end{split}\nonumber
\end{equation}
where
\begin{equation}
\begin{split}
&\Lambda(x,y,t,k)=\Psi_{j}^{\pm\infty}(x,t,k)(\Psi_{j}^{\pm\infty})^{-1}(y,t,k),\\ &\Lambda^{\natural}(y,t,k)=k(Q-Q_{j}^{\pm\infty})(y,t)+\frac{i}{2}|q|^{2}(y,t)\sigma_{3}-\frac{i}{2}A_{j}^{2}\sigma_{3}.
\end{split}
\end{equation}
The existence, analyticity and differentiation of $\Psi_{j}$ and $\mu_{j}$ can be proven directly, here we just list their properties, for details, see \cite{xu2013}.
\begin{lemma}
The Jost solutions $\Psi_{j}(x,t,k)$ of (\ref{laxpair1}) admit the following properties:\\
$\blacktriangleright$ $\Psi_{1}^{(1)}(x,t,k)$ analytic in $\{k\in\mathbb{C}\mid\mathrm{Im}k^{2}>0\}$ and $\Psi_{2}^{(2)}$ analytic in $\{k\in\mathbb{C}\mid\mathrm{Im}k^{2}>0\}$, $\Psi_{1}^{(2)}(x,t,k)$ analytic in $\{k\in\mathbb{C}\mid\mathrm{Im}k^{2}<0\}$ and $\Psi_{2}^{(1)}$ analytic in $\{k\in\mathbb{C}\mid\mathrm{Im}k^{2}<0\}$, where $\gamma_{j}\cup\overline{\gamma}_{j}=\{k\in\mathbb{C}\mid\mathrm{Re}^{2}k-\mathrm{Im}^{2}k=B,\mathrm{Re}^{2}k\leq B_{j}+A_{j}^{2}/4\}$. The distribution of analysis region see the first one in Figure \ref{con1};\\
$\blacktriangleright$ $\Psi_{j}(k)$ satisfy the symmetries
\begin{equation}
\Psi_{j}(k)=\sigma_{1}\sigma_{3}\overline{\Psi_{j}(\overline{k})}\sigma_{3}\sigma_{1},\quad \sigma_{1}=\left(
\begin{array}{cc}
0&1\\
1&0
\end{array}
\right);\nonumber
\end{equation}
$\blacktriangleright$ $(\Psi_{0j}^{-1}\Psi_{j})^{(1)}(x,t,k)$ and $(\Psi_{0j}^{-1}\Psi_{j})^{(2)}$ admit the asymptotic
\begin{equation}
\begin{split}
&\left((\Psi_{j}^{\pm\infty})^{-1}\Psi_{j}\right)^{(1)}(x,t,k)=e^{(1)}+\mathcal{O}(k^{-1}),\quad k\rightarrow\infty;\\
&\left((\Psi_{j}^{\pm\infty})^{-1}\Psi_{j}\right)^{(2)}(x,t,k)=e^{(2)}+\mathcal{O}(k^{-1}),\quad k\rightarrow\infty;\nonumber
\end{split}
\end{equation}
where the superscripts $\left((\Psi_{j}^{\pm\infty})^{-1}\Psi_{j}\right)^{(j)}$ denote the $j^{th}$ column of the matrix $(\Psi_{j}^{\pm\infty})^{-1}\Psi_{j}$, $e^{(j)}$ denotes the $j^{th}$ column of a identity matrix $I_{2\times2}$.
\end{lemma}

\subsection{scattering~data}\label{sec2.2}
For the reason that $\Psi_{j}(x,t,k)$ are solutions of the Lax pair (\ref{laxpair1}), there exist a scattering matrix $S(\lambda)$ obey the scattering relation
\begin{equation}
\Psi_{2}(x,t,k)=\Psi_{1}(x,t,k)S(k), \quad k\in\mathbb{R}, \quad k\neq B_{j}.\label{psis}
\end{equation}
From the symmetries of $\Psi_{j}(x,t,k)$, the scattering matrix $S(k)$ with the following structure
\begin{equation}
S(k)=\left(
\begin{array}{cc}
\overline{a(\overline{k})}&b(k)\\
-\overline{b(\overline{k})}&a(k)
\end{array}
\right),\nonumber
\end{equation}
and $\det{S}(k)=1$. From the analyticities of $\Psi_{j}(x,t,k)$, one can derive the analyticities of $a(k)$ and $\overline{a}(k)$ are in $\mathbb{C}^{+}\backslash\{\gamma_{j}\cup\overline{\gamma}_{j}\}_{j=1}^{2}$ and $\mathbb{C}^{-}\backslash\{\gamma_{j}\cup\overline{\gamma}_{j}\}_{j=1}^{2}$, respectively. The scatter coefficient $r(k)=\frac{\overline{b(\overline{k})}}{a(k)}$. And $a(k)$ with the asymptotic
$a(k)=1+\mathcal{O}(k^{-1})$ for $k\rightarrow\infty$.  Generally, the map  $q\rightarrow\{a,b,r\}$ is the direct scattering map.

From the analyticities and asymptotics of $\Psi_{j}(x,t,k)$, a piecewise matrix function $m(x,t,k)$ is given by
\begin{eqnarray}
m(x,t,k)=\begin{cases}
\begin{split}
\left(\frac{\Psi_{1}^{(1)}e^{it\theta}}{a},\Psi_{2}^{(2)}e^{-it\theta}\right), \quad \Omega^{+}\setminus\gamma_{j},\\
\left(\Psi_{2}^{(1)}e^{it\theta},\frac{\Psi_{1}^{(2)}e^{-it\theta}}{\overline{a}}\right), \quad \Omega^{-}\setminus\overline{\gamma}_{j},
\end{split}\label{lambda}
\end{cases}\label{M}
\end{eqnarray}
where $\theta=2k^{4}+\xi k^{2}$, $\xi=x/t$ and  the analyticity regions are defined by $\Omega^{\pm}=\{k\in\mathbb{C}\mid\pm\mathrm{Im}k^{2}>0\}$.
This matrix function $m(x,t,k)$ admits the jump condition
\begin{equation}
m_{+}(x,t,k)=m_{-}(x,t,k)J(x,t,k),\quad k^{2}\in\complement.
\end{equation}
The jump contour $\complement$ can be viewed as the
boundary of the regions $\Omega^{\pm}$.
The jump matrix $J(x,t,k)$ is given by
\begin{equation}
J(x,t,k)=\left(
\begin{array}{cc}
1+r\overline{r}&\overline{r}\\
r&1
\end{array}
\right).
\end{equation}
Due to the multi-value of spectrum parameter $k^{2}$,
we reduce by symmetry from
scattering data on the oriented contour $\complement$ to scattering data on the oriented contour $\mathbb{R}$.
Both contours with orientation are shown in Figure \ref{con1}.
Introduce the transformation \cite{kdj1978,lka2017}
\begin{equation}
\tilde{m}(x,t,\lambda)=\bigtriangledown k^{-\frac{\widehat{\sigma}_{3}}{2}}m(x,t,k),\label{ktolambda}
\end{equation}
where
$\bigtriangledown=\left(
\begin{array}{cc}
1&0\\
-\frac{i}{2}\overline{q}&1
\end{array}
\right)$.
Using transformation (\ref{ktolambda}), we can now reduce the spectrum problem with $k\in\complement$ to $\lambda\in\mathbb{R}$.
And $\tilde{m}$ satisfies the asymptotic $\tilde{m}\rightarrow I$ as $\lambda\rightarrow\infty$. The modified scattering coefficient $\varrho(\lambda)=\frac{r(k)}{k}$.
$X_{j}^{2}$ and $\Omega_{j}$ rewritten as
\begin{equation}
X_{j}^{2}(\lambda)=(\lambda-B_{j})^{2}+\frac{A_{j}^{4}}{4}+A_{j}^{2}B_{j},\quad
\Omega_{j}(\lambda)=2(\lambda+B_{j})X_{j}(\lambda).\nonumber
\end{equation}
The branch points $E_{j}=B_{j}+iD_{j}$, where $D_{j}^{2}=\frac{A_{j}^{4}}{4}+A_{j}^{2}B_{j}$. The branch cuts $\gamma_{j}=[B_{j},E_{j}]$ and $\overline{\gamma}_{j}=[\overline{E}_{j},B_{j}]$.
For $\lambda\rightarrow\infty$, $X_{j}(\lambda)$ and $\Omega_{j}(\lambda)$ with the asymptotics
\begin{equation}
X_{j}(\lambda)=\lambda-B+\mathcal{O}(\lambda^{-1}),\quad \Omega_{j}(\lambda)=2\lambda^{2}+\omega_{j}+\mathcal{O}(\lambda^{-1}).\quad \lambda\rightarrow\infty.\nonumber
\end{equation}
Functions $\tilde{m}_{\pm}(x,t,\lambda)$ are analyticity in the regions $\mathbb{C}^{\pm}\backslash\{\gamma_{j}\cup\overline{\gamma}_{j}\}$ ($\gamma_{j}$ and $\overline{\gamma}_{j}$ see Figure \ref{con2})
and satisfy the jump condition
\begin{equation}
\tilde{m}_{+}(x,t,\lambda)=\tilde{m}_{-}(x,t,\lambda)e^{-it(\xi\lambda+2\lambda^{2})\hat{\sigma}_{3}}\tilde{J}(x,t,\lambda),
\end{equation}
where jump matrix $\tilde{J}(x,t,\lambda)$ is given by
\begin{equation}
\tilde{J}(x,t,\lambda)=
\left(
\begin{array}{cc}
1+\lambda\varrho\overline{\varrho}&\overline{\varrho}\\
\lambda\varrho&1
\end{array}
\right).\quad~ \lambda\in\mathbb{R}.\\
\end{equation}

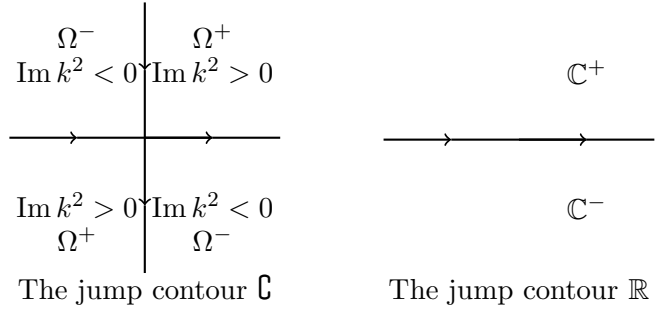
\begin{figure}[H]
\begin{center}
\begin{tikzpicture}[scale=0.45]
\draw[thick,->] (0,4) -- (0,2);
\draw[thick,->] (0,2) -- (0,-2);
\draw[thick,-] (0,-2) -- (0,-4);
\draw[thick,->] (-4,0)--(-2,0);
\draw[thick,-] (-2,0)--(2,0);
\draw[thick,->] (0,0)--(2,0);
\draw[thick,-] (2,0) -- (4,0);
\draw [] (2,2) circle [radius=0] node[] {$\im{k^{2}}>0$};
\draw [] (-2,-2) circle [radius=0] node[] {$\im{k^{2}}>0$};
\draw [] (-2,2) circle [radius=0] node[] {$\im{k^{2}}<0$};
\draw [] (2,-2) circle [radius=0] node[] {$\im{k^{2}}<0$};
\draw [] (2,3) circle [radius=0] node[] {$\Omega^{+}$};
\draw [] (-2,-3) circle [radius=0] node[] {$\Omega^{+}$};
\draw [] (-2,3) circle [radius=0] node[] {$\Omega^{-}$};
\draw [] (2,-3) circle [radius=0] node[] {$\Omega^{-}$};
\draw [] (0,-4.5) circle [radius=0] node[] {The jump contour $\complement$};
\end{tikzpicture}
\quad \quad\quad
\begin{tikzpicture}[scale=0.45]
\draw[white] (0,-4) -- (0,4);
\draw[thick,->] (-4,0)--(-2,0);
\draw[thick,-] (-2,0)--(2,0);
\draw[thick,->] (0,0)--(2,0);
\draw[thick,-] (2,0) -- (4,0);
\draw [] (2,2) circle [radius=0] node[] {$\mathbb{C}^{+}$};
\draw [] (2,-2) circle [radius=0] node[] {$\mathbb{C}^{-}$};
\draw [] (0,-4.5) circle [radius=0] node[] {The jump contour $\mathbb{R}$};
\end{tikzpicture}
\end{center}
\caption{The jump contour $\complement=\mathbb{R}\cup i\mathbb{R}$ and $\mathbb{R}$. In this figure we omit the branch cuts. }\label{con1}
\end{figure}

Now, we successfully map the spectral problem of $k$-plane to $\lambda$-plane. For the reason that there exist branch cuts $\gamma_{j}\cup\overline{\gamma}_{j}$, the function $\tilde{m}(x,t,\lambda)$  does not continuation as $\lambda\in\gamma_{j}\cup\overline{\gamma}_{j}$. We consider the jump condition about the branch cuts $\gamma_{j}\cup\overline{\gamma}_{j}$ as shown in Lemma \ref{lemma}:
\begin{lemma}\label{lemma}
For $\lambda\in\gamma_{j}$ or $\lambda\in\overline{\gamma}_{j}$, the jump matrices are given by
\begin{eqnarray}
\tilde{J}=\begin{cases}
\left(
\begin{array}{cc}
1&0\\
\lambda f(\lambda)&1
\end{array}
\right),~\lambda\in\gamma_{1}\\
\left(
\begin{array}{cc}
\frac{\tilde{a}_{-}}{\tilde{a}_{+}}&ie^{i\phi_{2}}\\
0&\frac{\tilde{a}_{+}}{\tilde{a}_{-}}
\end{array}
\right),\lambda\in\gamma_{2}
\end{cases},\quad
\tilde{J}=\begin{cases}
\left(
\begin{array}{cc}
1&-\overline{f(\overline{\lambda})}\\
0&1
\end{array}
\right),\quad \lambda\in\overline{\gamma}_{1}\\
\left(
\begin{array}{cc}
\frac{\tilde{\overline{a}}_{+}}{\tilde{\overline{a}}_{-}}&0\\
ie^{-i\phi_{2}}&\frac{\tilde{\overline{a}}_{-}}{\tilde{\overline{a}}_{+}}
\end{array}
\right),\lambda\in\overline{\gamma}_{2}
\end{cases}\nonumber
\end{eqnarray}
where $f(\lambda)=\varrho_{+}(\lambda)-\varrho_{-}(\lambda)$.
\end{lemma}

\begin{proof}
For $\lambda\in\gamma_{j}\cup\overline{\gamma}_{j}$, introduce the
\begin{equation}
\aleph_{j}(x,t,\lambda)=I+\int_{\pm\infty}^{x}\Lambda(x,y,t,\lambda)\Lambda^{\natural}(y,t,\lambda)\aleph_{j}(y,t,\lambda)\Psi_{j}^{\pm\infty}(y,t,\lambda)(\Psi_{j}^{\pm\infty})^{-1}(x,t,\lambda)\mathrm{d}y.
\end{equation}
For every fixed $(y,t)$, the function $\Psi_{j}^{\pm\infty}(x,t,\lambda)(\Psi_{j}^{\pm\infty})^{-1}(y,t,\lambda)$ is a solution of the $x$-part with $q$ replaced by $q_{j}^{\pm\infty}$. Since this solution equals the identity matrix at $x=y$ and the matrix $M$ in the Lax pair (\ref{laxpair1}) is a polynomial in $\lambda$, we conclude that $\Psi_{j}^{\pm\infty}(x,t,\lambda)(\Psi_{j}^{\pm\infty})^{-1}(y,t,\lambda)$ is an entire function of $\lambda$, well defined for $\lambda\in\gamma_{1}\cup\gamma_{2}$. Thus, $\Psi_{j\pm}$ and $\aleph_{j}\Psi_{j\pm}^{\pm\infty}$ solve the same integral equation for $\lambda\in\Sigma_{j'},j'\neq j$. Hence, $\Psi_{1\pm}(x,t,\lambda)$ and $\Psi_{2\pm}(x,t,\lambda)$ can be written as follows for $\lambda\in\gamma_{1}\cup\overline{\gamma}_{1}\cup\gamma_{2}\cup\overline{\gamma}_{2}$:
\begin{subequations}
\begin{align}
&\Psi_{1\pm}=\aleph_{1}\Psi_{1\pm}^{-\infty},\quad \Psi_{2}=\aleph_{2}\Psi_{2}^{+\infty}\quad \lambda\in\gamma_{1}\cup\overline{\gamma}_{1},\\
&\Psi_{2\pm}=\aleph_{2}\Psi_{2\pm}^{+\infty},\quad \Psi_{1}=\aleph_{1}\Psi_{1}^{-\infty}\quad \lambda\in\gamma_{2}\cup\overline{\gamma}_{2},
\end{align}
\end{subequations}
 The scattering matrix $\tilde{S}_{\pm}(\lambda)$ on branch cuts $\gamma_{1}\cup\overline{\gamma}_{1}\cup\gamma_{2}\cup\overline{\gamma}_{2}$  and $\det{\tilde{S}}_{\pm}(\lambda)=1$. There exist the relations
 \begin{subequations}
 \begin{align}
 &\Psi_{2\pm}(x,t,\lambda)=\Psi_{1}(x,t,\lambda)\tilde{S}_{\pm}(\lambda),\quad \lambda\in\gamma_{2}\cup\overline{\gamma}_{2},\\
 &\Psi_{2}(x,t,\lambda)=\Psi_{1\pm}(x,t,\lambda)\tilde{S}_{\pm}(\lambda),\quad \lambda\in\gamma_{1}\cup\overline{\gamma}_{1},
 \end{align}\label{sigma12}
 \end{subequations}
 For $\lambda\in\gamma_{2}\cup\overline{\gamma}_{2}$, one can derive $\tilde{S}_{\pm}(\lambda)=\Psi_{1}^{-1}(x,t,\lambda)\aleph_{2}(x,t,\lambda)\Psi_{2\pm}^{+\infty}(x,t,\lambda)$.
Letting $x=t=0$, one have $\tilde{S}_{\pm}(\lambda)=P_{2}(\lambda)Y_{2\pm}(\lambda)$, where $P_{2}(\lambda)=\Psi_{1}^{-1}(0,0,\lambda)\aleph_{2}(0,0,\lambda)$. Hence
 \begin{equation}
 \tilde{S}_{+}(\lambda)=\tilde{S}_{-}(\lambda)\left(
 \begin{array}{cc}
 0&ie^{i\phi_{2}}\\
 ie^{-i\phi_{2}}&0
 \end{array}
 \right),\quad \lambda\in\gamma_{2}\cup\overline{\gamma}_{2}.
 \end{equation}
This implies that
 \begin{equation}
 \tilde{S}_{12+}=ie^{i\phi_{2}}\tilde{S}_{11-},\quad \tilde{S}_{22+}=ie^{i\phi_{2}}\tilde{S}_{21-}.
 \end{equation}
 The jump relation across $\gamma_{2}$ follows
 \begin{equation}
 \left(
 \begin{array}{cc}
 \frac{\Psi_{1}^{(1)}}{\tilde{a}_{+}}&\Psi_{2+}^{(2)}
 \end{array}
 \right)=\left(
 \begin{array}{cc}
 \frac{\Psi_{1}^{(1)}}{\tilde{a}_{-}}&\Psi_{2-}^{(2)}
 \end{array}
 \right)\left(
 \begin{array}{cc}
 \frac{\tilde{a}_{-}}{\tilde{a}_{+}}&c_{2}\\
 0&\frac{\tilde{a}_{+}}{\tilde{a}_{-}}
 \end{array}
 \right),
 \end{equation}
 where $c_{2}$ is some function about $\lambda$. Furtherly,
 \begin{equation}
 \frac{\Psi_{2+}^{(2)}}{\tilde{a}_{+}}-\frac{\Psi_{2-}^{(2)}}{\tilde{a}_{-}}=\frac{c_{2}}{\tilde{a}_{+}\tilde{a}_{-}}\Psi_{1}^{(1)}.
 \end{equation}
From (\ref{sigma12}),
 \begin{equation}
 \Psi_{2\pm}^{(2)}=\tilde{S}_{12\pm}\Psi_{1}^{(1)}+\tilde{S}_{22\pm}\Psi_{1}^{(2)}.
 \end{equation}
 One can derive $\tilde{S}_{22\pm}=\det{\left(
 \begin{array}{cc}
 \Psi_{1}^{(1)}&\Psi_{2\pm}^{(2)}
 \end{array}
 \right)}=\tilde{a}_{\pm}$ and
  \begin{equation}
 \frac{\Psi_{2+}^{(2)}}{\tilde{a}_{+}}-\frac{\Psi_{2-}^{(2)}}{\tilde{a}_{-}}=\left(\frac{\tilde{S}_{12+}}{\tilde{S}_{22+}}-\frac{\tilde{S}_{12-}}{\tilde{S}_{22-}}\right)\Psi_{1}^{(1)}.
 \end{equation}
Hence
  \begin{equation}
\frac{\tilde{S}_{12+}}{\tilde{S}_{22+}}-\frac{\tilde{S}_{12-}}{\tilde{S}_{22-}}=ie^{i\phi_{2}}\frac{\tilde{S}_{11-}\tilde{S}_{22-}-\tilde{S}_{12-}\tilde{S}_{21-}}{\tilde{S}_{22+}\tilde{S}_{22-}}=\frac{ie^{i\phi_{2}}}{\tilde{a}_{+}\tilde{a}_{-}}.
 \end{equation}
 Thus $c_{2}=ie^{i\phi_{2}}$.

  For $k\in\gamma_{1}\cup\overline{\gamma}_{1}$, one can derive $\tilde{S}_{\pm}(\lambda)=(\Phi_{1\pm}^{-\infty})^{-1}(x,t,\lambda)\aleph_{1}(x,t,\lambda)\Phi_{2}(x,t,\lambda)$.
Letting $x=t=0$, one have $\tilde{S}_{\pm}(\lambda)=Y_{1\pm}^{-1}(\lambda)P_{1}(k)$, where $P_{1}(\lambda)=\aleph_{1}^{-1}(0,0,\lambda)\Phi_{2}(0,0,\lambda)$. Hence
 \begin{equation}
 \tilde{S}_{-}(\lambda)\tilde{S}_{+}^{-1}(\lambda)=Y_{1-}^{-1}Y_{1+}=\left(
 \begin{array}{cc}
 0&ie^{i\phi_{1}}\\
 ie^{-i\phi_{1}}&0
 \end{array}
 \right).
 \end{equation}
This implies that
  \begin{equation}
 \tilde{S}_{-}(\lambda)=\left(
 \begin{array}{cc}
 0&ie^{i\phi_{1}}\\
 ie^{-i\phi_{1}}&0
 \end{array}
 \right)\tilde{S}_{+}(k).
 \end{equation}
That is
 \begin{equation}
 \tilde{S}_{21-}=ie^{-i\phi_{1}}\tilde{S}_{11+},\quad \tilde{S}_{22-}=ie^{-i\phi_{1}}\tilde{S}_{12+},\quad \lambda\in\gamma_{1}\cup\overline{\gamma}_{1}.
 \end{equation}
 According the jump relation across $\gamma_{2}$ as
 \begin{equation}
 \left(
 \begin{array}{cc}
 \frac{\Psi_{1+}^{(1)}}{\tilde{a}_{+}}&\Psi_{2}^{(2)}
 \end{array}
 \right)=\left(
 \begin{array}{cc}
 \frac{\Psi_{1-}^{(1)}}{\tilde{a}_{-}}&\Psi_{2}^{(2)}
 \end{array}
 \right)\left(
 \begin{array}{cc}
 1&0\\
 c_{1}&1
 \end{array}
 \right),
 \end{equation}
 where $c_{1}$ is some function about $\lambda$.
 \begin{equation}
 \frac{\Psi_{1+}^{(1)}}{\tilde{a}_{+}}-\frac{\Psi_{1-}^{(1)}}{\tilde{a}_{-}}=c_{1}\Psi_{2}^{(2)}.
 \end{equation}
From (\ref{sigma12}),
 \begin{equation}
 \Psi_{1\pm}^{(1)}=\tilde{S}_{22\pm}\Psi_{2}^{(1)}-\tilde{S}_{21\pm}\Psi_{2}^{(2)}.
 \end{equation}
 One can derive $\tilde{S}_{22\pm}=\det{\left(
 \begin{array}{cc}
 \Psi_{1}^{(1)}&\Psi_{2\pm}^{(2)}
 \end{array}
 \right)}=\tilde{a}_{\pm}$. Thus
 \begin{equation}
 \frac{\Psi_{1+}^{(1)}}{\tilde{a}_{+}}-\frac{\Psi_{1-}^{(1)}}{\tilde{a}_{-}}=\left(\frac{\tilde{S}_{21-}}{\tilde{S}_{22-}}-\frac{\tilde{S}_{21+}}{\tilde{S}_{22+}}\right)\Psi_{2}^{(2)}.
 \end{equation}
As above, one can derive
  \begin{equation}
\frac{\tilde{S}_{21-}}{\tilde{S}_{22-}}-\frac{\tilde{S}_{21+}}{\tilde{S}_{22+}}=\frac{ie^{i\phi_{1}}}{\tilde{a}_{+}\tilde{a}_{-}}.
 \end{equation}
 Thus $c_{1}=\frac{ie^{-i\phi_{1}}}{\tilde{a}_{+}\tilde{a}_{-}}$.\end{proof}
Furthermore, we note that the scattering datas obey the following relationships from the proof of Lemma \ref{lemma}
\begin{equation}
\begin{cases}
\tilde{a}_{+}=-ie^{-i\phi_{1}}\tilde{b}_{-}\\
\tilde{b}_{+}=-ie^{i\phi_{1}}\tilde{a}_{-}
\end{cases},\quad \lambda\in\gamma_{1}\cup\overline{\gamma}_{1},
\quad
\begin{cases}
\tilde{a}_{+}=-ie^{i\phi_{2}}\tilde{\overline{b}}_{-}\\
\tilde{b}_{+}=ie^{i\phi_{2}}\tilde{\overline{a}}_{-}
\end{cases},\quad \lambda\in\gamma_{2}\cup\overline{\gamma}_{2}.\nonumber
\end{equation}

\begin{figure}[H]
\begin{center}
\begin{tikzpicture}[scale=0.65]
\draw[thick,->] (-6,0)--(-3,0);
\draw[thick,->] (-3,0)--(0,0);
\draw[thick,->] (0,0)--(3,0);
\draw[thick,-] (3,0) -- (6,0);
\draw[thick,->] (-1.5,-2) -- (-1.5,-1);
\draw[thick,->] (-1.5,-1) -- (-1.5,1);
\draw[thick,-] (-1.5,1) -- (-1.5,2);
\draw[thick,->] (1.5,-2) -- (1.5,-1);
\draw[thick,->] (1.5,-1) -- (1.5,1);
\draw[thick,-] (1.5,1) -- (1.5,2);
\draw [] (1.8,2) circle [radius=0] node[right] {$E_{2}$};
\draw [] (-2.5,2) circle [radius=0] node[right] {$E_{1}$};
\draw [] (-2.5,-2) circle [radius=0] node[right] {$\overline{E}_{1}$};
\draw [] (1.8,-2) circle [radius=0] node[right] {$\overline{E}_{2}$};
\draw [] (4.4,0) circle [radius=0] node[below] {$\mathbb{R}$};
\draw [] (-1.5,1) circle [radius=0] node[left] {$\gamma_{1}$};
\draw [] (1.5,1) circle [radius=0] node[right] {$\gamma_{2}$};
\draw [] (-1.5,-1) circle [radius=0] node[left] {$\overline{\gamma}_{1}$};
\draw [] (1.5,-1) circle [radius=0] node[right] {$\overline{\gamma}_{2}$};
\draw [] (-1.0,0) circle [radius=0] node[below] {$-B$};
\draw [] (1.0,0) circle [radius=0] node[below] {$B$};
\end{tikzpicture}
\end{center}
\caption{The jump contour $\Sigma=\mathbb{R}\cup\gamma_{1}\cup\overline{\gamma}_{1}\cup\gamma_{2}\cup\overline{\gamma}_{2}$.}\label{con2}
\end{figure}
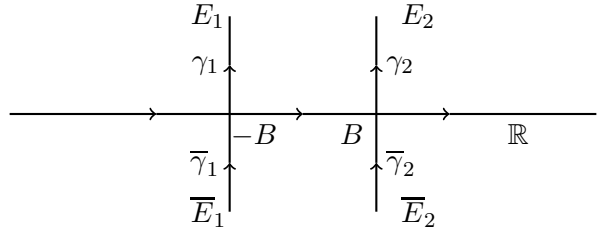
From the Lemma \ref{lemma}, we summarize the function $\tilde{m}(x,t,\lambda)$ satisfies the RH problem:
\begin{prob}
\begin{itemize} $\tilde{m}(x,t,\lambda)$ satisfies the RH problem
\item  $\tilde{m}(x,t,\lambda)$ is analytical in $\mathbb{C}\backslash\Sigma$, where jump contour $\Sigma$ see Figure \ref{con2}.
\item  $\tilde{m}(x,t,\lambda)$ satisfies the jump condition
\begin{equation}
\tilde{m}_{+}(x,t,\lambda)=\tilde{m}_{-}(x,t,\lambda)e^{-it(\xi\lambda+2\lambda^{2})\hat{\sigma}_{3}}\tilde{J}(x,t,\lambda),
\end{equation}
where the jump matrix $\tilde{J}(x,t,\lambda)$ is given by
\begin{equation}
\tilde{J}(x,t,\lambda)=\begin{cases}
\left(
\begin{array}{cc}
1+\lambda\varrho\overline{\varrho}&\overline{\varrho}\\
\lambda\varrho&1
\end{array}
\right),\quad~ \lambda\in\mathbb{R},\\
\left(
\begin{array}{cc}
1&0\\
\lambda f(\lambda)&1
\end{array}
\right),\quad\quad\lambda\in\gamma_{1},\\
\left(
\begin{array}{cc}
\frac{\tilde{a}_{-}}{\tilde{a}_{+}}&i\\
0&\frac{\tilde{a}_{+}}{\tilde{a}_{-}}
\end{array}
\right),\quad\quad~\lambda\in\gamma_{2},\\
\left(
\begin{array}{cc}
1&-\overline{f(\overline{\lambda})}\\
0&1
\end{array}
\right),\quad~~ \lambda\in\overline{\gamma}_{1},\\
\left(
\begin{array}{cc}
\frac{\tilde{\overline{a}}_{+}}{\tilde{\overline{a}}_{-}}&0\\
i&\frac{\tilde{\overline{a}}_{-}}{\tilde{\overline{a}}_{+}}
\end{array}
\right),\quad\quad~\lambda\in\overline{\gamma}_{2}.
\end{cases}\nonumber
\end{equation}
\item $\tilde{m}(x,t,\lambda)$ satisfies the asymptotic behavior
\begin{equation}
\tilde{m}(x,t,\lambda)\rightarrow I,\quad \lambda\rightarrow\infty.
\end{equation}
\end{itemize}
\end{prob}

\subsection{the basic Riemann-Hilbert problem}
It is necessary to regularize function $\tilde{m}(x,t,\lambda)$ to get the RH problem of subsequent deformations in the $L^{2}$. Define function $\hat{m}(x,t,\lambda)$ by
\begin{equation}
\hat{m}(x,t,\lambda)=\tilde{m}(x,t,\lambda)\nu^{-\sigma_{3}}(\lambda),\nonumber
\end{equation}
where $\nu=\left(\frac{\lambda-E_{1}}{\lambda-\overline{E}_{1}}\right)^{1/4}$ admits $\nu_{+}=i\nu_{-}$. This transformation implies that
$\hat{a}\hat{\overline{a}}=a\overline{a},~ \hat{b}\hat{\overline{b}}=b\overline{b},~ \hat{\varrho}=\varrho\nu^{-2},~\hat{\varrho}\hat{\overline{\varrho}}=\varrho\overline{\varrho}=|\varrho|^{2}$.
$\hat{a}$, $\hat{b}$ and $\frac{\hat{\overline{\varrho}}}{1+\lambda\hat{\varrho}\hat{\overline{\varrho}}}=\hat{a}\hat{b}$ are bounded near $E_{1}$, $\hat{\varrho}$ is bounded near $E_{2}$.

In order to study the long-time asymptotic of the step-like initial value problem for the DNLS equation (\ref{dnls3})
via the Deift-Zhou method, we refresh the RH problem based on the previous analyses as follow:
\begin{prob}\label{rhbsic}
\begin{itemize} $\hat{m}(x,t,\lambda)$ satisfies the RH problem
\item  $\hat{m}(x,t,\lambda)$ is analytical in $\mathbb{C}\backslash\Sigma$, where jump contour $\Sigma$ see Figure \ref{con2}.
\item  $\hat{m}(x,t,\lambda)$ satisfies the jump condition
\begin{equation}
\hat{m}_{+}(x,t,\lambda)=\hat{m}_{-}(x,t,\lambda)\hat{J}(\lambda),
\end{equation}
where the jump matrix
\begin{equation}
\hat{J}(x,t,\lambda)=e^{-it\theta(\lambda)\sigma_{3}}\hat{J}^{(0)}(\lambda)e^{it\theta(\lambda)\sigma_{3}},\quad \lambda\in\Sigma,
\end{equation}
and
\begin{equation}
\hat{J}^{(0)}(\lambda)=\begin{cases}
\left(\begin{array}{cc}
1+\lambda\hat{\varrho}\hat{\overline{\varrho}}&\hat{\overline{\varrho}}\\
\lambda\hat{\varrho}&1
\end{array}
\right),\quad \lambda\in\mathbb{R},\\
\left(\begin{array}{cc}
-i&0\\
\lambda\hat{f}&i
\end{array}
\right),\quad\quad\quad \lambda\in\gamma_{1},\\
\left(\begin{array}{cc}
-i&-\hat{\overline{f}}\\
0&i
\end{array}
\right),\quad\quad~ \lambda\in\overline{\gamma}_{1},\\
\left(
\begin{array}{cc}
\frac{\hat{a}_{-}}{\hat{a}_{+}}&i\nu^{2}\\
0&\frac{\hat{a}_{+}}{\hat{a}_{-}}
\end{array}
\right),\quad\quad~ \lambda\in\gamma_{2},\\
\left(
\begin{array}{cc}
\frac{\hat{\overline{a}}_{+}}{\hat{\overline{a}}_{-}}&0\\
i\nu^{-2}&\frac{\hat{\overline{a}}_{-}}{\hat{\overline{a}}_{+}}
\end{array}
\right),\quad~ \lambda\in\overline{\gamma}_{2}.
\end{cases}\nonumber
\end{equation}
\item $\hat{m}(x,t,\lambda)$ satisfies the asymptotic behavior
\begin{equation}
\hat{m}(x,t,\lambda)\rightarrow I,\quad \lambda\rightarrow\infty.
\end{equation}
\end{itemize}
\end{prob}
The solution $q(x,t)$ of the DNLS equation (\ref{dnls3}) with the initial value is reconstructed by
\begin{equation}
q(x,t)=2i\lim_{\lambda\rightarrow\infty}\lambda\hat{m}_{12}(x,t,\lambda).
\end{equation}
The jump matrix $\hat{J}(x,t,\lambda)$ and $\hat{m}(x,t,\lambda)$ admit the following symmetries
\begin{subequations}
\begin{align}
&\hat{m}(x,t,\lambda)=\sigma_{1}\sigma_{3}\overline{\hat{m}(x,t,\overline{\lambda})}\sigma_{3}\sigma_{1},\quad\quad~ \lambda\in\mathbb{C}\backslash\Sigma,\nonumber\\
&\hat{J}(x,t,\lambda)=\begin{cases}
\sigma_{1}\sigma_{3}\overline{\hat{J}(x,t,\overline{\lambda})}\sigma_{3}\sigma_{1},\quad\quad \lambda\in\gamma_{1}\cup\overline{\gamma}_{1}\cup\gamma_{2}\cup\overline{\gamma}_{2},\\
\sigma_{1}\sigma_{3}\overline{\hat{J}(x,t,\overline{\lambda})}^{-1}\sigma_{3}\sigma_{1}, \quad \lambda\in\mathbb{R}.
\end{cases}\nonumber
\end{align}
\end{subequations}

\subsection{$g$-function}
Consider the initial value $D_{1}=D_{2}=D$, $B_{2}=-B_{1}=B=1$, $\phi_{1}=\phi$ and $\phi_{2}=0$,
for $\frac{D}{B}>1$ and $\varepsilon<|\xi|<\xi_{0}$, $\xi_{0}$ is some positive constant and $\varepsilon\in(0,\xi_{0})$, an $g$-function mechanism need to introduce for solving the genus 3 asymptotics.
Define the $g$-function by
\begin{equation}
g(\lambda)=\int_{\overline{E}_{2}}^{\lambda}\mathrm{d}g,\quad \lambda\in\mathbb{C}\backslash\Sigma^{mod},\label{g}
\end{equation}
where
$\Sigma^{mod}=\gamma_{1}\cup\overline{\gamma}_{1}\cup\gamma_{2}\cup\overline{\gamma}_{2}\cup\gamma_{(\beta,\alpha)}\cup\gamma_{(\overline{\alpha},\overline{\beta})}\cup\gamma_{(\overline{\beta},\beta)}$.
And $\mathrm{d}g$ defined by $\frac{\mathrm{d}{g}}{\mathrm{d}{\lambda}}=\frac{\curlyvee}
{\hbar}$,
where
\begin{equation}
\begin{split}
&\curlyvee=4(\lambda-\mu)(\lambda-\alpha)(\lambda-\overline{\alpha})(\lambda-\beta)(\lambda-\overline{\beta}),\\
&\hbar=\left((\lambda-E_{1})(\lambda-\overline{E}_{1})(\lambda-E_{2})(\lambda-\overline{E}_{2})(\lambda-\alpha)(\lambda-\overline{\alpha})(\lambda-\beta)(\lambda-\overline{\beta})\right)^{1/2},
\end{split}\nonumber
\end{equation}
with $\mu\in\mathbb{R}$, $\alpha=\mathrm{Re}\alpha+i\mathrm{Im}\alpha$, $\beta=\mathrm{Re}\beta+i\mathrm{Im}\beta$ which are determined by
\begin{subequations}
\begin{align}
&\int_{a_{1}}\hat{\mathrm{d}g}=\int_{a_{2}}\hat{\mathrm{d}g}=\int_{a_{3}}\hat{\mathrm{d}g}=0,\\
&\lim_{\lambda\rightarrow\infty}\left(\frac{\mathrm{d}g}{\mathrm{d}\lambda}-4\lambda\right)=\xi,\quad \lim_{\lambda\rightarrow\infty}4\left(\frac{\mathrm{d}g}{\mathrm{d}\lambda}-4\lambda-\xi\right)=0.
\end{align}\label{condition}
\end{subequations}
the contour $a_{j}$ see Figure \ref{remiannsur}.
The systems of (\ref{condition}) ensure that
\begin{equation}
\begin{split}
&\mathrm{d}g(\lambda)=4\lambda+\xi+\mathcal{O}(\lambda^{-2}),\quad \lambda\rightarrow\infty,\\
&g(\lambda)=\theta(\lambda)+g^{(0)}+\mathcal{O}(\lambda^{-1}),\quad \lambda\rightarrow\infty.
\end{split}\label{g0}
\end{equation}

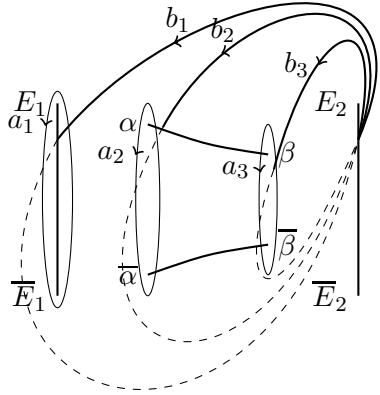
\begin{figure}[H]
\begin{center}
\begin{tikzpicture}[scale=0.4]
\draw[thick](5,-3.2)--(5,3.2);
\draw[thick] (-5,-3.2)--(-5,3.2);
\draw[black] (-5,0) ellipse (0.5 and 3.6);
\draw[black] (-2,0) ellipse (0.4 and 3.2);
\draw[black] (2,0) ellipse (0.3 and 2.5);
\draw[thick] (-2,2.5).. controls  (0,1.8) ..(2,1.5);
\draw[thick] (-2,-2.5).. controls  (0,-1.8) ..(2,-1.5);
\draw[thick] (5,2.0).. controls  (8,9) and (-1,7)..(-5,2.0);
\draw[dashed] (-5,2.0).. controls  (-10,-7) and (2,-11)..(5,2.0);
\draw[thick] (5,2.0).. controls  (7,8) and (1,7)..(-1.5,2.4);
\draw[dashed] (-1.5,2.4).. controls  (-6,-6) and (2,-8)..(5,2.0);
\draw[thick] (5,2.0).. controls  (6,7) and (3.5,6)..(2.19,1);
\draw[dashed] (2.19,1).. controls  (0.8,-3) and (2,-5)..(5,2.0);
\draw[thick,->] (-1,5.25)--(-1.15,5.2);
\draw[thick,->] (0.5,4.89)--(0.45,4.87);
\draw[thick,->] (3.7,4.5)--(3.69,4.48);
\draw[thick,->] (-5.4,2.5)--(-5.41,2.45);
\draw[thick,->] (-2.405,1.5)--(-2.41,1.49);
\draw[thick,->] (1.73,1.0)--(1.75,0.9);
\draw[ ](5,-3.2)node[left] {$\overline{E}_{2}$} (-2,2.5) node[left] {$\alpha$} (-2,-2.5) node[left] {$\overline{\alpha}$}(2,1.5) node[right] {$\beta$} (2,-1.5) node[right] {$\overline{\beta}$};
\draw[ ](5,3.2)node[left] {$E_{2}$} (-5,-3.2)node[left] {$\overline{E}_{1}$} (-1,5.25) node[above] {$b_{1}$} (0.5,4.89) node[above] {$b_{2}$}  (3.7,4.5) node[left] {$b_{3}$};
\draw[ ](-5,3.2)node[left] {$E_{1}$} (-5.4,2.5)node[left] {$a_{1}$} (-2.405,1.5)node[left] {$a_{2}$} (1.73,1.0)node[left] {$a_{3}$};
\end{tikzpicture}
\end{center}
\caption{The contour $a_{j}$ and $b_{j}$ of Riemann surface with genus 3.}\label{remiannsur}
\end{figure}

\begin{lemma}\label{lemg}
The $g$-function defined by (\ref{g}) with the following properties:\\
$\blacktriangleright$ $g(\lambda)-\theta(\lambda)$ is analytic and bounded for $\lambda\in\hat{\mathbb{C}}\backslash\Sigma^{mod}$ with continuous boundary values on $\Sigma^{mod}$, where $\hat{\mathbb{C}}=\mathbb{C}\cup{\infty}$.\\
$\blacktriangleright$ $g(\lambda)$ admits the symmetry $g(\lambda)=\overline{g(\overline{\lambda})}$.\\
$\blacktriangleright$ $g(\lambda)$ admits the RH problem
\begin{equation}
\begin{split}
&g_{+}(\lambda)+g_{-}(\lambda)=
\begin{cases}
2\Delta_{1},\quad \lambda\in\gamma_{1}\cup\overline{\gamma}_{1},\\
2\Delta_{2},\quad \lambda\in\gamma_{(\overline{\alpha},\overline{\beta})}\cup\gamma_{(\beta,\alpha)},\\
0,\quad\quad~ \lambda\in\gamma_{2}\cup\overline{\gamma}_{2},
\end{cases}\\
&g_{+}(\lambda)-g_{-}(\lambda)=2\Delta_{3},\quad\quad \lambda\in\gamma_{(\overline{\beta},\beta)},
\end{split}\nonumber
\end{equation}
where
\begin{equation}
\Delta_{1}=g(E_{1})g(\overline{E}_{1}),\quad \Delta_{2}=g(\alpha)g(\overline{\alpha}),\quad \Delta_{3}=\frac{g_{+}(\beta)-g_{-}(\beta)}{2}=\frac{g_{+}(\overline{\beta})-g_{-}(\overline{\beta})}{2}.\nonumber
\end{equation}
\end{lemma}
\section{Deformation of the Jump Contour}\label{deformation}
In this section, our main purpose is to re-normalize the RH Problem \ref{rhbsic} such
that it is well-behaved as $t\rightarrow\infty$ along any characteristic by establishing  transformation $\hat{m}\rightarrow\hat{m}^{(j)}$. For deriving the long-time asymptotic of equation (\ref{dnls3}), the jump matrixes need to be transformed as a constant matrix or decay to identity matrix. Now we will perform five transformations for the RH problem.

Introduce the matrix function $\hat{m}^{(1)}(x,t,\lambda)$
\begin{equation}
\hat{m}^{(1)}(x,t,\lambda)=e^{-itg^{(0)}\sigma_{3}}\hat{m}(x,t,\lambda)e^{-it(g(\lambda)-\theta(\lambda))\sigma_{3}},
\end{equation}
where $g^{(0)}=(g-\theta)(\xi,\infty)$.
Function $\hat{m}^{(1)}(x,t,\lambda)$ admits the following RH problem:
\begin{prob}
Function $\hat{m}^{(1)}(x,t,\lambda)$ satisfies the following jump condition:
\begin{itemize}
\item  $\hat{m}^{(1)}(x,t,\lambda)$ is analytical in $\mathbb{C}\backslash\Sigma^{(1)}$, where $\Sigma^{(1)}$ see Figure \ref{sigma1}.
\item  $\hat{m}^{(1)}(x,t,\lambda)$ satisfies the jump condition
\begin{equation}
\hat{m}^{(1)}_{+}=\hat{m}^{(1)}_{-}\hat{J}^{(1)},
\end{equation}
where the jump matrix
\begin{equation}
\hat{J}^{(1)}=\begin{cases}
\left(\begin{array}{cc}
1&\hat{\overline{\varrho}}e^{-2itg}\\
0&1
\end{array}
\right)\left(\begin{array}{cc}
1&0\\
\lambda\hat{\varrho}e^{2itg}&1
\end{array}
\right),\quad\quad\quad \lambda\in(\mu,+\infty),\\
\left(
\begin{array}{cc}
1&0\\
\frac{\lambda\hat{\varrho}e^{2itg(k)}}{1+\lambda\hat{\varrho}\hat{\overline{\varrho}}}&1
\end{array}
\right)\left(
\begin{array}{cc}
1+\lambda\hat{\varrho}\hat{\overline{\varrho}}&0\\
0&\frac{1}{1+\lambda\hat{\varrho}\hat{\overline{\varrho}}}
\end{array}
\right)\left(
\begin{array}{cc}
1&\frac{\hat{\overline{\varrho}}e^{-2itg(k)}}{1+\lambda\hat{\varrho}\hat{\overline{\varrho}}}\\
0&1
\end{array}
\right),\quad \lambda\in(-\infty,\mu),\\
\left(
\begin{array}{cc}
-ie^{it(g_{+}-g_{-})}&0\\
\lambda\hat{f}e^{it(g_{+}+g_{-})}&ie^{-it(g_{+}-g_{-})}
\end{array}
\right),\quad\quad~~ \lambda\in\gamma_{1},\\
\left(
\begin{array}{cc}
\frac{\hat{a}_{-}}{\hat{a}_{+}}e^{it(g_{+}-g_{-})}&i\nu^{2}e^{-it(g_{+}-g_{-})}\\
0&\frac{\hat{a}_{+}}{\hat{a}_{-}}e^{-it(g_{+}-g_{-})}
\end{array}
\right),~~\quad\lambda\in\gamma_{2},\\
\left(
\begin{array}{cc}
-ie^{it(g_{+}-g_{-})}&-\hat{\overline{f}}e^{-it(g_{+}+g_{-})}\\
0&ie^{-it(g_{+}-g_{-})}
\end{array}
\right),\quad~~ \lambda\in\overline{\gamma}_{1},\\
\left(
\begin{array}{cc}
\frac{\hat{\overline{a}}_{+}}{\hat{\overline{a}}_{-}}e^{it(g_{+}-g_{-})}&0\\
i\nu^{-2}e^{it(g_{+}+g_{-})}&\frac{\hat{\overline{a}}_{-}}{\hat{\overline{a}}_{+}}e^{-it(g_{+}-g_{-})}
\end{array}
\right),\quad \lambda\in\overline{\gamma}_{2},\\
\left(
\begin{array}{cc}
e^{it(g_{+}-g_{-})}&0\\
0&e^{-it(g_{+}-g_{-})}
\end{array}
\right),\quad\quad\quad\quad~ \lambda\in\gamma_{(\beta,\alpha)}\cup\gamma_{(\overline{\alpha},\overline{\beta})}\cup\gamma_{(\overline{\beta},\beta)}.
\end{cases}\nonumber
\end{equation}
\item $\hat{m}^{(1)}(x,t,\lambda)$ satisfies the asymptotic behavior
\begin{equation}
\hat{m}^{(1)}(x,t,\lambda)\rightarrow I,\quad \lambda\rightarrow\infty.
\end{equation}
\end{itemize}
\end{prob}

\begin{figure}[H]
\begin{center}
\begin{tikzpicture}[scale=0.45]
\draw[dashed] (-6.2,0)--(6.2,0);
\draw[thick,->] (4,-3.2)--(4,-1.6);
\draw[thick,->] (4,-1.6)--(4,1.6);
\draw[thick,-] (4,1.6)--(4,3.2);
\draw[thick,->] (-4,-3.2)--(-4,-1.6);
\draw[thick,->] (-4,-1.6)--(-4,1.6);
\draw[thick,-] (-4,1.6)--(-4,3.2);
\draw[dashed] (-1.5,6) parabola bend (-2.8,3)  (-4,3.2);
\draw[dashed] (-1.5,-6) parabola bend (-2.8,-3)  (-4,-3.2);
\draw[dashed] (0.4,2) .. controls (0,0)  .. (0.4,-2);
\draw[dashed] (4,3.2)--(0.4,2);
\draw[dashed] (4,-3.2)--(0.4,-2);
\draw[dashed] (0.4,2).. controls (-1,2.0)..(-2.8,3);
\draw[dashed] (0.4,-2).. controls (-1,-2.0)..(-2.8,-3);
\draw[thick,->] (-6.2,0)--(-2.5,0);
\draw[thick,->] (-2.5,0)--(2.5,0);
\draw[thick,-] (2.5,0)--(6.2,0);
\draw[thick] (0.4,2).. controls (-1,2.0)..(-2.8,3);
\draw[thick] (0.4,-2).. controls (-1,-2.0)..(-2.8,-3);
\draw[thick] (0.4,2) .. controls (0,0)  .. (0.4,-2);
\draw[thick,->] (0.2,0.9)--(0.2,1);
\draw[thick,->] (0.2,-1)--(0.2,-0.9);
\draw[thick,->] (-0.9,2.05)--(-1,2.1);
\draw[thick,->] (-0.9,-2.05)--(-1,-2.1);
\draw[ ](0.3,0)node[below] {\footnotesize$\mu$} (4,-3.2)node[below] {\footnotesize$\overline{E}_{2}$};
\draw[ ](4,3.2)node[above] {\footnotesize$E_{2}$} (-4,-3.2)node[below] {\footnotesize$\overline{E}_{1}$};
\draw[ ](-4,3.2)node[above] {\footnotesize$E_{1}$} (5,0)node[above] {\footnotesize$\mathbb{R}$};
\draw[ ](-2.5,3.0)node[right] {\footnotesize$\alpha$} (-2.5,-3.0)node[right] {\footnotesize$\overline{\alpha}$} (0.3,2.0)node[above] {\footnotesize$\beta$}  (0.3,-2.0)node[below] {\footnotesize$\overline{\beta}$};
\draw[ ](-5,5)node[above] {$-$} (2,-5)node[above] {$-$} (2,1)node[above] {$-$} (-2,-1)node[below] {$-$}
(-5,-5)node[above] {$+$} (2,5)node[above] {$+$} (2,-1)node[below] {$+$} (-2,1)node[above] {$+$};
\end{tikzpicture}
\end{center}
\caption{The branch cuts and the level set $\mathrm{Im}g=0$ (the dash line and $\mathbb{R}\cup\gamma_{1}\cup\overline{\gamma}_{1}\cup\gamma_{2}\cup\overline{\gamma}_{2}\cup\gamma_{(\overline{\alpha},\overline{\beta})}\cup\gamma_{(\overline{\beta},\beta)}\cup\gamma_{(\beta,\alpha)}$). The region of $\mathrm{Im}g>0$ is "+" and $\mathrm{Im}g<0$ is "-". The jump contour $\Sigma^{(1)}=\mathbb{R}\cup\gamma_{1}\cup\overline{\gamma}_{1}\cup\gamma_{2}\cup\overline{\gamma}_{2}\cup\gamma_{(\overline{\alpha},\overline{\beta})}\cup\gamma_{(\overline{\beta},\beta)}\cup\gamma_{(\beta,\alpha)}$.}
\label{sigma1}
\end{figure}
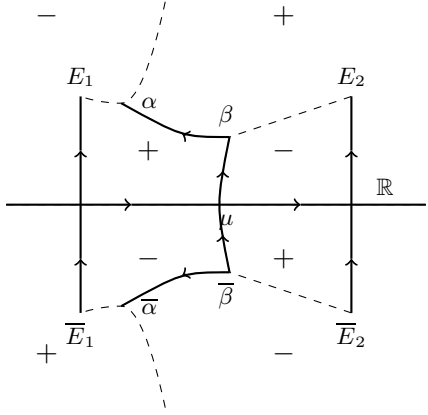

There is a bad factorization in jump matrix $\hat{J}^{(1)}$ for $\lambda\in(-\infty,\mu)$. To eliminate the intermediate matrix, we define function $\hat{m}^{(2)}(x,t,\lambda)$ by
\begin{equation}
\hat{m}^{(2)}(x,t,\lambda)=\hat{m}^{(1)}(x,t,\lambda)\delta^{-\sigma_{3}}(\lambda),
\end{equation}
where
\begin{equation}
\delta(\lambda)=\exp{\left[\frac{-i}{2\pi}\int_{-\infty}^{\mu}\frac{\ln{(1+s|\hat{\varrho}(s)|^{2})}}{s-\lambda}\mathrm{d}s\right]},\quad \lambda\in\mathbb{C}\backslash(-\infty,\mu].
\end{equation}
The function $\delta(\lambda)$ with the scalar RH problem:
\begin{lemma}\label{delta}
the function $\delta(\lambda)$ admits the following properties:\\
$\blacktriangleright$ $\delta(\lambda)$ and $\delta^{-1}(\lambda)$ are bounded and analytic for $\lambda\in\mathbb{C}\backslash(-\infty,\mu]$ with continuous boundary values on $(-\infty,\mu)$.\\
$\blacktriangleright$ $\delta$ admits the symmetry $\delta(\lambda)=\overline{\delta(\overline{\lambda})}^{-1}$.\\
$\blacktriangleright$ $\delta$ admits the jump condition
\begin{equation}
\begin{split}
&\delta_{+}=\delta_{-}(1+\lambda|\hat{\varrho}|^{2}),\quad (-\infty,\mu),\\
&\delta_{+}=\delta_{-},\quad\quad\quad\quad\quad~ (\mu,+\infty).
\end{split}\nonumber
\end{equation}
$\blacktriangleright$ $\delta(\lambda)$ admits the asymptotic behavior
\begin{equation}
\delta(\lambda)=1+\mathcal{O}(\lambda^{-1}),\quad \lambda\rightarrow\infty.
\end{equation}
\end{lemma}
Lemma \ref{delta} implies that $\delta^{\sigma_{3}}$ satisfies a $L^{2}$-RH problem. Hence $\hat{m}$ satisfies the RH problem \ref{rhbsic} iff $\hat{m}^{(2)}$ satisfies the RH problem:
\begin{prob}
$\hat{m}^{(2)}(x,t,\lambda)$ satisfies the following properties:
\begin{itemize}
\item  $\hat{m}^{(2)}(x,t,\lambda)$ is analytical in $\mathbb{C}\backslash\Sigma^{(2)}$, where $\Sigma^{(2)}$ see Figure \ref{sigma1}.
\item  $\hat{m}^{(2)}(x,t,\lambda)$ satisfies the jump condition
\begin{equation}
\hat{m}^{(2)}_{+}=\hat{m}^{(2)}_{-}\hat{J}^{(2)},
\end{equation}
where the jump matrix $\hat{J}^{(2)}$ is given by
\begin{equation}
\hat{J}^{(2)}=\begin{cases}
\left(\begin{array}{cc}
1&0\\
\frac{\lambda\hat{\varrho}}{1+\lambda\hat{\varrho}\hat{\overline{\varrho}}}\delta_{-}^{-2}e^{2itg}&1
\end{array}\right)\left(\begin{array}{cc}
1&\frac{\hat{\overline{\varrho}}}{1+\lambda\hat{\varrho}\hat{\overline{\varrho}}}\delta_{+}^{2}e^{-2itg}\\
0&1
\end{array}\right),\quad \lambda\in(-\infty,\mu),\\
\left(\begin{array}{cc}
1&\hat{\overline{\varrho}}\delta^{2}e^{-2itg}\\
0&1
\end{array}
\right)\left(\begin{array}{cc}
1&0\\
\lambda\hat{\varrho}\delta^{-2}e^{2itg}&1
\end{array}
\right),\quad\quad\quad\quad~ \lambda\in(\mu,+\infty),\\
\left(
\begin{array}{cc}
-ie^{it(g_{+}-g_{-})}&0\\
\lambda\hat{f}\delta^{-2}e^{it(g_{+}+g_{-})}&ie^{-it(g_{+}-g_{-})}
\end{array}
\right),\quad\quad\quad\quad\quad~ \lambda\in\gamma_{1},\\
\left(
\begin{array}{cc}
\frac{\hat{a}_{-}}{\hat{a}_{+}}e^{it(g_{+}-g_{-})}&i\nu^{2}\delta^{2}e^{-it(g_{+}+g_{-})}\\
0&\frac{\hat{a}_{+}}{\hat{a}_{-}}e^{-it(g_{+}-g_{-})}
\end{array}
\right),\quad\quad\quad\quad\quad \lambda\in\gamma_{2},\\
\left(
\begin{array}{cc}
-ie^{it(g_{+}-g_{-})}&-\hat{\overline{f}}\delta^{2}e^{-it(g_{+}+g_{-})}\\
0&ie^{-it(g_{+}-g_{-})}
\end{array}
\right),\quad\quad\quad\quad\quad \lambda\in\overline{\gamma}_{1},\\
\left(
\begin{array}{cc}
\frac{\hat{\overline{a}}_{+}}{\hat{\overline{a}}_{-}}e^{it(g_{+}-g_{-})}&0\\
i\nu^{-2}\delta^{-2}e^{it(g_{+}+g_{-})}&\frac{\hat{\overline{a}}_{-}}{\hat{\overline{a}}_{+}}e^{-it(g_{+}-g_{-})}
\end{array}
\right),\quad\quad\quad~~ \lambda\in\overline{\gamma}_{2},\\
\left(
\begin{array}{cc}
e^{it(g_{+}-g_{-})}&0\\
0&e^{-it(g_{+}-g_{-})}
\end{array}
\right),\quad\quad \lambda\in\gamma_{(\beta,\alpha)}\cup\gamma_{(\overline{\alpha},\overline{\beta})}\cup\gamma_{(\overline{\beta},\beta)}.
\end{cases}\nonumber
\end{equation}
\item $\hat{m}^{(2)}(x,t,\lambda)$ satisfies the asymptotic behavior
\begin{equation}
\hat{m}^{(2)}(x,t,\lambda)\rightarrow I,\quad \lambda\rightarrow\infty.
\end{equation}
\end{itemize}
\end{prob}
The purpose of the third deformation of the RH problem is to extend the jump matrix off the real axis. Then, the complex plane $\mathbb{C}$ is separated into six sectors which are respectively denoted by $U_{j} (j = 1, 2,..., 6)$. The distributions of $U_{j}$ are shown in Figure \ref{uj}. With this deformation, we define a new function $\hat{m}^{(3)}$ that deforms the oscillation term along the real axis onto new contours. Along the new contours, the deformed
oscillation term is decaying. The function $\hat{m}^{(3)}$  defined by
\begin{equation}
\hat{m}^{(3)}=\hat{m}^{(2)}\begin{cases}
\left(
\begin{array}{cc}
1&0\\
-\lambda\hat{\varrho}\delta^{-2}e^{2itg}&1
\end{array}
\right),\quad\quad\quad \lambda\in U_{1},\\
\left(
\begin{array}{cc}
1&-\frac{\hat{\overline{\varrho}}}{1+\lambda\hat{\varrho}\hat{\overline{\varrho}}}\delta^{2}e^{-2itg}\\
0&1
\end{array}
\right),\quad~~  \lambda\in U_{3},\\
\left(
\begin{array}{cc}
1&0\\
\frac{\lambda\hat{\varrho}}{1+\lambda\hat{\varrho}\hat{\overline{\varrho}}}\delta^{-2}e^{2itg}&1
\end{array}
\right),\quad\quad~  \lambda\in U_{4},\\
\left(
\begin{array}{cc}
1&\hat{\overline{\varrho}}\delta^{2}e^{-2itg}\\
0&1
\end{array}
\right),\quad\quad\quad\quad~ \lambda\in U_{6},\\
I,\quad\quad\quad\quad\quad\quad\quad\quad\quad\quad\quad~~ \lambda\in U_{2}\cup U_{5}.
\end{cases}\nonumber
\end{equation}
By using the identity
\begin{equation}
\begin{split}
&i=\lambda\hat{f}\left(\frac{\hat{\overline{\varrho}}}{1+\lambda\hat{\varrho}\hat{\overline{\varrho}}}\right)_{-},\quad -i=\lambda\hat{f}\left(\frac{\hat{\overline{\varrho}}}{1+\lambda\hat{\varrho}\hat{\overline{\varrho}}}\right)_{+},\\
&\frac{\hat{a}_{-}}{\hat{a}_{+}}+i\lambda\hat{\rho}_{+}\nu^{2}=0,\quad \frac{\hat{a}_{+}}{\hat{a}_{-}}+i\lambda\hat{\rho}_{-}\nu^{2}=0,
\end{split}\nonumber
\end{equation}
the function $\hat{m}^{(3)}$ admits the  RH problem:
\begin{prob}
$\hat{m}^{(3)}(x,t,\lambda)$ satisfies the following properties:
\begin{itemize}
\item  $\hat{m}^{(3)}(x,t,\lambda)$ is analytical in $\mathbb{C}\backslash\Sigma^{(3)}$, where $\Sigma^{(3)}$ see Figure \ref{sigma3}.
\item  $\hat{m}^{(3)}(x,t,\lambda)$ satisfies the jump condition
\begin{equation}
\hat{m}^{(3)}_{+}=\hat{m}^{(3)}_{-}\hat{J}^{(3)},
\end{equation}
where the jump matrix $\hat{J}^{(3)}=\hat{J}_{j}^{(3)}$ in the upper half-plane rewritten as
\begin{subequations}
\begin{align}
&\hat{J}_{1}^{(3)}=\left(
\begin{array}{cc}
1&0\\
-\lambda\hat{\varrho}\delta^{-2}e^{2itg}&1
\end{array}
\right),\quad \hat{J}_{2}^{(3)}=\left(
\begin{array}{cc}
1&-\frac{\hat{\overline{\varrho}}}{1+\lambda\hat{\varrho}\hat{\overline{\varrho}}}\delta^{2}e^{-2itg}\\
0&1
\end{array}
\right),\nonumber\\
&\hat{J}_{3}^{(3)}\left(
\begin{array}{cc}
0&(\lambda\hat{f})^{-1}\delta^{2}e^{-it(g_{+}-g_{-})}\\
\lambda\hat{f}\delta^{-2}e^{it(g_{+}+g_{-})}&0
\end{array}
\right),\nonumber\\
&\hat{J}_{4}^{(3)}=\left(
\begin{array}{cc}
0&i\nu^{2}\delta^{2}e^{-it(g_{+}+g_{-})}\\
i\nu^{-2}\delta^{-2}e^{it(g_{+}+g_{-})}&0
\end{array}
\right),\nonumber\\
&\hat{J}_{5}^{(3)}=e^{-itg_{-}\sigma_{3}}\left(
\begin{array}{cc}
1&-i(\lambda\hat{f})^{-1}\delta^{2}\\
0&1
\end{array}
\right)e^{itg_{+}\sigma_{3}},\nonumber\\
&\hat{J}_{6}^{(3)}=e^{-itg_{-}\sigma_{3}}\left(
\begin{array}{cc}
1&-i(\lambda\hat{f})^{-1}\delta^{2}\\
\lambda\hat{\varrho}\delta^{-2}&\hat{a}\hat{\overline{a}}
\end{array}
\right)e^{itg_{+}\sigma_{3}}.\nonumber
\end{align}
\end{subequations}
where the subscript of $\hat{J}^{(3)}_{j}$ denote the jump contour in Figure \ref{sigma3}.
\item $\hat{m}^{(3)}(x,t,\lambda)$ satisfies the asymptotic behavior
\begin{equation}
\hat{m}^{(3)}(x,t,\lambda)\rightarrow I,\quad \lambda\rightarrow\infty.
\end{equation}
\end{itemize}
\end{prob}

\begin{figure}[H]
\begin{center}
\begin{tikzpicture}[scale=0.45]
\draw[dashed] (-6.2,0)--(6.2,0);
\draw[dashed] (-1.5,6) parabola bend (-2.8,3)  (-4,3.2);
\draw[dashed] (-1.5,-6) parabola bend (-2.8,-3)  (-4,-3.2);
\draw[dashed] (0.4,2) .. controls (0,0)  .. (0.4,-2);
\draw[dashed] (4,3.2)--(0.4,2);
\draw[dashed] (4,-3.2)--(0.4,-2);
\draw[dashed] (0.4,2).. controls (-1,2.0)..(-2.8,3);
\draw[dashed] (0.4,-2).. controls (-1,-2.0)..(-2.8,-3);
\draw[thick] (-2.8,3)--(-6.2,5);
\draw[thick] (-2.8,-3)--(-6.2,-5);
\draw[thick] (0.4,2)--(0.5,3);
\draw[thick] (0.5,3)--(6.2,5);
\draw[thick] (0.4,-2)--(0.5,-3);
\draw[thick] (0.5,-3)--(6.2,-5);
\draw[ ](0.4,1.8)node[right] {\footnotesize$\beta$} (0.4,-1.8)node[right] {\footnotesize$\overline{\beta}$};
\draw[ ](-2.6,3)node[right] {\footnotesize$\alpha$} (-2.6,-3)node[right] {\footnotesize$\overline{\alpha}$};
\draw[ ](0.3,0)node[below] {\footnotesize$\mu$} (4,-3.2)node[below] {\footnotesize$\overline{E}_{2}$};
\draw[ ](4,3.2)node[above] {\footnotesize$E_{2}$} (-4,-3.2)node[left] {\footnotesize$\overline{E}_{1}$};
\draw[ ](-4,3.2)node[left] {\footnotesize$E_{1}$} ;
\draw[ ] (3,1.5) node[right] {\footnotesize$U_{1}$} (3,-1.5) node[right] {\footnotesize$U_{6}$};
\draw[ ] (-3,1.5) node[right] {\footnotesize$U_{3}$} (-3,-1.5) node[right] {\footnotesize$U_{4}$};
\draw[ ] (0,4) node[right] {\footnotesize$U_{2}$} (0,-4) node[right] {\footnotesize$U_{5}$};
\end{tikzpicture}
\end{center}
\caption{The distribution of $U_{j},j=1,2,\cdots,6$.}\label{uj}
\end{figure}
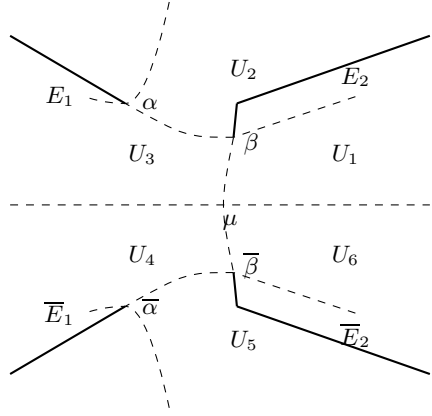

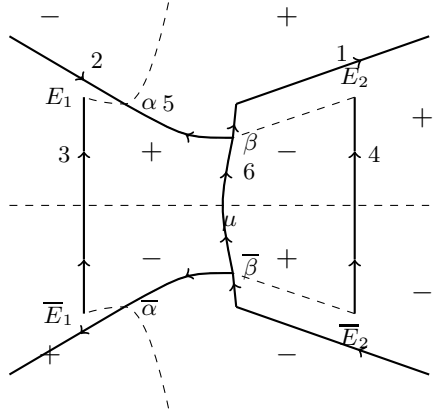
\begin{figure}[H]
\begin{center}
\begin{tikzpicture}[scale=0.45]
\draw[thick,->] (4,-3.2)--(4,-1.6);
\draw[thick,->] (4,-1.6)--(4,1.6);
\draw[thick,-] (4,1.6)--(4,3.2);
\draw[thick,->] (-4,-3.2)--(-4,-1.6);
\draw[thick,->] (-4,-1.6)--(-4,1.6);
\draw[thick,-] (-4,1.6)--(-4,3.2);
\draw[dashed] (-1.5,6) parabola bend (-2.8,3)  (-4,3.2);
\draw[dashed] (-1.5,-6) parabola bend (-2.8,-3)  (-4,-3.2);
\draw[dashed] (0.4,2) .. controls (0,0)  .. (0.4,-2);
\draw[dashed] (4,3.2)--(0.4,2);
\draw[dashed] (4,-3.2)--(0.4,-2);
\draw[dashed] (0.4,2).. controls (-1,2.0)..(-2.8,3);
\draw[dashed] (0.4,-2).. controls (-1,-2.0)..(-2.8,-3);
\draw[thick] (-2.8,3)--(-6.2,5);
\draw[thick,->] (-4.1,3.8)--(-4,3.7);
\draw[thick] (-2.8,-3)--(-6.2,-5);
\draw[thick,->] (-4.0,-3.7)--(-4.1,-3.8);
\draw[thick] (0.4,2)--(0.5,3);
\draw[thick,->] (0.42,2.3)--(0.45,2.4);
\draw[thick] (0.5,3)--(6.2,5);
\draw[thick,->] (4.1,4.25)--(4.2,4.3);
\draw[thick] (0.4,-2)--(0.5,-3);
\draw[thick,->] (0.45,-2.4)--(0.42,-2.3);
\draw[thick] (0.5,-3)--(6.2,-5);
\draw[thick,->] (4.2,-4.3)--(4.1,-4.25);
\draw[dashed] (-6.2,0)--(6.2,0);
\draw[thick] (0.4,2).. controls (-1,2.0)..(-2.8,3);
\draw[thick] (0.4,-2).. controls (-1,-2.0)..(-2.8,-3);
\draw[thick] (0.4,2) .. controls (0,0)  .. (0.4,-2);
\draw[thick,->] (0.2,0.9)--(0.2,1);
\draw[thick,->] (0.2,-1)--(0.2,-0.9);
\draw[thick,->] (-0.9,2.05)--(-1,2.1);
\draw[thick,->] (-0.9,-2.05)--(-1,-2.1);
\draw[ ](0.4,1.8)node[right] {\footnotesize$\beta$} (0.4,-1.8)node[right] {\footnotesize$\overline{\beta}$};
\draw[ ](-2.6,3)node[right] {\footnotesize$\alpha$} (-2.6,-3)node[right] {\footnotesize$\overline{\alpha}$};
\draw[ ](0.3,0)node[below] {\footnotesize$\mu$} (4,-3.2)node[below] {\footnotesize$\overline{E}_{2}$};
\draw[ ](4,3.2)node[above] {\footnotesize$E_{2}$} (-4,-3.2)node[left] {\footnotesize$\overline{E}_{1}$};
\draw[ ](-4,3.2)node[left] {\footnotesize$E_{1}$} ;
\draw[ ](4.1,4.5)node[left] {\footnotesize$1$} (-4.1,4.3)node[right] {\footnotesize$2$};
\draw[ ](-4.1,1.5)node[left] {\footnotesize$3$} (4.1,1.5)node[right] {\footnotesize\footnotesize$4$};
\draw[ ](-1.5,2.5)node[above] {\footnotesize$5$} (0.4,1)node[right] {\footnotesize$6$};
\draw[ ](-5,5)node[above] {$-$} (2,-5)node[above] {$-$} (2,1)node[above] {$-$} (-2,-1)node[below] {$-$} (6,-2)node[below] {$-$}
(-5,-5)node[above] {$+$} (2,5)node[above] {$+$} (2,-1)node[below] {$+$} (-2,1)node[above] {$+$} (6,2)node[above] {$+$};
\end{tikzpicture}
\end{center}
\caption{The jump contour $\Sigma^{(3)}$.}\label{sigma3}
\end{figure}

Our purpose for performing the fourth deformation of the jump contour is to transform the jump matrix across $\gamma_{(\beta,\alpha)}\cup\gamma_{(\overline{\alpha},\overline{\beta})}\cup\gamma_{(\overline{\beta},\beta)}$ to a diagonal or off-diagonal matrix.  Then, the branch cuts $\gamma_{(\beta,\alpha)}\cup\gamma_{(\overline{\alpha},\overline{\beta})}\cup\gamma_{(\overline{\beta},\beta)}$ is separated into eight jump contours in the upper half-plane which are respectively denoted by $(5,6,..., 12)$ and there exist eight jump contours in the lower half-plane, see Figure \ref{sigma45}. The analytic regions enclosed by these jump contours named by $V_{j} (j=1,2,\cdots, 8)$, see Figure \ref{vj}. First, we deformate $\hat{v}_{5}^{(3)}$ and $\hat{v}_{6}^{(3)}$ as
\begin{subequations}
\begin{align}
&\hat{v}_{5}^{(3)}=e^{-itg_{-}\sigma_{3}}\left(
\begin{array}{cc}
1&0\\
i\lambda\hat{f}\delta^{-2}&1
\end{array}
\right)\left(
\begin{array}{cc}
0&-i(\lambda\hat{f})^{-1}\delta^{2}\\
-i\lambda\hat{f}\delta^{-2}&0
\end{array}
\right)\left(
\begin{array}{cc}
1&0\\
i\lambda\hat{f}\delta^{-2}&1
\end{array}
\right)e^{-itg_{+}\sigma_{3}},\nonumber\\
&\hat{v}_{6}^{(3)}=e^{-itg_{-}\sigma_{3}}\left(
\begin{array}{cc}
1&\frac{-i(\lambda\hat{f})^{-1}\delta^{2}}{\hat{a}\hat{\overline{a}}}\\
0&1
\end{array}
\right)\left(
\begin{array}{cc}
\frac{1}{\hat{a}\hat{\overline{a}}}&0\\
0&\hat{a}\hat{\overline{a}}
\end{array}
\right)\left(
\begin{array}{cc}
1&0\\
\frac{\lambda\hat{\varrho}\delta^{-2}}{\hat{a}\hat{\overline{a}}}&1
\end{array}
\right)e^{-itg_{+}\sigma_{3}}.\nonumber
\end{align}
\end{subequations}
We define the function $\hat{m}^{(4)}$ in the upper half-plane by
\begin{equation}
\hat{m}^{(4)}=\hat{m}^{(3)}\begin{cases}
\left(
\begin{array}{cc}
1&0\\
-\frac{e^{2itg}}{\hat{a}\hat{b}\delta^{2}}&1
\end{array}
\right),\quad\quad\quad\quad~  \lambda\in V_{1},\\
\left(\begin{array}{cc}
1&0\\
\frac{e^{2itg}}{\hat{a}\hat{b}\delta^{2}}&1
\end{array}
\right),\quad\quad\quad\quad\quad~ \lambda\in V_{2},\\
\left(
\begin{array}{cc}
1&0\\
-\frac{\lambda\hat{\varrho}}{\hat{a}\hat{\overline{a}}}\delta^{-2}e^{-2itg}&1
\end{array}
\right),\quad~ \lambda\in V_{3},\\
\left(
\begin{array}{cc}
1&-\hat{\overline{\varrho}}\delta^{2}e^{-2itg}\\
0&1
\end{array}
\right),\quad\quad~~ \lambda\in V_{4},\\
I,\quad\quad\quad\quad\quad\quad\quad\quad\quad\quad\quad elsewhere,
\end{cases}\nonumber
\end{equation}
and $\hat{m}^{(4)}(x,t,\lambda)$ admits the RH problem:
\begin{prob}
$\hat{m}^{(4)}(x,t,\lambda)$ satisfies the following properties:
\begin{itemize}
\item  $\hat{m}^{(4)}(x,t,\lambda)$ is analytical in $\mathbb{C}\backslash\Sigma^{(4)}$, where $\Sigma^{(4)}$ see Figure \ref{sigma45}.
\item  $\hat{m}^{(4)}(x,t,\lambda)$ satisfies the jump condition
\begin{equation}
\hat{m}^{(4)}_{+}=\hat{m}^{(4)}_{-}\hat{J}^{(4)},
\end{equation}
where the jump matrix $\hat{J}^{(4)}=\hat{J}_{j}^{(4)}$ in the upper half-plane rewritten as
\begin{subequations}
\begin{align}
&\hat{J}_{1}^{(4)}=\hat{J}_{1}^{(3)},\quad \hat{J}_{2}^{(4)}=\hat{J}_{2}^{(3)},\quad \hat{J}_{3}^{(4)}=\hat{J}_{3}^{(3)},\quad \hat{J}_{4}^{(4)}=\hat{J}_{4}^{(3)},\nonumber\\
&\hat{J}_{5}^{(4)}=\hat{J}_{7}^{(4)}=\left(
\begin{array}{cc}
1&0\\
-\frac{e^{2itg}}{\hat{a}\hat{b}\delta^{2}}&1
\end{array}
\right),\quad \hat{J}_{6}^{(4)}=\left(
\begin{array}{cc}
0&-\hat{a}\hat{b}\delta^{2}e^{-it(g_{+}+g_{-})}\\
\frac{1}{\hat{a}\hat{b}\delta^{2}e^{-it(g_{+}+g_{-})}}&0
\end{array}
\right),\nonumber\\
&\hat{J}_{8}^{(4)}=\left(
\begin{array}{cc}
1&0\\
\frac{1}{\hat{a}^{2}\hat{\overline{a}}\hat{b}}\delta^{-2}e^{2itg}&1
\end{array}
\right),\quad \hat{J}_{9}^{(4)}=\left(
\begin{array}{cc}
1&0\\
-\frac{\hat{\overline{a}}}{\hat{b}}\delta^{-2}e^{2itg}&1
\end{array}
\right),\nonumber\\
&\hat{J}_{10}^{(4)}=\left(
\begin{array}{cc}
1&-\frac{\hat{b}}{\hat{\overline{a}}}\delta^{2}e^{-2itg}\\
0&1
\end{array}
\right),\quad \hat{J}_{11}^{(4)}=\left(
\begin{array}{cc}
\frac{e^{it(g_{+}-g_{-})}}{\hat{a}\hat{\overline{a}}}&0\\
0&\hat{a}\hat{\overline{a}}e^{-it(g_{+}-g_{-})}
\end{array}
\right),\nonumber\\
&\hat{J}_{12}^{(4)}=\left(
\begin{array}{cc}
1&0\\
\frac{\hat{\overline{b}}}{\hat{a}^{2}\hat{\overline{a}}}\delta^{-2}e^{2itg}&1
\end{array}
\right),\nonumber
\end{align}
\end{subequations}
where the subscript of $\hat{J}^{(4)}_{j}$ denote the $j^{th}$ jump contour in Figure \ref{sigma45}.
\item $\hat{m}^{(4)}(x,t,\lambda)$ satisfies the asymptotic behavior
\begin{equation}
\hat{m}^{(4)}(x,t,\lambda)\rightarrow I,\quad \lambda\rightarrow\infty.
\end{equation}
\end{itemize}
\end{prob}

\begin{figure}[H]
\begin{center}
\begin{tikzpicture}[scale=0.45]
\draw[dashed] (-6.2,0)--(6.2,0);
\draw[dashed] (-1.5,6) parabola bend (-2.8,3)  (-4,3.2);
\draw[dashed] (-1.5,-6) parabola bend (-2.8,-3)  (-4,-3.2);
\draw[dashed] (0.4,2) .. controls (0,0)  .. (0.4,-2);
\draw[dashed] (4,3.2)--(0.4,2);
\draw[dashed] (4,-3.2)--(0.4,-2);
\draw[dashed] (0.4,2).. controls (-1,2.0)..(-2.8,3);
\draw[dashed] (0.4,-2).. controls (-1,-2.0)..(-2.8,-3);
\draw[thick] (0,0).. controls (-1,-0.8)..(-2.8,-3);
\draw[thick] (0,0).. controls (-1,0.8)..(-2.8,3);
\draw[thick] (-1,0.8)--(0.4,2);
\draw[thick] (-1,-0.8)--(0.4,-2);
\draw[thick] (0.4,-2)--(0.5,-3);
\draw[thick] (0.4,2)--(0.5,3);
\draw[thick] (0.5,3)..controls(-1.8,3.2)..(-2.8,3);
\draw[thick] (0.5,-3)..controls(-1.8,-3.2)..(-2.8,-3);
\draw[thick] (0,0)..controls(1.8,1)..(0.4,2);
\draw[thick] (0,0)..controls(1.8,-1)..(0.4,-2);
\draw[ ](0.4,2.3)node[right] {\footnotesize$\beta$} (0.4,-2.3)node[right] {\footnotesize$\overline{\beta}$};
\draw[ ](-2.8,3)node[above] {\footnotesize$\alpha$} (-2.8,-3)node[below] {\footnotesize$\overline{\alpha}$};
\draw[ ](0.3,0)node[below] {\footnotesize$\mu$} (4,-3.2)node[below] {\footnotesize$\overline{E}_{2}$};
\draw[ ](4,3.2)node[above] {\footnotesize$E_{2}$} (-4,-3.2)node[left] {\footnotesize$\overline{E}_{1}$};
\draw[ ](-4,3.2)node[left] {\footnotesize$E_{1}$} ;
\draw[ ](-1.3,2.6)node[]{\footnotesize$V_{1}$} (-1.1,1.8)node[]{\footnotesize$V_{2}$};
\draw[ ](-0.2,1)node[]{\footnotesize$V_{3}$} (0.6,1)node[]{\footnotesize$V_{4}$};
\draw[ ](-1.3,-2.6)node[]{\footnotesize$V_{5}$} (-1.1,-1.8)node[]{\footnotesize$V_{6}$};
\draw[ ](-0.2,-1)node[]{\footnotesize$V_{7}$} (0.6,-1)node[]{\footnotesize$V_{8}$};
\end{tikzpicture}
\end{center}
\caption{The distribution of $V_{j}$, $j=1,2,\cdots,8$.}\label{vj}
\end{figure}
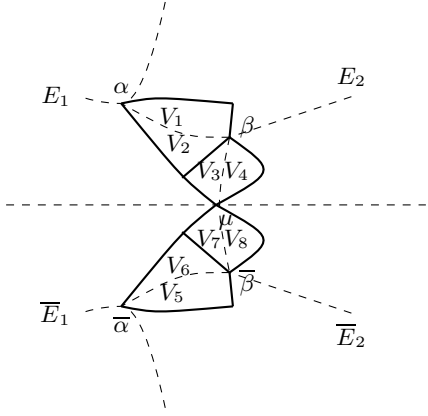

\begin{figure}[H]
\begin{center}
\begin{tikzpicture}[scale=0.45]
\draw[thick,->] (4,-3.2)--(4,-1.6);
\draw[thick,->] (4,-1.6)--(4,1.6);
\draw[thick,-] (4,1.6)--(4,3.2);
\draw[thick,->] (-4,-3.2)--(-4,-1.6);
\draw[thick,->] (-4,-1.6)--(-4,1.6);
\draw[thick,-] (-4,1.6)--(-4,3.2);
\draw[dashed] (-1.5,6) parabola bend (-2.8,3)  (-4,3.2);
\draw[dashed] (-1.5,-6) parabola bend (-2.8,-3)  (-4,-3.2);
\draw[dashed] (0.4,2) .. controls (0,0)  .. (0.4,-2);
\draw[dashed] (4,3.2)--(0.4,2);
\draw[dashed] (4,-3.2)--(0.4,-2);
\draw[dashed] (0.4,2).. controls (-1,2.0)..(-2.8,3);
\draw[dashed] (0.4,-2).. controls (-1,-2.0)..(-2.8,-3);
\draw[thick] (-2.8,3)--(-6.2,5);
\draw[thick,->] (-4.1,3.8)--(-4,3.7);
\draw[thick] (-2.8,-3)--(-6.2,-5);
\draw[thick,->] (-4.0,-3.7)--(-4.1,-3.8);
\draw[thick] (0.4,2)--(0.5,3);
\draw[thick,->] (0.42,2.3)--(0.45,2.4);
\draw[thick] (0.5,3)--(6.2,5);
\draw[thick,->] (4.1,4.25)--(4.2,4.3);
\draw[thick] (0.4,-2)--(0.5,-3);
\draw[thick,->] (0.45,-2.4)--(0.42,-2.3);
\draw[thick] (0.5,-3)--(6.2,-5);
\draw[thick,->] (4.2,-4.3)--(4.1,-4.25);
\draw[dashed] (-6.2,0)--(6.2,0);
\draw[thick] (0.4,2).. controls (-1,2.0)..(-2.8,3);
\draw[thick] (0.4,-2).. controls (-1,-2.0)..(-2.8,-3);
\draw[thick] (0.4,2) .. controls (0,0)  .. (0.4,-2);
\draw[thick,->] (0.2,0.9)--(0.2,1);
\draw[thick,->] (0.2,-1)--(0.2,-0.9);
\draw[thick,->] (-0.9,2.05)--(-1,2.1);
\draw[thick,->] (-0.9,-2.05)--(-1,-2.1);
\draw[thick] (0,0).. controls (-1,-0.8)..(-2.8,-3);
\draw[thick] (0,0).. controls (-1,0.8)..(-2.8,3);
\draw[thick] (-1,0.8)--(0.4,2);
\draw[thick] (-1,-0.8)--(0.4,-2);
\draw[thick] (0.5,3)..controls(-1.8,3.2)..(-2.8,3);
\draw[thick] (0.5,-3)..controls(-1.8,-3.2)..(-2.8,-3);
\draw[thick] (0,0)..controls(1.8,1)..(0.4,2);
\draw[thick] (0,0)..controls(1.8,-1)..(0.4,-2);
\draw[thick,->] (-0.9,3.15)--(-1,3.14);
\draw[thick,->] (-1.5,1.4)--(-1.55,1.45);
\draw[thick,->] (-0.3,1.4)--(-0.25,1.45);
\draw[thick,->] (1.0,1.55)--(0.9,1.6);
\draw[thick,->] (1.0,-0.6)--(0.9,-0.55);
\draw[thick,->] (-0.5,-0.42)--(-0.45,-0.4);
\draw[thick,->] (-2.0,-2)--(-2.05,-2.1);
\draw[thick,->] (-1.3,-3.15)--(-1.35,-3.16);
\draw[ ](0.4,2)node[right] {\footnotesize$\beta$} (0.4,-2)node[right] {\footnotesize$\overline{\beta}$};
\draw[ ](-2.8,3)node[below] {\footnotesize$\alpha$} (-2.8,-3)node[above] {\footnotesize$\overline{\alpha}$};
\draw[ ](0.3,0)node[below] {\footnotesize$\mu$} (4,-3.2)node[below] {\footnotesize$\overline{E}_{2}$};
\draw[ ](4,3.2)node[above] {\footnotesize$E_{2}$} (-4,-3.2)node[left] {\footnotesize$\overline{E}_{1}$};
\draw[ ](-4,3.2)node[left] {\footnotesize$E_{1}$} ;
\draw[ ](4.1,4.5)node[left] {\footnotesize$1$} (-4.1,4.3)node[right] {\footnotesize$2$};
\draw[ ](-4.1,1.5)node[left] {\footnotesize$3$} (4.1,1.5)node[right] {\footnotesize$4$};
\draw[ ](-1.5,3)node[above] {\footnotesize$5$} (-1.5,2.2)node[above] {\footnotesize$6$}  (-1.5,1.5)node[below] {\footnotesize$7$};
\draw[ ](-0.6,1)node[above] {\footnotesize$8$} (0.4,2.6)node[right] {\footnotesize$9$};
\draw[ ](1.2,1)node[right] {\footnotesize$10$} (0.2,0.8)node[left] {\footnotesize$11$}  (-0.5,0.4)node[left] {\footnotesize$12$};
\draw[ ](-5,5)node[above] {$-$} (2,-5)node[above] {$-$} (2,1)node[above] {$-$} (-2,-1)node[below] {$-$} (6,-2)node[below] {$-$}
(-5,-5)node[above] {$+$} (2,5)node[above] {$+$} (2,-1)node[below] {$+$} (-2,1)node[above] {$+$} (6,2)node[above] {$+$};
\end{tikzpicture}
\end{center}
\caption{The jump contours $\Sigma^{(4)}$ and $\Sigma^{(5)}$.}\label{sigma45}
\end{figure}
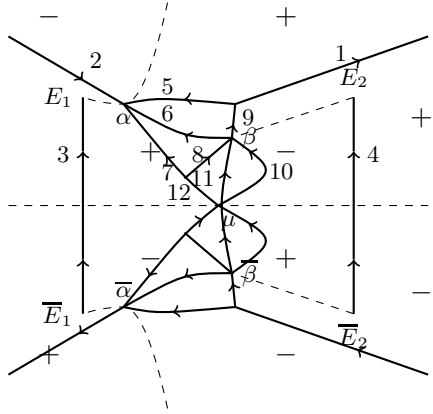

For making the jump matrix across the branch cuts $\gamma_{1}\cup\overline{\gamma}_{1}\cup\gamma_{2}\cup\overline{\gamma}_{2}\cup\gamma_{(\overline{\alpha},\overline{\beta})}\cup\gamma_{(\beta,\alpha)}$ and across $\gamma_{(\overline{\beta},\beta)}$ constant in $\lambda$. We introduce the matrix function $\hat{m}^{(5)}$ by
\begin{equation}
\hat{m}^{(5)}(x,t,\lambda)=e^{-ih(\infty)\sigma_{3}}\hat{m}^{(4)}(x,t,\lambda)e^{ih(\lambda)\sigma_{3}}.
\end{equation}
The function $h(\lambda)$ in $\hat{m}^{(5)}$ defined as
\begin{equation}
h(\lambda)=\frac{\hbar(\lambda)}{2i\pi}\int_{\Sigma^{mod}}\frac{H(s)}{s-\lambda}\mathrm{d}{s},
\end{equation}
where
\begin{equation}
\begin{split}
&H(k)=\begin{cases}
\frac{2\tau_{1}+h_{1}}{\hbar_{+}},\quad\quad\quad~~ \lambda\in\gamma_{1}\cup\overline{\gamma}_{1},\\
\frac{2\tau_{2}+h_{2}}{\hbar_{+}},\quad\quad\quad~~ \lambda\in\gamma_{(\beta,\alpha)}\cup\gamma_{(\overline{\alpha},\overline{\beta})},\\
\frac{2\tau_{3}+h_{3}}{\hbar},\quad\quad\quad~~ \lambda\in\gamma_{(\overline{\beta},\beta)},\\
\frac{h_{4}}{\hbar_{+}},\quad\quad\quad\quad\quad~ \lambda\in\gamma_{2}\cup\overline{\gamma}_{2},
\end{cases}
\end{split}
\end{equation}
with
\begin{equation}
\begin{split}
&h_{1}=-i\ln{(\hat{a}_{+}\hat{a}_{-}\delta^{2}e^{i\phi})},\quad \lambda\in\gamma_{1},\\
&h_{2}=-i\ln{(i\hat{a}_{+}\hat{a}_{-}\delta^{2})},\quad\quad \lambda\in\gamma_{(\beta,\alpha)},\\
&h_{3}=-i\ln{(\hat{a}\hat{\overline{a}})},\quad\quad\quad\quad~~ \lambda\in\gamma_{(\mu,\beta)},\\
&h_{4}=-i\ln{(\mu^{2}\delta^{2})},\quad\quad\quad~~ \lambda\in\gamma_{2}.
\end{split}\nonumber
\end{equation}
\begin{lemma}\label{h}
$h(\lambda)$ admits the following properties:\\
$\blacktriangleright$ $h(\lambda)$ admits the symmetry
\begin{equation}
h(\lambda)=\overline{h(\overline{\lambda})},\quad \lambda\in\hat{\mathbb{C}}\backslash\Sigma^{mod}.
\end{equation}
$\blacktriangleright$ $h(\lambda)$ admits the asymptotic
\begin{equation}
\begin{split}
&h(\lambda)=h(\infty)+\mathcal{O}(\lambda^{-1}),\quad \lambda\rightarrow\infty,\\
&h(\infty)=-\frac{1}{2i\pi}\int_{\Sigma^{mod}}s^{3}H(s)\mathrm{d}s.
\end{split}\label{hinfty}
\end{equation}
$\blacktriangleright$ $e^{ih\sigma_{3}}$ is bounded and analytic for $k\in\hat{\mathbb{C}}\backslash\Sigma^{mod}$.\\
$\blacktriangleright$ $h(\lambda)$ satisfies the jump conditions
\begin{equation}
\begin{split}
&h_{+}+h_{-}=\begin{cases}
2\tau_{1}-i\ln{(\hat{a}_{+}\hat{a}_{-}\delta^{2}e^{i\phi})},\quad\quad\quad~~ \lambda\in\gamma_{1},\\
2\tau_{1}+i\ln{(\hat{\overline{a}}_{+}\hat{\overline{a}}_{-}\delta^{-2}e^{-i\phi})},\quad~~~~ \lambda\in\overline{\gamma}_{1},\\
2\tau_{2}-i\ln{(i\hat{a}\hat{b}\delta^{2})},\quad\quad\quad\quad\quad~~ \lambda\in\gamma_{(\beta,\alpha)},\\
2\tau_{2}+i\ln{(-i\hat{\overline{a}}\hat{\overline{b}}\delta^{-2})},\quad\quad\quad\quad~ \lambda\in\gamma_{(\overline{\alpha},\overline{\beta})},\\
-i\ln{(\mu^{2}\delta^{2})},\quad\quad\quad\quad\quad\quad\quad~~ \lambda\in\gamma_{2}\cup\overline{\gamma}_{2},
\end{cases}\\
&h_{+}-h_{-}=\begin{cases}
2\tau_{3}-i\ln{(\hat{a}\hat{\overline{a}})},\quad\quad\quad\quad\quad\quad\quad \lambda\in\gamma_{(\mu,\beta)},\\
2\tau_{3}+i\ln{(\hat{a}\hat{\overline{a}})},\quad\quad\quad\quad\quad\quad\quad \lambda\in\gamma_{(\overline{\beta},\mu)}.
\end{cases}
\end{split}\nonumber
\end{equation}
\end{lemma}

\begin{prob}
Function $\hat{m}^{(5)}(x,t,\lambda)$ satisfies the following properties:
\begin{itemize}
\item  $\hat{m}^{(5)}(x,t,\lambda)$ is analytical in $\mathbb{C}\backslash\Sigma^{(5)}$, where $\Sigma^{(5)}$ see Figure \ref{sigma45}.
\item  $\hat{m}^{(5)}(x,t,\lambda)$ satisfies the jump condition
\begin{equation}
\hat{m}^{(5)}_{+}=\hat{m}^{(5)}_{-}\hat{J}^{(5)},
\end{equation}
where the jump matrix $\hat{J}^{(5)}=\hat{J}_{j}^{(5)}$ are given by
\begin{subequations}
\begin{align}
&\hat{J}_{1}^{(5)}=\left(
\begin{array}{cc}
1&0\\
\frac{\hat{\overline{b}}}{\hat{a}}\delta^{-2}e^{2itg}e^{2ih}&1
\end{array}
\right),\quad \hat{J}_{2}^{(5)}=\left(
\begin{array}{cc}
1&\hat{a}\hat{b}\delta^{2}e^{-2itg}e^{-2ih}\\
0&1
\end{array}
\right),\nonumber\\
&\hat{J}_{3}^{(5)}=\left(
\begin{array}{cc}
0&ie^{-2i(t\Delta_{1}+\tau_{1})}\\
ie^{2i(t\Delta_{1}+\tau_{1})}&0
\end{array}
\right),\quad \hat{J}_{4}^{(5)}=\left(
\begin{array}{cc}
0&i\\
i&0
\end{array}
\right),\nonumber\\
&\hat{J}_{5}^{(5)}=\hat{J}_{7}^{(5)}=\left(
\begin{array}{cc}
1&0\\
-\frac{e^{2itg}}{\hat{a}\hat{b}\delta^{2}}e^{2ih}&1
\end{array}
\right)\quad \hat{J}_{6}^{(5)}=\left(
\begin{array}{cc}
0&ie^{-2i(t\Delta_{2}+\tau_{2})}\\
ie^{2i(t\Delta_{2}+\tau_{2})}&0
\end{array}
\right),\nonumber\\
&\hat{J}_{8}^{(5)}=\left(
\begin{array}{cc}
1&0\\
\frac{1}{\hat{a}^{2}\hat{\overline{a}}\hat{b}}\delta^{-2}e^{2itg}e^{2ih}&1
\end{array}
\right),\quad \hat{J}_{9}^{(5)}=\left(
\begin{array}{cc}
1&0\\
-\frac{\hat{\overline{a}}}{\hat{b}}\delta^{-2}e^{2itg}e^{2ih}&1
\end{array}
\right),\nonumber\\
&\hat{J}_{10}^{(5)}=\left(
\begin{array}{cc}
1&-\frac{\hat{b}}{\hat{\overline{a}}}\delta^{2}e^{-2itg}e^{-2ih}\\
0&1
\end{array}
\right),\quad \hat{J}_{11}^{(5)}=\left(
\begin{array}{cc}
e^{2i(t\Delta_{3}+\tau_{3})}&0\\
0&e^{-2i(t\Delta_{3}+\tau_{3})}
\end{array}
\right),\nonumber\\
&\hat{J}_{12}^{(5)}=\left(
\begin{array}{cc}
1&0\\
\frac{\hat{\overline{b}}}{\hat{a}^{2}\hat{\overline{a}}}\delta^{-2}e^{2itg}e^{2ih}&1
\end{array}
\right),\nonumber
\end{align}
\end{subequations}
where the subscript of $\hat{J}^{(5)}_{j}$ denote the $j^{th}$ jump contour in Figure \ref{sigma45}.
\item $\hat{m}^{(5)}(x,t,\lambda)$ satisfies the asymptotic behavior
\begin{equation}
\hat{m}^{(5)}(x,t,\lambda)\rightarrow I,\quad \lambda\rightarrow\infty.
\end{equation}
\end{itemize}
\end{prob}
The lower half-plane can be derived by the symmetries.

\section{The Long-Time  Asymptotic}\label{section4}
Let $D_{\varepsilon}(\alpha)$, $D_{\varepsilon}(\overline{\alpha})$, $D_{\varepsilon}(\beta)$, $D_{\varepsilon}(\overline{\beta})$ and $D_{\varepsilon}(\mu)$ as the small disks of $\alpha$, $\overline{\alpha}$, $\beta$, $\overline{\beta}$ and $\mu$. Define $\mathcal{D}=D_{\varepsilon}(\alpha)\cup D_{\varepsilon}(\overline{\alpha})\cup D_{\varepsilon}(\beta)\cup D_{\varepsilon}(\overline{\beta})\cup D_{\varepsilon}(\mu)$. The approximate solution is
\begin{equation}
m^{app}=\begin{cases}
m^{\alpha},\quad\quad \lambda\in D_{\varepsilon}(\alpha),\\
m^{\beta},\quad\quad \lambda\in D_{\varepsilon}(\beta),\\
m^{\overline{\alpha}},\quad\quad \lambda\in D_{\varepsilon}(\overline{\alpha}),\\
m^{\overline{\beta}},\quad\quad \lambda\in D_{\varepsilon}(\overline{\beta}),\\
m^{\mu},\quad\quad \lambda\in D_{\varepsilon}(\mu),\\
m^{mod},\quad elsewhere,
\end{cases}
\end{equation}
and the jump contour $\Sigma^{app}=\Sigma^{mod}\cup\partial\mathcal{D}\cup\mathcal{A}\cup\overline{\mathcal{A}}\cup\mathcal{Z}\cup\overline{\mathcal{Z}}\cup\mathcal{X}$. $\mathcal{Y}$, $\overline{\mathcal{Y}}$, $\mathcal{Z}$, $\overline{\mathcal{Z}}$ and $\mathcal{X}$ are defined in the next subsections.

\subsection{the model problem}
From the fifth deformation of the jump contour, the jump matrix approaches the identity matrix as $t\rightarrow\infty$ on $\Sigma^{(5)}\setminus(\Sigma_{1}\cup\Sigma_{2}\cup\gamma_{(\overline{\alpha},\overline{\beta})}\cup\gamma_{(\beta,\alpha)}\cup\gamma_{(\overline{\beta},\beta)})$. Inspired by this idea, for $t\rightarrow\infty$, the solution of function $\hat{m}^{(5)}$ approaches the solution of function $m^{mod}$. Function $m^{mod}$ admits the RH problem
\begin{equation}
m^{mod}_{+}(x,t,\lambda)=m^{mod}_{-}(x,t,\lambda)v^{mod}(x,t,\lambda),\quad \lambda\in\Sigma^{mod},\label{modRH}
\end{equation}
where define $\Sigma^{mod}=\gamma_{1}\cup\overline{\gamma}_{1}\cup\gamma_{2}\cup\overline{\gamma}_{2}\cup\gamma_{(\beta,\alpha)}\cup\gamma_{(\overline{\alpha},\overline{\beta})}\cup\gamma_{(\overline{\beta},\beta)}$
and the jump matrix
\begin{equation}
v^{mod}=\begin{cases}
\left(
\begin{array}{cc}
0&ie^{-2i(t\Delta_{1}+\tau_{1})}\\
ie^{2i(t\Delta_{1}+\tau_{1})}&0
\end{array}
\right),\quad \lambda\in\gamma_{1}\cup\overline{\gamma}_{1},\\
\left(
\begin{array}{cc}
0&i\\
i&0
\end{array}
\right),\quad\quad\quad\quad\quad\quad\quad\quad\quad\quad~ \lambda\in\gamma_{2}\cup\overline{\gamma}_{2},\\
\left(
\begin{array}{cc}
0&ie^{-2i(t\Delta_{2}+\tau_{2})}\\
ie^{2i(t\Delta_{2}+\tau_{2})}&0
\end{array}
\right),\quad \lambda\in\gamma_{(\beta,\alpha)}\cup\gamma_{(\overline{\alpha},\overline{\beta})},\\
\left(
\begin{array}{cc}
e^{2i(t\Delta_{3}+\tau_{3})}&0\\
0&e^{-2i(t\Delta_{3}+\tau_{3})}
\end{array}
\right),\quad~~ \lambda\in\gamma_{(\overline{\beta},\beta)},\\
\end{cases}\nonumber
\end{equation}
Define a vector valued function $\mathbf{M}(\lambda,u,e)$ by
\begin{equation}
\mathbf{M}(\lambda,u,e)=\left(
\frac{\Theta(\varphi(\lambda^{+})+u+e)}{\Theta(\varphi(\lambda^{+})+e)},\frac{\Theta(\varphi(\lambda^{+})-u-e)}{\Theta(\varphi(\lambda^{+})-e)}
\right),\quad \lambda\in\mathbb{C}\backslash\Sigma^{mod},\quad u,e\in\mathbb{C}^{3},\label{Theta}
\end{equation}
where $\Theta$ is the Riemann theta function
\begin{equation}
\Theta(z)=\sum_{N\in\mathbb{Z}^{3}}e^{2i\pi(\frac{1}{2}N^{T}\tau N+N^{T}z)},\quad z\in\mathbb{C}^{3}.\label{theta}
\end{equation}
Riemann theta function $\Theta$ have the properties for all $z\in\mathbb{C}^{3}$,
\begin{equation}
\Theta(z+e^{(j)})=\Theta(z),\quad \Theta(z+\tau^{(j)})=e^{2i\pi(-z_{j}-\frac{\tau_{ij}}{2})}\Theta(z),\quad \Theta(z)=\Theta(-z),\quad j=1,2,3.
\end{equation}
where $\tau=(\tau_{jl})_{3\times3}$ is a $3\times3$ period matrix. $\tau$ defined by $\tau_{jl}=\int_{b_{j}}\zeta_{l}$ (see \cite{hmf1992} for its detailed informations), where $\int_{a_{i}}\zeta_{j}=\delta_{ij}$. $\zeta_{j}$ is given by $\zeta_{j}=\sum\limits_{l=1}^{3}A_{jl}\hat{\zeta}_{l}$, where $\hat{\zeta}_{l}=\hbar^{-1}\lambda^{l-1}$ and $(A^{-1})_{jl}=\int_{a_{j}}\hat{\zeta}_{l}$. Since $\hbar$ admits $\hbar(\lambda^{+})=-\hbar(\lambda^{-})$, we have the symmetry for $\zeta(\lambda^{+})=-\zeta(\lambda^{-})$. And
\begin{equation}
\int_{a_{j}^{+}}\zeta=\frac{1}{2}\int_{a_{j}}\zeta=\frac{1}{2}e^{(j)},\quad \int_{b_{j}^{+}}\zeta=\frac{1}{2}\int_{b_{j}}\zeta=\frac{1}{2}\tau^{(j)},
\end{equation}
$a_{j}^{+}$ and $b_{j}^{+}$are the restrictions of $a_{j}$ and $b_{j}$ to the upper sheet, $e^{(j)}$ is the $j$th column of the identity matrix $I_{3\times3}$, $\tau^{(j)}$ is also. $\varphi$ in (\ref{Theta})  is Abel map $\varphi:M\rightarrow\mathbb{C}^{3}$ with base point $\overline{E}_{2}$,
\begin{equation}
\varphi(P)=\int_{\overline{E}_{2}}^{P}\zeta, \quad P\in M.\label{varphi}
\end{equation}
And $\varphi$ satisfy the jump conditions
\begin{equation}
\begin{split}
&\varphi_{+}(\lambda^{+})+\varphi_{-}(\lambda^{+})=-\varphi_{+}(\lambda^{-})-\varphi_{-}(\lambda^{-})=\begin{cases}
\tau^{(1)},\quad\quad\quad\quad\quad\quad~~ \lambda\in\gamma_{1}\cup\overline{\gamma}_{1},\\
\tau^{(2)},\quad\quad\quad\quad\quad\quad~~\lambda\in\gamma_{(\beta,\alpha)},\\
\tau^{(2)}+e^{(1)}+e^{(2)},\quad~ \lambda\in\gamma_{(\overline{\alpha},\overline{\beta})},\\
0,\quad\quad\quad\quad\quad\quad\quad\quad \lambda\in\gamma_{2}\cup\overline{\gamma}_{2},
\end{cases}\\
&\varphi_{+}(\lambda^{+})-\varphi_{-}(\lambda^{+})=\tau^{(2)}-\tau^{(3)}+e^{(1)}+e^{(2)},\quad \lambda\in\gamma_{(\overline{\beta},\beta)},
\end{split}\nonumber
\end{equation}
where $\varphi_{+}(\lambda^{\pm})$ and $\varphi_{-}(\lambda^{\pm})$ are the boundary value of $\varphi(l^{\pm})$ for $l\in\mathbb{C}$ approaches $\lambda$ from the right and left of the contour, respectively.

The vector valued function $\mathbf{M}(\lambda,u,e)$ satisfy the jump condition
\begin{equation}
\mathbf{M}_{+}(\lambda,u,e)=\mathbf{M}_{-}(\lambda,u,e)\begin{cases}
\left(
\begin{array}{cc}
0&e^{2i\pi u_{1}}\\
e^{-2i\pi u_{1}}&0
\end{array}
\right),\quad\quad\quad\quad\quad \lambda\in\gamma_{1}\cup\overline{\gamma}_{1},\\
\left(\begin{array}{cc}
0&1\\
1&0
\end{array}
\right),\quad\quad\quad\quad\quad\quad\quad\quad\quad~~ \lambda\in\gamma_{2}\cup\overline{\gamma}_{2},\\
\left(
\begin{array}{cc}
0&e^{2i\pi u_{2}}\\
e^{-2i\pi u_{2}}&0
\end{array}
\right),\quad\quad\quad\quad\quad \lambda\in\gamma_{(\beta,\alpha)}\cup\gamma_{(\overline{\alpha},\overline{\beta})},\\
\left(
\begin{array}{cc}
e^{-2i\pi(u_{2}-u_{3})}&0\\
0&e^{2i\pi(u_{2}-u_{3})}
\end{array}
\right),\quad \lambda\in\gamma_{(\overline{\beta},\beta)}.
\end{cases}\nonumber
\end{equation}
Define a vector valued function $\mathbf{N}(\lambda,u,e)$ by
\begin{equation}
\mathbf{N}(\lambda,u,e)=\frac{1}{2}\left(
\begin{array}{cc}
(\nu_{1}+\nu_{1}^{-1})\mathbf{M}_{1}(\lambda,u,e)&(\nu_{1}-\nu_{1}^{-1})\mathbf{M}_{2}(\lambda,u,e)\\
(\nu_{1}-\nu_{1}^{-1})\mathbf{M}_{1}(\lambda,u,-e)&(\nu_{1}+\nu_{1}^{-1})\mathbf{M}_{2}(\lambda,u,-e)
\end{array}
\right),\quad \lambda\in\mathbb{C}\backslash\Sigma^{mod},\quad u,e\in\mathbb{C}^{3},\nonumber
\end{equation}
where
\begin{equation}
\nu_{1}(\lambda)=\left(
\frac{(\lambda-E_{1})(\lambda-E_{2})(\lambda-\alpha)(\lambda-\beta)}{(\lambda-\overline{E}_{1})(\lambda-\overline{E}_{2})(\lambda-\overline{\alpha})(\lambda-\overline{\beta})}
\right)^{1/4},\quad \lambda\in\mathbb{C}\backslash\Sigma^{mod}.
\end{equation}
For $k\rightarrow\infty$, $\nu_{1}(\lambda)=1+O(\lambda^{-1})$. Let $\hat{\nu}_{1}$ denote the function $M\rightarrow\hat{\mathbb{C}}$ which is given by $\nu_{1}^{2}$ on the upper sheet and by $-\nu_{1}^{2}$ on the lower sheet of $M$. $\hat{\nu}_{1}(\lambda^{\pm})=\pm\nu_{1}^{2}(\lambda)$ for $\lambda\in\Sigma^{mod}$. Then $\hat{\nu}_{1}$ is a meromorphic function on $M$. Noting that $\hat{\nu}_{1}$ has four simple zeros at $E_{1}$, $E_{2}$, $\alpha$ and $\beta$, we see that $\hat{\nu}_{1}$ has degree four. Hence, function $\hat{\nu}_{1}-1$ has four zeros on $M$ counting multiplicity. These zeros are $\infty^{+}, P_{1}, P_{2}, P_{3}\in M$, where $P_{1}P_{2}P_{3}=D$.
$\mathbf{N}(\lambda,u,e)$ satisfies the jump condition
\begin{equation}
\mathbf{N}_{+}(\lambda,u,e)=\mathbf{N}_{-}(\lambda,u,e)\begin{cases}
\left(
\begin{array}{cc}
0&ie^{2i\pi u_{1}}\\
ie^{-2i\pi u_{1}}&0
\end{array}
\right),\quad\quad\quad\quad\quad\quad \lambda\in\gamma_{1}\cup\overline{\gamma}_{1},\\
\left(\begin{array}{cc}
0&i\\
i&0
\end{array}
\right),\quad\quad\quad\quad\quad\quad\quad\quad\quad\quad\quad~ \lambda\in\gamma_{2}\cup\overline{\gamma}_{2},\\
\left(
\begin{array}{cc}
0&ie^{2i\pi u_{2}}\\
ie^{-2i\pi u_{2}}&0
\end{array}
\right),\quad\quad\quad\quad\quad\quad \lambda\in\gamma_{(\beta,\alpha)}\cup\gamma_{(\overline{\alpha},\overline{\beta})},\\
\left(
\begin{array}{cc}
-e^{-2i\pi(u_{2}-u_{3})}&0\\
0&-e^{2i\pi(u_{2}-u_{3})}
\end{array}
\right),\quad \lambda\in\gamma_{(\overline{\beta},\beta)}.
\end{cases}\nonumber
\end{equation}

Define the complex vector $d(\xi)\in\mathbb{C}^{3}$ by $d=\varphi(D)+\mathcal{K}$, where $\varphi{(D)}=\sum\limits_{1}^{3}\varphi(P_{j})$ and $\mathcal{K}=\frac{1}{2}(e^{(1)}+e^{(3)}+\tau^{(1)}+\tau^{(2)}+\tau^{(3)})$. Define the vector function $v(\xi,t)=v(t)=-\frac{1}{\pi}(t\bigtriangleup_{1}+\tau_{1},t\bigtriangleup_{2}+\tau_{2},t(\bigtriangleup_{2}-\bigtriangleup_{3})+\tau_{2}-\tau_{3}+\frac{\pi}{2})$.

 The solution of the model RH problem shown as follows:
\begin{lemma}\label{mmod}
For each choice of the constants $\{\bigtriangleup_{j},\tau_{j}\}_{j=1}^{3}$ and for $t\geq0$, the function $m^{mod}(x,t,\lambda)$ defined by
\begin{equation}
m^{mod}(x,t,\lambda)=\mathbf{N}(\infty,v(t),d)^{-1}\mathbf{N}(\lambda,v(t),d),\quad \lambda\in\hat{\mathbb{C}}\backslash\Sigma^{mod},\label{modN}
\end{equation}
is the unique solution of the RH problem (\ref{modRH}). And this solution satisfies
\begin{equation}
\lim_{\lambda\rightarrow\infty}\lambda m^{mod}_{12}(x,t,\lambda)
=-\frac{i}{2}\mathrm{Im}(E_{1}+E_{2}+\alpha+\beta)\frac{\Theta(\varphi(\infty^{+})+d)\Theta(\varphi(\infty^{+})-v(t)-d)}{\Theta(\varphi(\infty^{+})+v(t)+d)\Theta(\varphi(\infty^{+})-d)}.\nonumber
\end{equation}

\end{lemma}
\begin{proof}
Define a multivalued meromorphic function $\mathbf{P}_{j}(P),~ j=1,2$ by
\begin{equation}
\mathbf{P}_{1}(P)=(\hat{\nu}_{1}(P)-1)\frac{\Theta(\varphi(P)-v(t)-d)}{\Theta(\varphi(P)-d)},\quad \mathbf{P}_{2}(P)=(\hat{\nu}_{1}(P)-1)\frac{\Theta(\varphi(P)+v(t)-d)}{\Theta(\varphi(P)-d)}.\nonumber
\end{equation}
Using the symmetry $\varphi(\lambda^{+})=-\varphi(\lambda^{-})$, then
\begin{equation}
\mathbf{N}(\lambda,v(t),d)=\frac{1}{2\nu(\lambda)}\left(
\begin{array}{cc}
-\mathbf{P}_{1}(\lambda^{-})&\mathbf{P}_{1}(\lambda^{+})\\
\mathbf{P}_{2}(\lambda^{+})&-\mathbf{P}_{2}(\lambda^{-})
\end{array}
\right),\quad \lambda\in\hat{\mathbb{C}}\backslash\Sigma^{mod}.
\end{equation}
The function $\Theta(\varphi(\lambda^{\pm})\pm v(t)-d)$ are bounded on $\hat{\mathbb{C}}\backslash\Sigma^{mod}$. And $\Theta(\varphi(P)-d)$ has zero divisor which is a factor of $\hat{\nu}_{1}-1$. $\mathbf{N}(\lambda,v(t),d)$ is an analytic function which is bounded away from the eight branch points. Thus one can derive
\begin{equation}
|\mathbf{N}(\lambda,v(t),d)|\leq C|\lambda-\lambda_{0}|^{-1/4},\quad \lambda\in\hat{\mathbb{C}}\backslash\Sigma^{mod},
\end{equation}
where $\lambda_{0}$ is one of the eight branch points.
For $\lambda\rightarrow\infty$,
\begin{equation}
\lim_{\lambda\rightarrow\infty}\mathbf{N}(\lambda,v(t),d)=\mathbf{N}(\infty,v(t),d)=\left(\begin{array}{cc}
\mathbf{M}_{1}(\infty,v(t),d)&0\\
0&\mathbf{M}_{2}(\infty,v(t),-d)
\end{array}
\right),\nonumber
\end{equation}
where
\begin{equation}
\mathbf{M}_{1}(\infty,v(t),d)=\frac{\Theta(\varphi(\infty^{+})+v(t)+d)}{\Theta(\varphi(\infty^{+})+d)},\quad \mathbf{M}_{2}(\infty,v(t),d)=\frac{\Theta(\varphi(\infty^{+})-v(t)+d)}{\Theta(\varphi(\infty^{+})+d)}.\nonumber
\end{equation}
The values $\Theta(\varphi(\infty^{+})+v(t)+d)$ and $\Theta(\varphi(\infty^{+})+d)$ are finite and nonzero\cite{hmf1992,jdf1973}. These imply that $\mathbf{N}(\infty,v(t),d)$ is invertible. Equation (\ref{modN}) is successfully proved:
\begin{equation}
\begin{split}
&\lim_{\lambda\rightarrow\infty}\lambda m^{mod}_{12}(x,t,\lambda)\\
&=\frac{1}{\mathbf{M}_{1}(\infty,v(t),d)}\lim_{\lambda\rightarrow\infty}\lambda\mathbf{N}_{12}(\lambda,v(t),d)\\
&=\frac{\Theta(\varphi(\infty^{+})+d)}{\Theta(\varphi(\infty^{+})+v(t)+d)}\lim_{\lambda\rightarrow\infty}\frac{\lambda(\nu_{1}(\lambda)-\nu_{1}^{-1}(\lambda))}{2}\frac{\Theta(\varphi(\infty^{+})-v(t)-d)}{\Theta(\varphi(\infty^{+})-d)}\\
&=-\frac{i}{2}\mathrm{Im}(E_{1}+E_{2}+\alpha+\beta)\frac{\Theta(\varphi(\infty^{+})+d)\Theta(\varphi(\infty^{+})-v(t)-d)}{\Theta(\varphi(\infty^{+})+v(t)+d)\Theta(\varphi(\infty^{+})-d)}.
\end{split}\nonumber
\end{equation}\end{proof}

\subsection{the local model near $\alpha$, $\beta$ and $\mu$}
The jump matrix $\hat{m}^{(5)}$ of the fifth deformation with the property that $\hat{v}^{(5)}-I$ decays to zero for $\lambda\in\Sigma^{(5)}\setminus\Sigma^{mod}$ as $t\rightarrow\infty$. But this decay is not uniform decay as $\lambda$ approaches $\Sigma^{mod}$. So, for the parts of $\Sigma^{(5)}$ that lie near $\Sigma^{mod}$, it is necessary to introduce the local solutions which are better approximations of $\hat{m}^{(5)}$ than $m^{mod}$. These local approximations help us derive the approximate error estimates.

$\mathbf{\underline{local~model~near~\alpha}}$:
We define a function $m^{(\alpha0)}(x,t,k)$ for $k$ near $\alpha$ by
\begin{equation}
m^{(\alpha0)}(x,t,\lambda)=\hat{m}^{(5)}(x,t,\lambda)e^{-i(\frac{i}{2}\ln{(-\delta(\lambda))^{2}\hat{a}(\lambda)\hat{b}(\lambda)}+tg(\alpha)+h(\lambda))\sigma_{3}},\quad \lambda\in D_{\varepsilon}(\alpha)\backslash\Sigma^{(5)}.\nonumber
\end{equation}
Lemma \ref{delta} and Lemma \ref{h} imply that the exponential of $m^{(\alpha0)}$ is bounded and analytic for $\lambda\in D_{\varepsilon}(\alpha)\backslash\Sigma^{(5)}$.
The function $m^{(\alpha0)}$ satisfies the following jump condition
\begin{equation}
m_{+}^{(\alpha0)}(x,t,\lambda)=m_{-}^{(\alpha0)}(x,t,\lambda)v^{(\alpha0)}(x,t,\lambda),\quad \lambda\in\Sigma^{(5)}\cap D_{\varepsilon}(\alpha),\label{alpha0}
\end{equation}
where
\begin{equation}
v^{(\alpha0)}=\begin{cases}
\left(
\begin{array}{cc}
1&-e^{-2itg_{\alpha}}\\
0&1
\end{array}
\right),\quad \lambda\in\mathcal{A}_{1},\\
\left(
\begin{array}{cc}
1&0\\
e^{2itg_{\alpha}}&1
\end{array}
\right),\quad\quad~ \lambda\in\mathcal{A}_{2}\cup\mathcal{A}_{4},\\
\left(
\begin{array}{cc}
0&1\\
-1&0
\end{array}
\right),\quad\quad\quad~ \lambda\in\mathcal{A}_{3},
\end{cases}\nonumber
\end{equation}
with
$g_{\alpha}(\lambda)=g(\lambda)-g(\alpha)$ for $\lambda\in D_{\varepsilon}(\alpha)\backslash\gamma_{(\beta,\alpha)}$.
And $\mathcal{A}_{j}=\overline{\mathcal{S}}_{j-1}\cap\overline{\mathcal{S}}_{j}$, $\overline{\mathcal{S}}_{0}=\overline{\mathcal{S}}_{4}$. Let $\mathcal{A}=\cup\mathcal{A}_{j}$, $\overline{\mathcal{A}}$ is the conjugate of $\mathcal{A}$. The jump contour of RH problem (\ref{alpha0}) and $\mathcal{S}_{j}$ please see Figure \ref{alpha}.

For relating this RH problem to the Airy function of Appendix \ref{airy}, we introduce $\zeta(\lambda)=\left(\frac{3itg_{\alpha}(\lambda)}{2}\right)^{2/3}$. We define a parametrix $m^{\alpha}(x,t,\lambda)$ for $\hat{m}^{(5)}$ near $\alpha$ by
\begin{equation}
m^{\alpha}(x,t,\lambda)=Y_{\alpha}(x,t,\lambda)m^{Ai}(\zeta(\lambda))e^{-i(\frac{i}{2}\ln{(-\delta(\lambda))^{2}\hat{a}(\lambda)\hat{b}(\lambda)}+tg(\alpha)+h(\lambda))\sigma_{3}},\quad \lambda\in D_{\varepsilon}(\alpha)\backslash\Sigma^{(5)},\nonumber
\end{equation}
where
\begin{equation}
Y_{\alpha}(x,t,\lambda)=m^{mod}e^{-i(\frac{i}{2}\ln{(-\delta(\lambda))^{2}\hat{a}(\lambda)\hat{b}(\lambda)}+tg(\alpha)+h(\lambda))\sigma_{3}}(m_{as,N}^{Ai}\zeta(\lambda))^{-1},\quad N\geq0.\label{mode}
\end{equation}
Function $m_{as,N}^{Ai}\zeta(\lambda)$ is analytic near $\alpha$ for the jump contour $\mathcal{Y}_{3}$ satisfies (\ref{Ai}). And the first second terms in (\ref{mode}), i.e. $m^{mod}e^{-i(\frac{i}{2}\ln{(-\delta(\lambda))^{2}\hat{a}(\lambda)\hat{b}(\lambda)}+tg(\alpha)+h(\lambda))\sigma_{3}}$ also satisfies (\ref{Ai}) on $\mathcal{Y}_{3}$. According the Lemma \ref{Airy} in the Appendix \ref{airy}, we have
\begin{equation}
m^{\alpha}(\lambda)(m^{mod}(\lambda))^{-1}=I+O(t^{-N-1}),\quad t\rightarrow\infty,\quad \lambda\in\partial D_{\varepsilon}(\alpha),\quad N\geq0.\label{malpha}
\end{equation}

\begin{figure}[H]
\begin{center}
\begin{tikzpicture}[scale=0.45]
\draw[thick] (0,0) circle (4cm);
\draw[thick] (-4,0) .. controls (-2,0.3)  .. (0,0);
\draw[thick] (4,0) .. controls (2,-0.3)  .. (0,0);
\draw[thick] (2,3.5) .. controls (0.3,1)  .. (0,0);
\draw[thick] (-2,-3.5) .. controls (0.3,-1.4)  .. (0,0);
\draw[thick,->] (-2.1,0.25)-- (-2,0.25);
\draw[thick,->] (2.1,-0.25)-- (2,-0.25);
\draw[thick,->] (0.7,1.5)--(0.69,1.48);
\draw[thick,->] (-0.3,-1.84)--(-0.28,-1.79);
\draw[ ](-1.5,-1.5)node[] {$\mathcal{S}_{1}$} (0.6,-1.5)node[right] {$\mathcal{S}_{2}$};
\draw[ ](1,1.5)node[right] {$\mathcal{S}_{3}$} (-1.5,1.5)node[right] {$\mathcal{S}_{4}$};
\draw[ ](-4,0)node[left] {$\mathcal{A}_{1}$} (-2,-3.5)node[below] {$\mathcal{A}_{2}$};
\draw[ ](4,0)node[right] {$\mathcal{A}_{3}$} (2,3.5)node[above] {$\mathcal{A}_{4}$};
\draw[ ](0,0.3)node[right] {$\alpha$};
\end{tikzpicture}
\end{center}
\caption{The jump contour $\mathcal{A}_{j}$ and $\mathcal{S}_{j}$.}\label{alpha}
\end{figure}
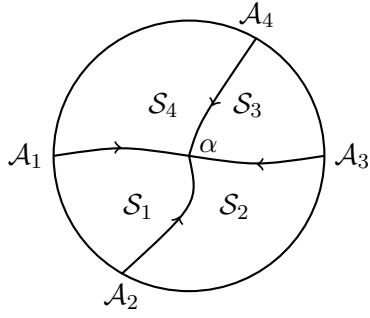

$\mathbf{\underline{local~model~near~\beta}}$: Define $m^{(\beta0)}(x,t,\lambda)$ for $\lambda$ near $\beta$ as
\begin{equation}
m^{(\beta0)}(x,t,\lambda)=\hat{m}^{(5)}(x,t,\lambda)e^{-i(\frac{i}{2}\ln(\delta(\lambda))^{2}\hat{a}(\lambda)\hat{b}(\lambda)+h(\lambda))\sigma_{3}},\quad \lambda\in D_{\varepsilon}(\beta)\backslash\Sigma^{(5)}.\nonumber
\end{equation}
Function $m^{(\beta0)}(x,t,\lambda)$ satisfies the following jump condition
\begin{equation}
m_{+}^{(\beta0)}(x,t,\lambda)=m_{-}^{(\beta0)}(x,t,\lambda)v^{(\beta0)}(x,t,\lambda),\quad \lambda\in\Sigma^{(5)}\cap D_{\varepsilon}(\beta),
\end{equation}
where
\begin{equation}
v^{(\beta0)}=\begin{cases}
\left(
\begin{array}{cc}
1&-\frac{1}{\hat{a}\hat{\overline{a}}}e^{-2itg}\\
0&1
\end{array}
\right),\quad\quad\quad\quad\quad\quad~~ \lambda\in\mathcal{Z}_{1},\\
\left(
\begin{array}{cc}
1&0\\
\hat{a}\hat{\overline{a}}e^{2itg}&1
\end{array}
\right),\quad\quad\quad\quad\quad\quad\quad\quad \lambda\in\mathcal{Z}_{2},\\
\left(
\begin{array}{cc}
0&e^{-it(g_{+}-g_{-})}\\
-e^{it(g_{+}-g_{-})}&0
\end{array}
\right),\quad~~ \lambda\in\mathcal{Z}_{3},\\
\left(
\begin{array}{cc}
1&0\\
\frac{1}{\hat{a}\hat{\overline{a}}}e^{2itg}&1
\end{array}
\right),\quad\quad\quad\quad\quad\quad\quad\quad \lambda\in\mathcal{Z}_{4},\\
\left(
\begin{array}{cc}
\frac{e^{it(g_{+}-g_{-})}}{a\overline{a}}&0\\
0&\hat{a}\hat{\overline{a}}e^{-it(g_{+}-g_{-})}
\end{array}
\right),\quad~ \lambda\in\mathcal{Z}_{5}.\\
\end{cases}\nonumber
\end{equation}
Let $\mathcal{Z}=\cup\mathcal{Z}_{j}$, $\overline{\mathcal{Z}}$ is the conjugate of $\mathcal{Z}$. The jump contour $\mathcal{Z}_{j}$ and $\mathcal{Z}_{j}=\overline{\mathcal{T}}_{j-1}\cap\overline{\mathcal{T}}_{j}$ please see Figure \ref{beta}.
Define $g_{\beta}(\lambda)$ as
\begin{equation}
g_{\beta}(\lambda)=\int_{\beta}^{\lambda}\mathrm{d}g=\begin{cases}
g(\lambda)-g_{-}(\beta),\quad \lambda\in\mathcal{T}_{1}\cup\mathcal{T}_{2}\cup\mathcal{T}_{5},\\
g(\lambda)-g_{+}(\beta),\quad \lambda\in\mathcal{T}_{3}\cup\mathcal{T}_{4}.
\end{cases}
\end{equation}
Introduce a transformation
\begin{equation}
m^{(\beta1)}(x,t,\lambda)=m^{(\beta0)}(x,t,\lambda)A(\lambda),\quad \lambda\in D_{\varepsilon}(\beta)\backslash\Sigma^{(5)},
\end{equation}
where
\begin{equation}
A(\lambda)=\begin{cases}
(\hat{a}\hat{\overline{a}})^{-\frac{\sigma_{3}}{2}}e^{-itg_{-}(\beta)\sigma_{3}},\quad \lambda\in\mathcal{T}_{1}\cup\mathcal{T}_{2}\cup\mathcal{T}_{5},\\
(\hat{a}\hat{\overline{a}})^{\frac{\sigma_{3}}{2}}e^{-itg_{+}(\beta)\sigma_{3}},\quad~~ \lambda\in\mathcal{T}_{3}\cup\mathcal{T}_{4}.
\end{cases}
\end{equation}
We can derive that $m^{(\beta1)}(x,t,\lambda)$ satisfies a jump condition
\begin{equation}
m^{(\beta1)}_{+}(x,t,\lambda)=m^{(\beta1)}_{-}(x,t,\lambda)v^{(\beta1)}(x,t,\lambda),
\end{equation}
where the jump matrix
\begin{equation}
v^{(\beta1)}=\begin{cases}
\left(
\begin{array}{cc}
1&-e^{-2itg_{\beta}}\\
0&1
\end{array}
\right),\quad \lambda\in\mathcal{Z}_{1},\\
\left(
\begin{array}{cc}
1&0\\
e^{2itg_{\beta}}&1
\end{array}
\right),\quad\quad~ \lambda\in\mathcal{Z}_{2}\cup\mathcal{Z}_{4},\\
\left(
\begin{array}{cc}
0&1\\
-1&0
\end{array}
\right),\quad\quad\quad~ \lambda\in\mathcal{Z}_{3},\\
I,\quad\quad\quad\quad\quad\quad\quad\quad~ \lambda\in\mathcal{Z}_{5}.\\
\end{cases}\nonumber
\end{equation}
Similar with local near $\alpha$, we need to transform this model to Airy function of Appendix \ref{airy}. Let $\zeta(\lambda)=(\frac{3it}{2}g_{\beta}(\lambda))^{2/3}$. Define $m^{\beta}(x,t,\lambda)$ by
\begin{equation}
m^{\beta}(x,t,\lambda)=Y_{\beta}(x,t,\lambda)m^{Ai}(\zeta(\lambda))A^{-1}(\lambda)e^{i(\frac{i}{2}\ln(\hat{a}\hat{b}\delta^{2})+tg(\beta)+h(\lambda))\sigma_{3}},\quad \lambda\in D_{\varepsilon}(\beta)\backslash\Sigma^{(5)},\nonumber
\end{equation}
where
\begin{equation}
Y_{\beta}(x,t,\lambda)=m^{mod}e^{-i(\frac{i}{2}\ln(\hat{a}\hat{b}\delta^{2})+tg(\beta)+h(\lambda))\sigma_{3}}A(\lambda)(m_{as,N}^{Ai}(\zeta(\lambda)))^{-1}.\nonumber
\end{equation}
According the asymptotic function (\ref{Ai}) of Lemma \ref{Airy} in the Appendix \ref{airy}, one can derive
\begin{equation}
m^{\beta}(\lambda)(m^{mod}(\lambda))^{-1}=I+O(t^{-N-1}),\quad t\rightarrow\infty,\quad \lambda\in\partial D_{\varepsilon}(\beta),\quad N\geq0.\label{mbeta}
\end{equation}

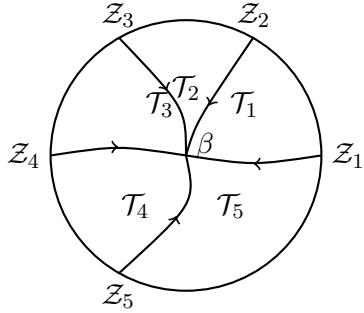
\begin{figure}[H]
\begin{center}
\begin{tikzpicture}[scale=0.45]
\draw[thick] (0,0) circle (4cm);
\draw[thick] (-4,0) .. controls (-2,0.3)  .. (0,0);
\draw[thick] (4,0) .. controls (2,-0.3)  .. (0,0);
\draw[thick] (2,3.5) .. controls (0.3,1)  .. (0,0);
\draw[thick] (-2,-3.5) .. controls (0.3,-1.4)  .. (0,0);
\draw[thick] (-2,3.5) .. controls (0,1.4)  .. (0,0);
\draw[thick,->] (-2.1,0.25)-- (-2,0.25);
\draw[thick,->] (2.1,-0.25)-- (2,-0.25);
\draw[thick,->] (0.7,1.5)--(0.69,1.48);
\draw[thick,->] (-0.3,-1.84)--(-0.28,-1.79);
\draw[thick,->] (-0.6,2)--(-0.56,1.92);
\draw[ ](-1.5,-1.5)node[] {$\mathcal{T}_{4}$} (0,2)node[] {$\mathcal{T}_{2}$} (0.6,-1.5)node[right] {$\mathcal{T}_{5}$};
\draw[ ](1,1.5)node[right] {$\mathcal{T}_{1}$} (-1.5,1.5)node[right] {$\mathcal{T}_{3}$};
\draw[ ](-4,0)node[left] {$\mathcal{Z}_{4}$} (-2,-3.5)node[below] {$\mathcal{Z}_{5}$};
\draw[ ](4,0)node[right] {$\mathcal{Z}_{1}$} (2,3.5)node[above] {$\mathcal{Z}_{2}$}(-2,3.5)node[above] {$\mathcal{Z}_{3}$};
\draw[ ](0,0.3)node[right] {$\beta$};
\end{tikzpicture}
\end{center}
\caption{The jump contour $\mathcal{Z}_{j}$ and $\mathcal{T}_{j}$}\label{beta}
\end{figure}
$\mathbf{\underline{local~model~near~\mu}}$: Define $m^{(\mu0)}(x,t,\lambda)$ for $\lambda$ near $\mu$ by
\begin{equation}
m^{(\mu0)}(x,t,\lambda)=\hat{m}^{(5)}(x,t,\lambda)e^{-ih\sigma_{3}}B(\lambda),\quad \lambda\in D_{\varepsilon}(\mu)\backslash\Sigma^{(5)},
\end{equation}
where
\begin{equation}
B(\lambda)=\begin{cases}
e^{-itg_{-}(\mu)\sigma_{3}},\quad \lambda\in\mathcal{R}_{1}\cup\mathcal{R}_{2}\cup\mathcal{R}_{6},\\
e^{-itg_{+}(\mu)\sigma_{3}},\quad \lambda\in\mathcal{R}_{3}\cup\mathcal{R}_{4}\cup\mathcal{R}_{5},
\end{cases}
\end{equation}
and function $g_{\mu}(\lambda)$ for $\lambda$ near $\mu$, i.e. $\lambda\in D_{\varepsilon}(\mu)$, defined by
\begin{equation}
g_{\mu}(\lambda)=\int_{\mu}^{\lambda}\mathrm{d}g=\begin{cases}
g(\lambda)-g_{-}(\mu),\quad \lambda\in\mathcal{R}_{1}\cup\mathcal{R}_{2}\cup\mathcal{R}_{6},\\
g(\lambda)-g_{+}(\mu),\quad \lambda\in\mathcal{R}_{3}\cup\mathcal{R}_{4}\cup\mathcal{R}_{5}.
\end{cases}
\end{equation}
Function $m^{(\mu0)}(x,t,\lambda)$ satisfies jump condition
\begin{equation}
m^{(\mu0)}_{+}(x,t,\lambda)=m^{(\mu0)}_{-}(x,t,\lambda)v^{(\mu0)}(x,t,\lambda),
\end{equation}
where the jump matrix
\begin{equation}
v^{(\mu0)}=\begin{cases}
\left(
\begin{array}{cc}
1&0\\
\frac{\hat{\overline{b}}}{\hat{a}}\delta^{-2}e^{2itg_{\mu}}&1
\end{array}
\right),\quad\quad~ \lambda\in\mathcal{X}_{1},\\
\left(
\begin{array}{cc}
1&-\frac{\hat{b}}{\hat{\overline{a}}}\delta^{2}e^{-2itg_{\mu}}\\
0&1
\end{array}
\right),\quad\quad \lambda\in\mathcal{X}_{2},\\
\left(
\begin{array}{cc}
\frac{1}{a\overline{a}}&0\\
0&a\overline{a}
\end{array}
\right),\quad\quad\quad\quad\quad~ \lambda\in\mathcal{X}_{3},\\
\left(
\begin{array}{cc}
1&0\\
\frac{\hat{\overline{b}}}{\hat{a}^{2}\hat{\overline{a}}}\delta^{-2}e^{2itg_{\mu}}&1
\end{array}
\right),\quad~~ \lambda\in\mathcal{X}_{4},\\
\left(
\begin{array}{cc}
1&-\frac{\hat{b}}{\hat{a}\hat{\overline{a}}^{2}}\delta^{2}e^{-2itg_{\mu}}\\
0&1
\end{array}
\right),\quad~ \lambda\in\mathcal{X}_{5},\\
\left(
\begin{array}{cc}
\hat{a}\hat{\overline{a}}&0\\
0&\frac{1}{\hat{a}\hat{\overline{a}}}
\end{array}
\right),\quad\quad\quad\quad\quad~ \lambda\in\mathcal{X}_{6}.\\
\end{cases}\nonumber
\end{equation}
The jump contour please see Figure \ref{mu}.
For eliminating the jump across $\gamma_{(\overline{\beta},\beta)}$, we introduce a function
\begin{equation}
\tilde{\delta}(\lambda)=\exp\left[\frac{i}{2\pi}\int_{\gamma_{(\mu,\beta)}}\frac{\ln{\hat{a}(s)\hat{\overline{a}}(s)}}{s-\lambda}\mathrm{d}s\right]
\exp\left[\frac{-i}{2\pi}\int_{\gamma_{(\overline{\beta},\mu})}\frac{\ln{\hat{a}(s)\hat{\overline{a}}(s)}}{s-\lambda}\mathrm{d}s\right],\quad \lambda\in\mathbb{C}\backslash\gamma_{(\overline{\beta},\beta)}.\nonumber
\end{equation}

\begin{lemma}
The function $\tilde{\delta}(\lambda)$ admits the following properties:\\
$\blacktriangleright$ $\tilde{\delta}(\lambda)$ and $\tilde{\delta}^{-1}(\lambda)$ are bounded and analytic functions for $\lambda\in D_{\varepsilon}(\mu)\backslash\gamma_{(\overline{\beta},\beta)}$.\\
$\blacktriangleright$ $\tilde{\delta}(\lambda)$ admits the symmetry
\begin{equation}
\tilde{\delta}=(\overline{\tilde{\delta}})^{-1},\quad \lambda\in\mathbb{C}\backslash\gamma_{(\overline{\beta},\beta)}.
\end{equation}
$\blacktriangleright$ $\tilde{\delta}(\lambda)$ admits the following jump condition
\begin{equation}
\tilde{\delta}_{+}=\tilde{\delta}_{-}\begin{cases}
\frac{1}{\hat{a}\hat{\overline{a}}},\quad  \lambda\in\gamma_{(\mu,\beta)},\\
\hat{a}\hat{\overline{a}},\quad  \lambda\in\gamma_{(\overline{\beta},\mu)}.\\
\end{cases}
\end{equation}
$\blacktriangleright$ $\tilde{\delta}(\lambda)$ can be rewritten as
\begin{equation}
\tilde{\delta}(\lambda)=exp\left[i\nu_{2}[\ln_{\beta}{(\lambda-\mu)}+\ln_{\overline{\beta}}{(\lambda-\mu)}]+\tilde{\chi}(\lambda)\right],
\end{equation}
where $\nu_{2}=\frac{\ln{(1+|q|^{2}})}{2\pi}$ and
\begin{equation}
\begin{split}
\tilde{\chi}=&\frac{1}{2i\pi}L_{\beta}(\beta,\lambda)\ln{(1+\varrho(\beta)\overline{\varrho}(\beta))}+\frac{1}{2i\pi}L_{\overline{\beta}}(\overline{\beta},\lambda)\ln{(1+\varrho(\overline{\beta})\overline{\varrho}(\overline{\beta}))}\\
&+\frac{1}{2i\pi}\int_{\gamma_{\mu,\beta}}L_{\beta}(s,\lambda)\ln{(1+\varrho(s)\overline{\varrho}(s))}-\frac{1}{2i\pi}\int_{\gamma_{\overline{\beta},\mu}}L_{\beta}(s,\lambda)\ln{(1+\varrho(s)\overline{\varrho}(s))}.
\end{split}\nonumber
\end{equation}
Some symbols denote respectively
\begin{equation}
\begin{split}
&\ln_{\beta}(\lambda-\mu)=\ln(\lambda-\mu),\quad \lambda\in D_{\varepsilon}\backslash\gamma_{(\mu,\beta)},\\
&\ln_{\overline{\beta}}(\lambda-\mu)=\ln(\lambda-\mu),\quad \lambda\in D_{\varepsilon}\backslash\gamma_{(\overline{\beta},\mu)},\\
&\mathrm{L}_{\beta}(s,\lambda)=\ln(\lambda-s),\quad s\in\gamma_{(\mu,\beta)},\quad \lambda\in D_{\varepsilon}\backslash\gamma_{(\mu,\beta)},\\
&\mathrm{L}_{\overline{\beta}}(s,\lambda)=\ln(\lambda-s),\quad s\in\gamma_{(\overline{\beta},\mu)},\quad \lambda\in D_{\varepsilon}\backslash\gamma_{(\overline{\beta},\mu)}.
\end{split}\nonumber
\end{equation}
\end{lemma}
$\mathbf{Remark}$:
Function $\delta$ in Lemma \ref{delta} with a new version
\begin{equation}
\delta(\lambda)=exp\left[-i\nu_{2}\ln{(\lambda-\mu)}+\chi(\lambda)\right],\quad \lambda\in\mathbb{C}\backslash(-\infty,\mu].
\end{equation}
where $\nu_{2}=\frac{\ln{(1+|q|^{2})}}{2}$ and
\begin{equation}
\begin{split}
\chi(\lambda)=&\frac{1}{2i\pi}\ln(\lambda+1)\ln(\frac{1+|\varrho_{+}(-1)|^{2}}{1+|\varrho_{-}(-1)|^{2}})-\\
&\frac{1}{2i\pi}\left(\int_{-\infty}^{-1}+\int_{-1}^{\mu}\right)\ln(\lambda-s)\ln(1+|\varrho(s)|^{2}),\quad \lambda\in\mathbb{C}\backslash(-\infty,\mu].
\end{split}\nonumber
\end{equation}
$\varrho_{+}(-1)$ and $\varrho_{-}(-1)$ denote the values of $\varrho(\lambda)$ on the left and right sides of $\gamma_{1}\cup\overline{\gamma}_{1}$.
Moreover, functions $\tilde{\delta}$ and $\delta$ satisfy the following relations:
\begin{lemma}Let $\delta_{2}=\delta\tilde{\delta}$, thus
\begin{equation}
\delta_{2}(\lambda)=p(z(\lambda))\delta_{0}(\lambda)\delta_{1}(\lambda), \quad \lambda\in D_{\varepsilon}(\mu)\backslash((-\infty,\mu]\cup\gamma_{(\overline{\beta},\beta)}),
\end{equation}
where
\begin{equation}
\begin{split}
&p(z)=exp\left[-i\nu_{2}[\ln_{-\frac{\pi}{2}}z-\ln{z}-\ln_{0}z]\right],\quad z\in\mathbb{C}\backslash(\mathbb{R}\cup i\mathbb{R}_{-}),\\
&\delta_{0}(t)=e^{\frac{\pi\nu_{2}}{2}}t^{-\frac{i\nu_{2}}{2}}exp^{-i\nu_{2}\ln{\psi_{\mu}(\mu)}}e^{\chi(\mu)+\tilde{\chi}(\mu)},\quad t>0,\\
&\delta_{1}(\lambda)=e^{-i\nu_{2}\ln{\frac{\psi_{\mu}(\lambda)}{\psi_{\mu}(\mu)}}}e^{\chi(\lambda)-\chi(\mu)+\tilde{\chi}(\lambda)-\tilde{\chi}(\mu)},\quad \lambda\in D_{\varepsilon}(\mu).
\end{split}\nonumber
\end{equation}
\end{lemma}

Define $m^{(\mu1)}(x,t,\lambda)$ by
\begin{equation}
m^{(\mu1)}(x,t,\lambda)=m^{(\mu0)}(x,t,\lambda)\tilde{\delta}(\lambda)^{-\sigma_{3}},\quad \lambda\in D_{\varepsilon}(\mu).
\end{equation}
$m^{(\mu1)}(x,t,\lambda)$ satisfies a jump condition
\begin{equation}
m^{(\mu1)}_{+}(x,t,\lambda)=m^{(\mu1)}_{-}(x,t,\lambda)v^{(\mu1)}(x,t,\lambda),
\end{equation}
the jump matrix is given by
\begin{equation}
v^{(\mu1)}=\begin{cases}
\left(
\begin{array}{cc}
1&0\\
\frac{\hat{\overline{b}}}{\hat{a}}\delta_{2}^{-2}e^{2itg_{\mu}}&1
\end{array}
\right),\quad\quad~ \lambda\in\mathcal{X}_{1},\\
\left(
\begin{array}{cc}
1&-\frac{\hat{b}}{\hat{\overline{a}}}\delta_{2}^{2}e^{-2itg_{\mu}}\\
0&1
\end{array}
\right),\quad~~ \lambda\in\mathcal{X}_{2},\\
\left(
\begin{array}{cc}
1&0\\
\frac{\hat{\overline{b}}}{\hat{a}^{2}\hat{\overline{a}}}\delta_{2}^{-2}e^{2itg_{\mu}}&1
\end{array}
\right),\quad~~ \lambda\in\mathcal{X}_{4},\\
\left(
\begin{array}{cc}
1&-\frac{\hat{b}}{\hat{a}\hat{\overline{a}}^{2}}\delta_{2}^{2}e^{-2itg_{\mu}}\\
0&1
\end{array}
\right),\quad \lambda\in\mathcal{X}_{5},\\
I,\quad\quad\quad\quad\quad\quad\quad\quad\quad\quad~~ \lambda\in\mathcal{X}_{3}\cup\mathcal{X}_{6}.
\end{cases}\nonumber
\end{equation}
Let $\mathcal{X}=\mathcal{X}_{1}\cup\mathcal{X}_{2}\cup\mathcal{X}_{4}\cup\mathcal{X}_{5}$.
For relating $m^{(\mu1)}$ to the solution of the parabolic cylinder functions in Appendix \ref{pcf}, we introduce $z$ for $2iz^{2}=2itg_{\mu}$. For the reason that $g_{\mu}(\lambda)$ has a double zero at $\lambda=\mu$, we let
\begin{equation}
z=it^{1/2}(\lambda-\mu)\psi_{\mu}(\lambda),\label{psimu}
\end{equation}
where $\psi_{\mu}(\lambda)$ is analytic function for $\lambda\in D_{\varepsilon}(\mu)$.

Define $m^{(\mu2)}(x,t,z(\lambda))$ by
\begin{equation}
m^{(\mu2)}(x,t,z(\lambda))=m^{(\mu1)}(x,t,\lambda)\delta_{0}(t)^{\sigma_{3}},\quad \lambda\in D_{\varepsilon}(\mu)\backslash\Sigma^{(5)}.
\end{equation}
Function $m^{(\mu2)}(x,t,z(\lambda))$ satisfies the RH problem $m^{(\mu2)}_{+}(x,t,\lambda)=m^{(\mu2)}_{-}(x,t,\lambda)v^{(\mu2)}(x,t,\lambda)$, where
\begin{equation}
v^{(\mu2)}=\begin{cases}
\left(
\begin{array}{cc}
1&0\\
\hat{r}\delta_{1}^{-2}\rho^{-2}(z)e^{2iz^{2}}&1
\end{array}
\right),\quad\quad\quad\quad~~~ \arg{z}=\frac{\pi}{4},\\
\left(
\begin{array}{cc}
1&\hat{\overline{r}}\delta_{1}^{2}\rho^{2}(z)e^{-2iz^{2}},\\
0&1,
\end{array}
\right),\quad\quad\quad\quad\quad~ \arg{z}=\frac{3\pi}{4},\\
\left(
\begin{array}{cc}
1&0\\
-\frac{\hat{r}(1+\hat{r}\hat{\overline{r}})}{(1+|q|^{2})^{2}}\delta_{1}^{-2}\rho^{-2}(z)e^{2iz^{2}}&1
\end{array}
\right),\quad~ \arg{z}=\frac{5\pi}{4},\\
\left(
\begin{array}{cc}
1&-\frac{\hat{\overline{r}}(1+\hat{r}\hat{\overline{r}})}{(1+|q|^{2})^{2}}\delta_{1}^{2}\rho^{2}(z)e^{-2iz^{2}}\\
0&1
\end{array}
\right),\quad~~~ \arg{z}=\frac{7\pi}{4}.
\end{cases}\nonumber
\end{equation}

We define a matrix function $m^{\mu}(x,t,\lambda)$ for $\hat{m}^{(5)}$ near $\mu$ by
\begin{equation}
m^{\mu}(x,t,\lambda)=Y_{\mu}(x,t,\lambda)m^{pc}(q,z(\lambda))\delta_{0}(t)^{-\sigma_{3}}\tilde{\delta}^{\sigma_{3}}(\lambda)B^{-1}(\lambda)e^{ih(\lambda)\sigma_{3}},
\end{equation}
where $m^{pc}(q,z(\lambda))$ is the solution of the RH problem (\ref{pc}) in Appendix \ref{pcf}. Function $Y_{\mu}(x,t,\lambda)$ is analytic for $\lambda\in D_{\varepsilon}(\mu)$ and defined by
\begin{equation}
Y_{\mu}(x,t,\lambda)=m^{mod}(x,t,\lambda)e^{-ih(\lambda)\sigma_{3}}B(\lambda)\tilde{\delta}^{-\sigma_{3}}\delta_{0}^{\sigma_{3}}(t).\label{ymu}
\end{equation}
Function $m^{\mu}(x,t,\lambda)$ admits the following Lemma:
\begin{lemma}\label{mmu}
Function $m^{\mu}$ satisfies  the jump condition
\begin{equation}
m^{\mu}_{+}(x,t,\lambda)=m^{\mu}_{-}(x,t,\lambda)v^{\mu}(x,t,\lambda),
\end{equation}
where the jump matrix
\begin{equation}
v^{\mu}=\hat{J}^{(5)}, \quad \lambda\in\gamma_{(\overline{\beta},\beta)}\cap D_{\varepsilon}(\mu).
\end{equation}
For $t\rightarrow\infty$,
\begin{equation}
\begin{split}
&\|\hat{v}^{(5)}-v^{\mu}\|_{L^{1}(\chi)}=\mathcal{O}(t^{-1}\ln(t)),\\
&\|\hat{v}^{(5)}-v^{\mu}\|_{L^{2}(\chi)}=\mathcal{O}(t^{-3/4}\ln(t)),\\
&\|\hat{v}^{(5)}-v^{\mu}\|_{L^{\infty}(\chi)}=\mathcal{O}(t^{-1/2}\ln(t)),\\
&\|m^{mod}(m^{\mu})^{-1}-I\|_{L^{\infty}(\partial D_{\varepsilon}(\mu))}=\mathcal{O}(t^{-1/2}),
\end{split}\nonumber
\end{equation}
\begin{equation}
\frac{1}{2i\pi}\int_{\partial D_{\varepsilon}(\mu)}(m^{mod}(m^{\mu})^{-1}-I)\mathrm{d}{\lambda}=\frac{Y_{\mu}(x,t,\mu)m_{1}^{pc}Y_{\mu}^{-1}(x,t,\mu)}{\sqrt{t}\psi_{\mu}(\mu)}+\mathcal{O}(t^{-1}),
\end{equation}
where
\begin{equation}
m^{pc}_{1}=\left(\begin{array}{cc}
0&-e^{-\pi\nu_{2}}\beta^{pc}(q)\\
e^{-\pi\nu_{2}}\overline{\beta^{pc}(q)}&0
\end{array}\right).\label{m1pc}
\end{equation}
\end{lemma}

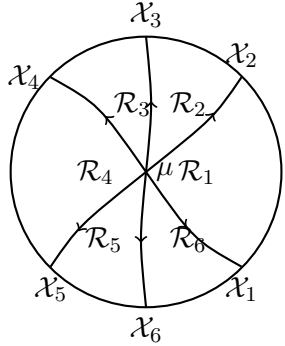
\begin{figure}[H]
\begin{center}
\begin{tikzpicture}[scale=0.45]
\draw[thick] (0,0) circle (4cm);
\draw[thick] (0,-4) .. controls (-0.2,-2)  .. (0,0);
\draw[thick] (0,4) .. controls (0.2,2)  .. (0,0);
\draw[thick] (2.82,-2.8) .. controls (1.4,-2)  .. (0,0);
\draw[thick] (-2.82,2.8) .. controls (-1.4,2)  .. (0,0);
\draw[thick] (-2.82,-2.8) .. controls (-2,-1.6)  .. (0,0);
\draw[thick] (2.82,2.8) .. controls (2,1.6)  .. (0,0);
\draw[thick,->] (0.15,2)--(0.15,2.1);
\draw[thick,->] (-0.15,-2)--(-0.15,-2.1);
\draw[thick,->] (1.15,-1.6)--(1.2,-1.65);
\draw[thick,->] (-1.15,1.6)--(-1.2,1.65);
\draw[thick,->] (1.98,1.72)--(2,1.75);
\draw[thick,->] (-1.98,-1.72)--(-2,-1.75);
\draw[ ](1.5,0)node[] {$\mathcal{R}_{1}$} (0.4,2)node[right] {$\mathcal{R}_{2}$} (-0.4,2)node[] {$\mathcal{R}_{3}$};
\draw[ ](-1.5,0)node[] {$\mathcal{R}_{4}$} (-0.4,-2)node[left] {$\mathcal{R}_{5}$}(0.4,-2)node[right] {$\mathcal{R}_{6}$};
\draw[ ](0,-4)node[below] {$\mathcal{X}_{6}$} (0,4)node[above] {$\mathcal{X}_{3}$}(2.82,-2.8)node[below] {$\mathcal{X}_{1}$};
\draw[ ](-2.82,2.8)node[left] {$\mathcal{X}_{4}$} (-2.82,-2.8)node[below] {$\mathcal{X}_{5}$}(2.82,2.8)node[above] {$\mathcal{X}_{2}$};
\draw[ ](0,0)node[right] {$\mu$};
\end{tikzpicture}
\end{center}
\caption{The jump contour $\mathcal{R}_{j}$ and $\mathcal{X}_{j}$, $j=1,2,3,4,5,6$.}\label{mu}
\end{figure}

\subsection{the asymptotic solution }
Function $m^{err}$ defined  by
\begin{equation}
m^{err}=\hat{m}^{(5)}(m^{app})^{-1},
\end{equation}
satisfies RH problem
\begin{equation}
m^{err}_{+}(x,t,\lambda)=m^{err}_{-}(x,t,\lambda)v^{err}(x,t,\lambda),\quad k\in\Sigma^{err},\label{merr}
\end{equation}
where $\Sigma^{err}=(\Sigma^{(5)}\backslash(\Sigma^{mod}\cup\overline{\mathcal{D}}))\cup\partial D\cup\mathcal{X}$, $\mathcal{X}=\mathcal{X}_{1}\cup\mathcal{X}_{2}\cup\mathcal{X}_{4}\cup\mathcal{X}_{5}$ and
\begin{equation}
v^{err}=\begin{cases}
m^{mod}\hat{J}^{(5)}(m^{mod})^{-1},\quad \lambda\in\Sigma^{err}\backslash\overline{\mathcal{D}},\\
m^{mod}(m^{\alpha})^{-1},\quad\quad\quad~~  \lambda\in\partial D_{\varepsilon}(D_{\alpha}),\\
m^{mod}(m^{\beta})^{-1},\quad\quad\quad~~ \lambda\in\partial D_{\varepsilon}(D_{\beta}),\\
m^{mod}(m^{\overline{\alpha}})^{-1},\quad\quad\quad~~ \lambda\in\partial D_{\varepsilon}(D_{\overline{\alpha}}),\\
m^{mod}(m^{\overline{\beta}})^{-1},\quad\quad\quad~~ \lambda\in\partial D_{\varepsilon}(D_{\overline{\beta}}),\\
m^{mod}(m^{\mu})^{-1},\quad\quad\quad~~ \lambda\in\partial D_{\varepsilon}(D_{\mu}),\\
m^{\mu}_{-}\hat{J}^{(5)}(m^{\mu}_{+})^{-1},\quad\quad\quad \lambda\in\mathcal{X}.\\
\end{cases}\nonumber
\end{equation}
Let $\hat{\omega}=v^{err}-I$, and $\hat{\omega}$ admits the following Lemma:
\begin{lemma}\cite{lj2017}
For $t\rightarrow\infty$, $\hat{\omega}$ satisfies
\begin{equation}
\parallel\hat{\omega}\parallel_{(L^{1}\cap L^{2}\cap L^{\infty})(\Sigma^{err}\setminus\overline{\mathcal{D}})}=\mathcal{O}(e^{-ct}),\quad c>0.
\end{equation}
\end{lemma}
Equations (\ref{malpha}) and (\ref{mbeta}) imply that
\begin{equation}
\parallel\hat{\omega}\parallel_{L^{\infty}(\partial D_{\varepsilon}(\alpha)\cup\partial D_{\varepsilon}(\beta)\cup\partial D_{\varepsilon}(\overline{\alpha})\cup\partial D_{\varepsilon}(\overline{\beta}))}=\mathcal{O}(t^{-N}),\quad t\rightarrow\infty,\quad N\geq1.
\end{equation}
For $t\rightarrow\infty$, Lemma \ref{mmu} yields that
\begin{equation}
\begin{split}
&\parallel\hat{\omega}\parallel_{L^{1}(\mathcal{X})}=\mathcal{O}(t^{-1}\ln{t}),\quad\parallel\hat{\omega}\parallel_{L^{2}(\mathcal{X})}=\mathcal{O}(t^{-3/4}\ln{t}),\\
&\parallel\hat{\omega}\parallel_{L^{\infty}(\mathcal{X})}=\mathcal{O}(t^{-1/2}\ln{t}),\quad\parallel\hat{\omega}\parallel_{L^{\infty}(\partial D_{\varepsilon}(\mu))}=\mathcal{O}(t^{-1/2}),
\end{split}\label{hatomega}
\end{equation}
Thus for $t\rightarrow\infty$, we have
\begin{equation}
\parallel\hat{\omega}\parallel_{(L^{1}\cap L^{2})(\Sigma^{err})}=\mathcal{O}(t^{-1/2}),\quad \parallel\hat{\omega}\parallel_{L^{\infty}(\Sigma^{err})}=\mathcal{O}(t^{-1/2}\ln{t}).\label{omegal2}
\end{equation}
Let $\hat{\mathcal{C}}$ denotes the Cauchy operator associated with $\Sigma^{err}$, for function $f(\lambda)$
\begin{equation}
(\hat{\mathcal{C}}f)(\lambda)=\frac{1}{2i\pi}\int_{\Sigma^{err}}\frac{f(s)}{s-\lambda}\mathrm{d}{s},\quad \lambda\in\mathbb{C}\backslash\Sigma^{err}.\nonumber
\end{equation}
Define $\hat{\mathcal{C}}_{\hat{\omega}}$: $L^{2}(\Sigma^{err})\rightarrow L^{2}(\Sigma^{err})$ by $\hat{\mathcal{C}}_{\hat{\omega}}f=\hat{\mathcal{C}}_{-}(f\hat{\omega})$ \cite{lj2017,lj2018}. We have the Lemma:
\begin{lemma}\label{cb}
 In the Banach space $\mathcal{B}(L^{2}(\Sigma^{err}))$, we have
\begin{equation}
\parallel\hat{\mathcal{C}}_{\hat{\omega}}\parallel_{\mathcal{B}(L^{2}(\Sigma^{err}))}\leq C\parallel\hat{\omega}\parallel_{L^{\infty}(\Sigma^{err})}=\mathcal{O}(t^{-1/2}\ln{t}),\quad t\rightarrow\infty,
\end{equation}
and $I-\hat{\mathcal{C}}_{\hat{\omega}}\in\mathcal{B}(L^{2}(\Sigma^{err}))$ is invertible for large enough $t$.
\end{lemma}
Define a $2\times2$-matrix function $\hat{\mu}(x,t,\lambda)$ for a large $t$ by
\begin{equation}
\hat{\mu}=I+(I-\hat{\mathcal{C}}_{\hat{\omega}})^{-1}\hat{\mathcal{C}}_{\hat{\omega}}I\in I+L^{2}(\Sigma^{err}),\label{hatmu}
\end{equation}
we consider the Neumann series representation of $(I-\hat{\mathcal{C}}_{\hat{\omega}})^{-1}$ as
$(I-\hat{\mathcal{C}}_{\hat{\omega}})^{-1}=\sum\limits_{j=0}^{\infty}\hat{\mathcal{C}}_{\hat{\omega}}^{j}$.
we have
\begin{equation}
\parallel(I-\hat{\mathcal{C}}_{\hat{\omega}})^{-1}\parallel_{\mathcal{B}(L^{2}(\Sigma^{err}))}\leq\sum_{j=0}^{\infty}\parallel\hat{\mathcal{C}}_{\hat{\omega}}\parallel_{\mathcal{B}(L^{2}(\Sigma^{err}))}^{j}
=\left(1-\parallel\hat{\mathcal{C}}_{\hat{\omega}}\parallel_{\mathcal{B}(L^{2}(\Sigma^{err}))}\right)^{-1},
\end{equation}
as $\parallel\hat{\mathcal{C}}_{\hat{\omega}}\parallel_{\mathcal{B}(L^{2}(\Sigma^{err}))}<1$. Using Lemma \ref{cb}, we have
\begin{equation}
\begin{split}
\parallel\hat{\mu}-I\parallel_{L^{2}(\Sigma^{err})}&=\parallel (I-\hat{\mathcal{C}}_{\hat{\omega}})^{-1}\hat{\mathcal{C}}_{\hat{\omega}}I\parallel_{L^{2}(\Sigma^{err})}\\
&\leq\parallel (I-\hat{\mathcal{C}}_{\hat{\omega}})^{-1}\parallel_{\mathcal{B}(L^{2}(\Sigma^{err}))}\parallel\hat{\mathcal{C}}_{\hat{\omega}}I\parallel_{L^{2}(\Sigma^{err})}\\
&\leq\frac{C\parallel\hat{\omega}\parallel_{(L^{2}(\Sigma^{err}))}}{1-1-\parallel\hat{\mathcal{C}}_{\hat{\omega}}\parallel_{\mathcal{B}(L^{2}(\Sigma^{err}))}}\\
&\leq C\parallel\hat{\omega}\parallel_{(L^{2}(\Sigma^{err}))}.
\end{split}
\end{equation}
Using equation (\ref{omegal2}), we have the estimate shown as follows:
\begin{lemma}\label{hatmu}
 Function $\hat{\mu}$ satisfies
\begin{equation}
\parallel\hat{\mu}-I\parallel_{L^{2}(\Sigma^{err})}=\mathcal{O}(t^{-1/2}),\quad t\rightarrow\infty.
\end{equation}
\end{lemma}
The defination (\ref{hatmu}) implies that $\hat{\mu}-I=\hat{\mathcal{C}}_{\hat{\omega}}\hat{\mu}$, thus $m^{err}=I+\hat{\mathcal{C}}(\hat{\mu}\hat{\omega})$ satisfies a $L^{2}$-RH problem. And this $L^{2}$-RH problem with a unique solution shown as follows:
\begin{lemma}\label{uniquesol}
 There exist the unique solution of the RH problem (\ref{merr}) given by
\begin{equation}
m^{err}(x,t,\lambda)=I+\frac{1}{2i\pi}\int_{\Sigma^{err}}\frac{\hat{\mu}(s)\hat{\omega}(s)}{s-\lambda}\mathrm{d}{s}.
\end{equation}
\end{lemma}
From the Lemma \ref{uniquesol}, we have
\begin{equation}
\begin{split}
\lim_{\lambda\rightarrow\infty}\lambda(m^{err}(x,t,\lambda)-I)&=\lim_{\lambda\rightarrow\infty}\frac{1}{2i\pi}\int_{\Sigma^{err}}\frac{\lambda\hat{\mu}(s)\hat{\omega}(s)}{s-\lambda}\mathrm{d}{s}\\
&=-\frac{1}{2i\pi}\int_{\Sigma^{err}}\hat{\mu}(s)\hat{\omega}(s)\mathrm{d}s.\label{hatmuomega}
\end{split}
\end{equation}
This implies that
\begin{equation}
\begin{split}
&\lim_{\lambda\rightarrow\infty}\lambda(m^{err}(x,t,\lambda)-I)\\
&=-\frac{1}{2i\pi}\int_{\partial D_{\varepsilon}(\mu)}\hat{\omega}(x,t,\lambda)\mathrm{d}\lambda-\frac{1}{2i\pi}\int_{\partial D_{\varepsilon}(\mu)}(\hat{\mu}(x,t,\lambda)-I)\hat{\omega}(x,t,\lambda)\mathrm{d}\lambda\\
&=-\frac{1}{2i\pi}\int_{\partial D_{\varepsilon}(\mu)}(m^{mod}(m^{\mu})^{-1}-I)\mathrm{d}\lambda+\mathcal{O}(\parallel\hat{\mu}-I\parallel_{L^{2}(\partial D_{\varepsilon}(\mu))}\parallel\hat{\omega}\parallel_{L^{2}(\partial D_{\varepsilon}(\mu))})\\
&=-\frac{Y_{\mu}(x,t,\mu)m_{1}^{pc}Y_{\mu}^{-1}(x,t,\mu)}{\sqrt{t}\psi_{\mu}(\mu)}+\mathcal{O}(t^{-1}),\quad t\rightarrow\infty.
\end{split}\nonumber
\end{equation}
Equation (\ref{hatomega}) and Lemma \ref{hatmu} imply that the contribution of $\mathcal{X}$ to the right hand side of equation (\ref{hatmuomega}) is
\begin{equation}
\mathcal{O}(\parallel\hat{\omega}\parallel_{L^{1}(\mathcal{X})})+\mathcal{O}(\parallel\hat{\mu}-I\parallel_{L^{2}(\mathcal{X})}\parallel\hat{\omega}\parallel_{L^{2}(\mathcal{X})})=\mathcal{O}(t^{-1}\ln{t})
\quad t\rightarrow\infty.\end{equation}
Thus, we have the following limit
\begin{equation}
\lim_{\lambda\rightarrow\infty}\lambda(m^{err}(x,t,\lambda)-I)=-\frac{Y_{\mu}(x,t,\mu)m_{1}^{pc}Y_{\mu}^{-1}(x,t,\mu)}{\sqrt{t}\psi_{\mu}(\mu)}+\mathcal{O}(t^{-1}\ln{t}),\quad t\rightarrow\infty.\label{merrI}
\end{equation}
Collecting the five transformations in Section \ref{deformation}, we have
\begin{equation}
\begin{split}
&\lim_{\lambda\rightarrow\infty}\lambda(\hat{m}(x,t,\lambda)-I)\\
&=e^{itg^{(0)}\sigma_{3}}e^{ih(\infty)\sigma_{3}}\lim_{\lambda\rightarrow\infty}
\lambda(m^{mod}-I+(m^{err}-I)m^{mod})e^{-ih(\infty)\sigma_{3}}e^{-itg^{(0)}\sigma_{3}}\\
&=e^{i(tg^{(0)}+h(\infty))\sigma_{3}}\left(\lim_{\lambda\rightarrow\infty}\lambda(m^{mod}-I)+\lim_{\lambda\rightarrow\infty}\lambda(m^{err}-I)\right)e^{-i(tg^{(0)}+h(\infty))\sigma_{3}}.
\end{split}\nonumber
\end{equation}
Thus, we have the solution of the DNLS equation (\ref{dnls3}) is given by
\begin{equation}
\begin{split}
q(x,t)&=2i\lim_{\lambda\rightarrow\infty}(\lambda\hat{m}(x,t,\lambda))_{12}\\
&=2ie^{2i(tg^{(0)}+h(\infty))}\left(\lim_{\lambda\rightarrow\infty}\lambda m^{mod}_{12}(x,t,\lambda)+\lim_{\lambda\rightarrow\infty}\lambda m^{err}_{12}(x,t,\lambda)\right).
\end{split}\label{q1}
\end{equation}
Insert Lemma \ref{mmod} and equation (\ref{merrI}) into equation (\ref{q1}), Lemma \ref{solution} can be derived.

\appendix

\begin{appendices}
\section{Airy Function}\label{airy}
For $\zeta\in\mathbb{C}\backslash Y$, we define the function $m^{Ai}(\zeta)$ as
\begin{equation}
m^{Ai}(\zeta)=\Psi(\zeta)\cdot
\begin{cases}
\begin{split}
&e^{\frac{2}{3}\zeta^{\frac{3}{2}}\sigma_{3}},\quad\quad\quad\quad\quad\quad\quad \zeta\in S_{1}\cup S_{4},\\
&\left(
\begin{array}{cc}
1&0\\
-1&1
\end{array}
\right)e^{\frac{2}{3}\zeta^{\frac{3}{2}}\sigma_{3}},\quad~ \zeta\in S_{2},\\
&\left(
\begin{array}{cc}
1&0\\
1&1
\end{array}
\right)e^{\frac{2}{3}\zeta^{\frac{3}{2}}\sigma_{3}},\quad\quad~ \zeta\in S_{3},
\end{split}
\end{cases}
\end{equation}
where
\begin{equation}
\Psi(\zeta)=\begin{cases}
\left(
\begin{array}{cc}
Ai(\zeta)&Ai(\omega^{2}\zeta)\\
Ai^{'}(\zeta)&\omega^{2}Ai^{'}(\omega^{2}\zeta)
\end{array}
\right)e^{-\frac{i\pi}{6}\sigma_{3}},\quad \zeta\in\mathbb{C}^{+},\\
\left(
\begin{array}{cc}
Ai(\zeta)&-\omega^{2}Ai(\omega\zeta)\\
Ai^{'}(\zeta)&-Ai^{'}(\omega\zeta)
\end{array}
\right)e^{-\frac{i\pi}{6}\sigma_{3}},\quad \zeta\in\mathbb{C}^{-},
\end{cases}
\end{equation}
and $\mathcal{A}=\cup \mathcal{A}_{j}\subset\mathbb{C}$, $j=1,2,3,4$, and
\begin{equation}
\begin{split}
&\mathcal{A}_{1}=\{y|0\leq y\leq\infty\},\quad \mathcal{A}_{2}=\{ye^{\frac{2i\pi}{3}}|0\leq y\leq\infty\},\\
&\mathcal{A}_{3}=\{-y|0\leq y\leq\infty\},\quad \mathcal{A}_{4}=\{ye^{\frac{-2i\pi}{3}}|0\leq y\leq\infty\},
\end{split}\nonumber
\end{equation}
see the Figure \ref{a1}.

\begin{figure}[H]
\begin{center}
\begin{tikzpicture}[scale=0.45]
\draw[thick,->] (-4.2,0)--(-2,0);
\draw[thick,-] (-2,0)--(2,0);
\draw[thick,->] (4.2,0)--(2,0);
\draw[thick,->] (-2.2,4) -- (-1,2);
\draw[thick,-] (-1,2) -- (0,0);
\draw[thick,->] (-2.2,-4) -- (-1,-2);
\draw[thick,-] (-1,-2) -- (0,0);
\draw [] (1.8,2) circle [radius=0] node[right] {$S_{1}$};
\draw [] (-2.5,2) circle [radius=0] node[right] {$S_{2}$};
\draw [] (-2.5,-2) circle [radius=0] node[right] {$S_{3}$};
\draw [] (1.8,-2) circle [radius=0] node[right] {$S_{4}$};
\draw [] (4.4,0) circle [radius=0] node[right] {$\mathcal{A}_{1}$};
\draw [] (-5.5,0) circle [radius=0] node[right] {$\mathcal{A}_{3}$};
\draw [] (-2.5,4.4) circle [radius=0] node[right] {$\mathcal{A}_{2}$};
\draw [] (-2.5,-4.4) circle [radius=0] node[right] {$\mathcal{A}_{4}$};
\end{tikzpicture}
\end{center}\label{a1}
\caption{The contour of the Airy RH problem.}
\end{figure}
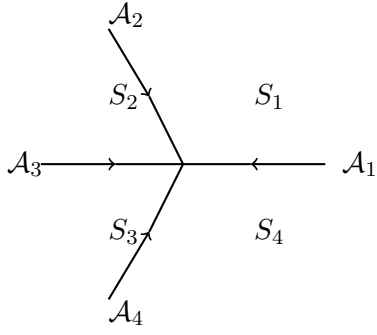

Define the asymptotic approximation $m_{as,N}^{Ai}(\zeta)$ as
\begin{equation}
m_{as,N}^{Ai}(\zeta)=\frac{e^{\frac{i\pi}{12}}}{2\sqrt{\pi}}\zeta^{-\frac{\sigma_{3}}{4}}\sum_{k=0}^{N}{\left(\frac{2}{3}\zeta^{\frac{3}{2}}\right)^{-k}}
\left(\begin{array}{cc}
(-1)^{k}u_{k}&u_{k}\\
-(-1)^{k}v_{k}&v_{k}
\end{array}
\right)e^{-\frac{i\pi}{4}\sigma_{3}},\quad \zeta\in\mathbb{C}\backslash \mathcal{A},
\end{equation}
where $N$ is the positive integer number, $u_{k}$ and $v_{k}$ are the real constants
\begin{equation}
u_{0}=v_{0}=1,\quad u_{k}=\frac{(2k+1)(2k+3)\cdots(6k-1)}{(216)^{k}k!},\quad v_{k}=-\frac{6k+1}{6k-1}u_{k},\quad k=1,2,3
\end{equation}

\begin{theorem}\label{Airy}
The function $m^{Ai}$ satisfies the following properties:\\
$\blacktriangleright$ The function $m^{Ai}$ analytic for $\zeta\in\mathbb{C}\backslash \mathcal{A}$ and admits the RH problem
\begin{equation}
m^{Ai}_{+}(\zeta)=m^{Ai}_{-}(\zeta)v^{Ai}(\zeta),\quad \zeta\in \mathcal{A}\backslash\{0\},
\end{equation}
where
\begin{equation}
v^{Ai}(\zeta)=\begin{cases}
\begin{split}
&\left(
\begin{array}{cc}
1&-e^{-\frac{4}{3}\zeta^{\frac{3}{2}}}\\
0&1
\end{array}
\right),\quad~~ \zeta\in \mathcal{A}_{1},\\
&\left(
\begin{array}{cc}
1&0\\
e^{\frac{4}{3}\zeta^{\frac{3}{2}}}&1
\end{array}
\right),\quad\quad\quad \zeta\in \mathcal{A}_{2}\cup \mathcal{A}_{4},\\
&\left(
\begin{array}{cc}
0&1\\
-1&0
\end{array}
\right),\quad\quad\quad\quad \zeta\in \mathcal{A}_{3},
\end{split}\nonumber
\end{cases}
\end{equation}
$\blacktriangleright$  The function $m^{Ai}_{as,N}(\zeta)$ is analytic for $\zeta\in\mathbb{C}\backslash (-\infty,0]$ and admits the jump condition
\begin{equation}
m^{Ai}_{as,N+}(\zeta)=m^{Ai}_{as,N-}(\zeta)
\left(
\begin{array}{cc}
0&1\\
-1&0
\end{array}
\right),\quad \zeta<0\label{Ai}.
\end{equation}
$\blacktriangleright$ The function $m^{Ai}_{as,N}(\zeta)$ approximates $m^{Airy}$ as $\zeta\rightarrow\infty$
\begin{equation}
(m^{Ai}_{as,N}(\zeta))^{-1}m^{Ai}(\zeta)=I+\mathcal{O}(\zeta^{-\frac{3(N+1)}{2}}),\quad \zeta\rightarrow\infty.
\end{equation}
where the error term is uniform  with respect to $\arg{\zeta}\in[1,2\pi]$.
\end{theorem}

\section{Parabolic Cylinder Function}\label{pcf}
Define the RH problem
\begin{equation}
m^{pc}_{+}(q,z)=m^{pc}_{-}(q,z)v^{pc}(q,z),\quad z\in \mathcal{P},\label{pc}
\end{equation}
where
\begin{equation}
v^{pc}(q,z)=\begin{cases}
\begin{split}
&\left(
\begin{array}{cc}
1&0\\
q\rho(q,z)^{-2}e^{2iz^{2}}&1
\end{array}
\right),\quad\quad\quad~~ z\in \mathcal{P}_{1},\\
&\left(
\begin{array}{cc}
1&\overline{q}\rho(q,z)^{2}e^{-2iz^{2}}\\
0&1
\end{array}
\right),\quad\quad\quad~~ z\in \mathcal{P}_{2},\\
&\left(
\begin{array}{cc}
1&0\\
-\frac{q}{1+|q|^{2}}\rho(q,z)^{-2}e^{2iz^{2}}&1
\end{array}
\right),\quad z\in \mathcal{P}_{3},\\
&\left(
\begin{array}{cc}
1&-\frac{\overline{q}}{1+|q|^{2}}\rho(q,z)^{2}e^{-2iz^{2}}\\
0&1
\end{array}
\right),\quad z\in \mathcal{P}_{4},
\end{split}
\end{cases}\nonumber
\end{equation}
and $\mathcal{P}=\cup \mathcal{P}_{j}\subset\mathbb{C}$, $j=1,2,3,4$,
\begin{equation}
\begin{split}
&\mathcal{P}_{1}=\{se^{\frac{i\pi}{4}}|0\leq s\leq\infty\},\quad \mathcal{P}_{2}=\{se^{\frac{3i\pi}{4}}|0\leq s\leq\infty\},\\
&\mathcal{P}_{3}=\{se^{\frac{-3i\pi}{4}}|0\leq s\leq\infty\},\quad \mathcal{P}_{4}=\{se^{\frac{-3i\pi}{4}}|0\leq s\leq\infty\},
\end{split}\nonumber
\end{equation}
see the Figure \ref{b}.

\begin{figure}[H]
\begin{center}
\begin{tikzpicture}[scale=0.45]
\draw[thick,->] (4,4) -- (2,2);
\draw[thick,-] (2,2) -- (-2,-2);
\draw[thick,->]  (-4,-4)--(-2,-2);
\draw[thick,->] (-4,4) -- (-2,2);
\draw[thick,-] (-2,2) -- (2,-2);
\draw[thick,->]  (4,-4)--(2,-2);
\draw [] (1.8,2.5) circle [radius=0] node[left] {$\mathcal{P}_{1}$};
\draw [] (-1.8,2.5) circle [radius=0] node[right] {$\mathcal{P}_{2}$};
\draw [] (-1.8,-2.5) circle [radius=0] node[right] {$\mathcal{P}_{3}$};
\draw [] (1.8,-2.5) circle [radius=0] node[left] {$\mathcal{P}_{4}$};
\end{tikzpicture}
\end{center}\label{b}
\caption{The contour of the Parabolic Cylinder RH problem.}
\end{figure}
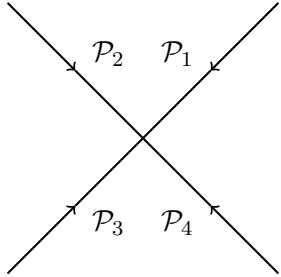

\begin{theorem}
The RH problem with a unique solution $m^{pc}(q,z)$
\begin{equation}
m^{pc}(q,z)=I+\frac{i}{z}\left(
\begin{array}{cc}
0&-e^{-\pi\nu}\beta^{\mathcal{P}}(q)\\
e^{\pi\nu}\overline{\beta^{\mathcal{P}}(q)}&0
\end{array}
\right)+\mathcal{O}(\frac{1}{z^{2}}),\quad z\rightarrow\infty,\quad q\in\mathbb{C}
\end{equation}
where the error term is uniform with respect to $\arg{z}\in[0,2\pi]$ and $q$ is compact subsets of $\mathbb{C}$, and $\beta^{\mathcal{P}}(q)$ is given by
\begin{equation}
\beta^{\mathcal{P}}(q)=\frac{\sqrt{\nu(q)}}{2}e^{i(-\frac{3\pi}{4}-2\nu(q)\ln{2}-\arg{q}+\arg{\Gamma(i\nu(q))})},\quad q\in\mathbb{C}.
\end{equation}
And for compact subset $K\subset\mathbb{C}$,
\begin{equation}
\sup_{q\in K}\sup_{z\in\mathbb{C}\backslash \mathcal{P}}|m^{pc}(q,z)|<\infty.
\end{equation}
\end{theorem}

\end{appendices}

\noindent\textbf{Acknowledgements}
This work is supported by the National Natural Science Foundation of China(12175069 and 12235007)
 and Science and Technology Commission of Shanghai Municipality (21JC1402500 and 22DZ2229014).

  \noindent\textbf{Data Availability Statements}

    The data that supports the findings of this study are available within the article.\vspace{2mm}

        \noindent{\bf Conflict of Interest}

    The authors have no conflicts to disclose.


\begin{thebibliography}{99}
\bibitem{mw1976}
W.Mio,  T.Ogino, K.Minami, S.Takeda, Modified nonlinear Schr$\ddot{o}$dinger equation for Alfv$\acute{e}$n waves propagating along the magnetic field in cold plasmas, J. Phys. Soc. Jpn. 41(1976)265-271.


\bibitem{me1976}
E.Mjolhus, On the modulational instability of hydromagnetic waves parallel to the magnetic field, J.
Plasma Phys. 16(1976)321-334.


\bibitem{ky1985}
Y.Kodama, Optical solitons in a monomode fiber, J. Stat. Phys. 39(1985)597-614.

\bibitem{wm1983}
M.Wadati, K.Sogo, Gauge transformation in soliton theory, J. Phys. Soc. Jpn. 52(1983)394-398.

\bibitem{zzc2021}
Z.C.Zhang, E.G.Fan, Inverse scattering transform and multiple high-order pole solutions for the
Gerdjikov-Ivanov equation under the zero/nonzero background, Z. Angew. Math. Phys. 72(2021)153.


\bibitem{emb2018}
M.B.Erdo$\check{g}$an, T.B.G$\ddot{u}$rel, N.Tzirakis, The derivative nonlinear Schr$\ddot{o}$dinger equation on the half
line, Ann. Inst. Henri Poincar$\acute{e}$, Anal. Non Lin$\acute{e}$aire, 35(2018)1947-1973.

\bibitem{fs2020}
S.Fromm, Admissible boundary values for the Gerdjikov-Ivanov equation with asymptotically time-periodic boundary data, Symmetry Integr. Geom. 16(2020)079.


\bibitem{f2001}
 E.G.Fan, A family of completely integrable multi-Hamiltonian systems explicitly related to some celebrated equations,
J. Math. Phys. 42(2001)4327-4344.


\bibitem{yh2015}
H.Yilmaz, Exact solutions of the Gerdjikov-Ivanov equation using Darboux transformations, J. Nonlin. Math. Phys. 22(2015)32-46.

\bibitem{zss2022}
S.S.Zhang, T.Xu, M.Li, X.F.Zhang, High-order algebraic solitons of Gerdjikov-Ivanov equation:
Asymptotic analysis and emergence of rogue waves, Physica D, 432(2022)133128.

\bibitem{xsw2012}
S.W.Xu, J.S.He, The rogue wave and breather solution of the Gerdjikov-Ivanov equation, J. Math.
Phys. 53(2012)063507.

\bibitem{hy2013}
Y.Hou, E.G.Fan, P.Zhao, Algebro-geometric solutions for the Gerdjikov-Ivanov hierarchy, J. Math.
Phys. 54(2013)073505.

\bibitem{zp2020}
P.Zhao, E.G.Fan, Finite gap integration of the derivative nonlinear Schr$\ddot{o}$dinger equation: a Riemann-Hilbert method. Physica D, 402(2020)132213.

\bibitem{ks2004}
S.Kakei, T.Kikuchi, Affine Lie group approach to a derivative nonlinear Schr$\ddot{o}$dinger equation and
its similarity reduction, Int. Math. Res. Notices, 78(2004)4181-4209.

\bibitem{ks2005}
S.Kakei, T.Kikuchi, Solutions of a derivative nonlinear Schr$\ddot{o}$dinger hierarchy and its similarity
reduction, Glasgow Math. J. 47A(2005)99-107.

\bibitem{cjb2022}
J.B.Chen, Y.P.Zhen, The complex Hamiltonian system in the Gerdjikov-Ivanov equation and its applications, Anal. Math. Phys. 12(2022)100.

\bibitem{ljq2016}
J.Q.Liu, P.A.Perry, C.Sulem, Global existence for the derivative nonlinear Schr$\ddot{o}$dinger equation by the method of inverse scattering, Commun. Part. Diff. Eq. 41(11)(2016)1692-1760.




\bibitem{dz1993}  P.A.Deift, X.Zhou, A steepest descent method for oscillatory Riemann-Hilbert problems: Asymptotics for the mKdV equation, Ann. Math.
137(2)(1993)295-368.



\bibitem{msv1973}  S.V.Manakov, Nonlinear Fraunhofer diffraction, Zh. Eksp. Teor.
Fiz. 65(1973)1392-1398(Russian);Sov.Phys.JETP. 38(1974)693-696.

\bibitem{its1981}  A.R.Its, Asymptotic behavior of the solution to the nonlinear Schr$\ddot{o}$dinger equation,and isomonodromic
deformations of systems of linear differential equations, Doklady Akademii Nauk SSSR 261(1981)14-18(Russian); Soviet. Math. Dokl. 24(1982)452-456(English).

\bibitem{dpa1993}  P.A.Deift,  A.R.Its,  X.Zhou, Long-time asymptotics for integrable nonlinear wave equations, in.
Zhou, Long-time asymptotics for integrable nonlinear wave equations. In: Important developments in
soliton theory, 181-204, Springer Ser. Nonlinear Dynam. Springer, Berlin, (1993).


\bibitem{dp1994}  P.A.Deift,  S.Venakides,  X.Zhou, The collisionless shock region for the long-time behavior of
solutions of the KdV equation, Commun. Pure Appl. Math. 47(1994)199-206.

\bibitem{dp1997} P.A.Deift,  S.Venakides,  X.Zhou, New results in small dispersion KdV by an extension of the steepest
descent method for Riemann-Hilbert problems, Int. Math. Res. Not. 6(1997)286-299.


\bibitem{gav1973}  A.V.Gurevich, L.P.Pitaevskii, Nonstationary structure of a collisionless shock-wave, Zhurnal Eksperimentalnoi I Teoreticheskoi Fiziki, 65(2)(1973)590-604.

\bibitem{ky1976}  E.Y.Khruslov, Asymptotic behavior of the solution of the Cauchy problem for the
Korteweg-de Vries equation with step like initial data, Sb. Math. 99(1976)261-281.


\bibitem{kvp1986}  V.P.Kotlyarov, E.Y.Khruslov, Solitons of the nonlinear Schr$\ddot{o}$dinger equation generated
by the continuous spectrum, Teor. Math. phys. 68(1986)751-761.


\bibitem{vs1986}  S.Venakides, Long time asymptotics of the Korteweg-de Vries equation, Trans. Am. Math. Soc.
293(1986)411-419.

\bibitem{bdm2011} A.Boutet de Monvel,  V.P.Kotlyarov, D. Shepelsky, Focusing NLS equation: long-time dynamics of step-like
initial data, Int. Math. Res. Notices, 7(2011)1613-1653.

\bibitem{rb2007}R.Buckingham, S.Venakides, Long-time asymptotics of the nonlinear Schr$\ddot{o}$dinger equation shock problem, Commun. Pur. Appl. Math.  60(9) (2007)1349-1414.


\bibitem{bm2009}  A.Boutet de Monvel,  A.R.Its,  V.P.Kotlyarov, Long-time asymptotics for the focusing NLS equation
with time-periodic boundary condition on the half-line, Commun. Math. Phys. 290(2)(2009)479-522.


\bibitem{bdm2022}  A.Boutet de Monvel, J.Lenells, D.Shepelsky, The focusing NLS equation with step-like oscillating background: the genus 3 sector, Commun. Math. Phys. 390(3)(2022)1081-1148.


\bibitem{bdm2021} A.Boutet de Monvel, J.Lenells, D.Shepelsky, The focusing NLS equation with step-like oscillating background: scenarios of long-time asymptotics, Commun. Math. Phys. 383(2)(2021)893-952.

    \bibitem{gb2017} G.Biondini, D.Mantzavinos, Long-time asymptotics for the focusing nonlinear Schr$\ddot{o}$dinger equation
with nonzero boundary conditions at infinity and asymptotic stage of modulational instability, Commun. Pur. Appl. Math. 70(12)(2017)2300-2365.

\bibitem{gb2021} G.Biondini, S.T.Li, D.Mantzavinos, Long-time asymptotics for the focusing nonlinear Schr$\ddot{o}$dinger equation with nonzero boundary conditions
in the presence of a discrete spectrum, Commun. Math. Phys. 382(3)(2021)1495-1577.

\bibitem{gb2014} G.Biondini, G.Kova$\check{c}$i$\check{c}$, Inverse scattering transform for the focusing nonlinear
Schr$\ddot{o}$dinger equation with nonzero boundary conditions, J. Math. Phys. 55(2014)031506.


  \bibitem{am2011} A.Minakov, Long-time behavior of the solution to the mKdV equation with step-like initial data, J. Phys. A: Math. Theor. 44 (2011) 085206.

  \bibitem{tg2020} T.Grava, A.Minakov, On the long-time asymptotic behavior of the modified Korteweg-de Vries equation with step-like ininial data, SIAM J. Math. Anal. 52(2020)5892-5993.


\bibitem{xu2013} J.Xu, E.G.Fan, Y.Chen, Long-time asymptotic for the derivative nonlinear Schr$\ddot{o}$dinger equation with step-like initila value, Math. Phys. Anal. Geom. 16(2013)253-288.


    \bibitem{ln2019} B.L.Guo, N.Liu, The Gerdjikov-Ivanov-type derivative nonlinear Schr$\ddot{o}$dinger equation: Long-time dynamics of nonzero boundary conditions, Math. Meth. Appl. Sci. 42(2019)4839-4861.


\bibitem{tsf2018}
S.F.Tian, T.T.Zhang, Long-time Asymptotic Behavior for the Gerdjikov-Ivanov type of derivative Nonlinear Schr$\ddot{o}$dinger equation with time-periodic boundary condition, P. Am. Math. Soc. 146(4)(2018)1713-1729.

\bibitem{ljq2018}
J.Q.Liu, P.A.Perry, C.Sulem, Long-time behavior of solutions to the derivative nonlinear Schr$\ddot{o}$dinger equation for soliton-free initial data, Ann. I. H. Poincar$\acute{e}$ AN, 35(2018)217-265.

\bibitem{dv1994}
 P.Deift, S.Venakides,  X.Zhou, The collisionless shock region for the long-time behavior of solutions
of the KdV equation, Commun. Pur. Appl. Math. 47(2)(1994)199-206.



\bibitem{pa1999}
P.A.Deift, Orthogonal Polynomials and Random Matrices: a Riemann-Hilbert Approach,
Courant Lecture Notes in Mathematics, vol. 3, New York University, Courant Institute of Mathematical Sciences, New York; American Mathematical Society, Providence, RI,
(1999).

\bibitem{lp2002}
P.A.Deift,  X.Zhou, A priori $L^{p}$-estimates for solutions of Riemann-Hilbert problems, Int. Math. Res. Not. 40
(2002)2121-2154.

\bibitem{as2006}
A.S.Fokas, A.R.Its, A.A.Kapaev, V.Y.Novokshenov, Painlev$\acute{e}$ Transcendents:
The Riemann-Hilbert Approach, Mathematical Surveys and Monographs, vol. 128, American
Mathematical Society, Providence, RI, (2006).

\bibitem{zx1989}
X.Zhou, The Riemann-Hilbert problem and inverse scattering, SIAM J. Math. Anal. 20(4)(1989)966-986.

\bibitem{lj2017}
J.Lenells,  The nonlinear steepest descent method for Riemann-Hilbert problems of low regularity,
Indiana Math. J. 66(4)(2017)1287-1332.

\bibitem{lj2018}
J.Lenells,  Matrix Riemann-Hilbert problems with jumps across Carleson contours, Monatsh. Math.
186(1)(2018)111-152.

\bibitem{kav1997}
A.V.Kitaev, A.H.Vartanian, Leading-order temporal asymptotics of the modified nonlinear
Schr$\ddot{o}$dinger equation: solitonless sector, Inverse Probl. 13(1997)1311-1339.

\bibitem{kdj1978}
 D.J.Kaup, A.C.Newell, An exact solution for a derivative nonlinear Schr$\ddot{o}$dinger equation, J. Math.
Phys. 19(1978)789-801.

\bibitem{lka2017}
L.K.Arruda, J.Lenells, Long-time asymptotics for the derivative nonlinear Schr$\ddot{o}$dinger equation on the half-line, Nonlinearity, 30(2017)4141-4172.


\bibitem{hmf1992} H.M.Farkas, I.Kra, Riemann surfaces, 2nd ed., Graduate Texts in Mathematics, vol. 71, Springer-Verlag, New York, (1992).
\bibitem{jdf1973} J.D.Fay, Theta functions on Riemann surfaces, Lecture Notes in Mathematics, vol. 352, Springer-Verlag, Berlin-New York, (1973).



\bibitem{hmf1992}
H.M.Farkas, I.Kra, Riemann surfaces, 2nd ed., Graduate Texts in Mathematics, vol. 71, Springer-Verlag,
New York, (1992).


























\end{thebibliography}
\end{document}